\documentclass[letter,11pt]{article}

\usepackage{setspace}
\usepackage{subcaption}
\usepackage{enumerate}
\usepackage{amsmath}
\usepackage{amsfonts}
\usepackage{graphicx}
\usepackage{amssymb}
\usepackage{multirow}
\usepackage{geometry}
\usepackage{soul}
\usepackage{colortbl}
\usepackage{wrapfig}
\usepackage{algorithm}
\usepackage{algorithmicx}
\usepackage{algpseudocode}
\usepackage{amsthm}
\usepackage{times}
\usepackage{todonotes}
\usepackage{thmtools} 
\usepackage{thm-restate}

\algtext*{EndWhile}
\algtext*{EndIf}
\algtext*{EndFor}

\newtheorem{lemma}{Lemma}
\newtheorem{lemma*}[lemma]{Lemma*}
\newtheorem{theorem}[lemma]{Theorem}
\newtheorem{theorem*}[lemma]{Theorem*}
\newtheorem{corollary}[lemma]{Corollary}
\newtheorem{definition}[lemma]{Definition}
\newtheorem{observation}[lemma]{Observation}

\newtheorem{fact}[lemma]{Fact}

\usepackage{wrapfig}
\usepackage{bm}

\usepackage{booktabs}
\usepackage[pdftex, plainpages = false, pdfpagelabels, 
                 bookmarks=false,
                 bookmarksopen = true,
                 bookmarksnumbered = true,
                 breaklinks = true,
                 linktocpage,
                 pagebackref,
                 colorlinks = true,  
                 linkcolor = blue,
                 urlcolor  = blue,
                 citecolor = red,
                 anchorcolor = green,
                 hyperindex = true,
                 hyperfigures
                 ]{hyperref}

\usepackage{wrapfig}
\usepackage{bm}

\usepackage{graphicx,color}
\usepackage{empheq}
\usepackage{tcolorbox}
\definecolor{blueish}{rgb}{0.122, 0.435, 0.698}
\definecolor{dagstuhlyellow}{rgb}{0.99,0.78,0.07}
\definecolor{lightgray}{rgb}{0.9,0.9,0.9}
\newtcbox{\colbox}{
size=title,
  nobeforeafter,
  colframe=white,
  colback=blue!5!white,
  arc=10pt,
  tcbox raise base}

\usepackage{xspace}
\newcommand{\pname}{\textsc}
\newcommand{\ProblemFormat}[1]{\pname{#1}}
\newcommand{\ProblemIndex}[1]{\index{problem!\ProblemFormat{#1}}}
\newcommand{\ProblemName}[1]{\ProblemFormat{#1}\ProblemIndex{#1}{}\xspace}

\newcommand{\probEB}{\ProblemName{Edge Bipartization}}
\newcommand{\probOCT}{\ProblemName{Odd Cycle Transversal}}
\newcommand{\probEMC}{\ProblemName{Edge Multiway Cut}}
\newcommand{\probVMC}{\ProblemName{Vertex Multiway Cut}}

\newcommand{\probGFVS}{\ProblemName{Group Feedback Vertex Set}}
\newcommand{\probGFES}{\ProblemName{Group Feedback Edge Set}}
\newcommand{\probGFEVS}{\ProblemName{Group Feedback Edge/Vertex Set}}
\newcommand{\probEMultiC}{\ProblemName{Edge Multicut}}

\newcommand{\probVMultiC}{\ProblemName{Vertex Multicut}}

\newcommand{\probEVMultiC}{\ProblemName{Edge/Vertex Multicut}}

\newcommand{\probFVS}{\ProblemName{Feedback Vertex Set}}

\newcommand{\probSteiner}{\ProblemName{Steiner Tree}}

\title{\LARGE Subexponential Parameterized Algorithms for Cut and Cycle Hitting Problems on $H$-Minor-Free Graphs\thanks{Daniel Lokshtanov is supported by United States - Israel Binational Science Foundation (BSF) grant No. 2018302 and National Science Foundation (NSF) award CCF-2008838. Saket Saurabh is supported by the European Research Council (ERC) under the European Union's Horizon 2020 research and innovation programme (grant agreement No. 819416), and Swarnajayanti Fellowship (No. DST/SJF/MSA01/2017-18). Jie Xue is supported in part by the National Science Foundation (NSF) grant CCF-1814172.}}
\author{\normalsize Sayan Bandyapadhyay\thanks{University of Bergen, Norway, \texttt{sayan.bandyapadhyay@gmail.com}.}
\and \normalsize William Lochet\thanks{University of Bergen, Norway, \texttt{william.lochet@gmail.com}.}
\and \normalsize Daniel Lokshtanov\thanks{University of California, Santa Barbara, USA, \texttt{daniello@ucsb.edu}.}
\and \normalsize Saket Saurabh\thanks{Institute of Mathematical Sciences, Chennai, India, \texttt{saket@imsc.res.in}.}
\and \normalsize Jie Xue\thanks{New York University Shanghai, China, \texttt{jiexue@nyu.edu}.}}
\date{}

\begin{document}

\pagenumbering{gobble}

 \maketitle
 \thispagestyle{empty}
 
\begin{abstract}
We design the first subexponential-time (parameterized) algorithms for several cut and cycle-hitting problems on 
$H$-minor free graphs. In particular, we obtain the following results (where $k$ is the solution-size parameter). 
\begin{itemize}
\item $2^{O(\sqrt{k}\log k)} \cdot n^{O(1)}$ time algorithms for \probEB and \probOCT; 
\item a $2^{O(\sqrt{k}\log^4 k)} \cdot n^{O(1)}$ time algorithm for \probEMC and a $2^{O(r \sqrt{k} \log k)} \cdot n^{O(1)}$ time algorithm for \probVMC (with undeletable terminals), where $r$ is the number of terminals to be separated;
\item a $2^{O((r+\sqrt{k})\log^4 (rk))} \cdot n^{O(1)}$ time algorithm for \probEMultiC and a $2^{O((\sqrt{rk}+r) \log (rk))} \cdot n^{O(1)}$ time algorithm for \probVMultiC (with undeletable terminals), where $r$ is the number of terminal pairs to be separated;
\item a $2^{O(\sqrt{k} \log g \log^4 k)} \cdot n^{O(1)}$ time algorithm for \probGFES and a $2^{O(g \sqrt{k}\log(gk))} \cdot n^{O(1)}$ time algorithm for \probGFVS, where $g$ is the size of the group.
\item In addition, our approach also gives $n^{O(\sqrt{k})}$ time algorithms for all above problems with the exception of $n^{O(r+\sqrt{k})}$ time for \probEVMultiC and $(ng)^{O(\sqrt{k})}$ time for \probGFEVS. 
\end{itemize}
All of our FPT algorithms (the first four items above) are randomized, as they use known randomized kernelization algorithms as sub-routines.

We obtain our results by giving a new decomposition theorem on graphs of bounded genus, or more generally, an $h$-almost-embeddable graph for an arbitrary but fixed constant $h$. Our new decomposition theorem generalizes known Contraction Decomposition Theorem. Prior studies on this topic exhibited that the classes of planar graphs~[Klein, SICOMP, 2008], graphs of bounded genus~[Demaine, Hajiaghayi and Mohar, Combinatorica 2010] and $H$-minor free graphs~[Demaine, Hajiaghayi and Kawarabayashi, STOC 2011] admit a Contraction Decomposition Theorem. In particular we show the following. 
Let $G$ be a graph of bounded genus, or more generally, an $h$-almost-embeddable graph for an arbitrary but fixed constant $h$.
Then for every $p \in \mathbb{N}$, there exist disjoint sets $Z_1,\dots,Z_p \subseteq V(G)$ such that for every $i \in \{1,\dots,p\}$ and every $Z' \subseteq Z_i$, the treewidth of $G / (Z_i \backslash Z')$ is upper bounded by $O(p+|Z'|)$, where the constant hidden in $O(\cdot)$ depends on $h$.
Here $G/(Z_i \backslash Z')$ denotes the graph obtained from $G$ by contracting every edge with both endpoints in $Z_i \backslash Z'$. When $Z'=\emptyset$, this corresponds to classical Contraction Decomposition Theorem.






\end{abstract}

\newpage

\pagenumbering{arabic}

\section{Introduction}
The study of subexponential time parameterized algorithms on planar and $H$-minor free graphs has been one of the most active sub-areas of parameterized algorithms, which led to exciting results and powerful methods. Examples include Bidimensionality~\cite{demaine2005subexponential}, applications of Baker's layering technique~\cite{FominLMPPS16,DornFLRS13,Tazari12,BuiP92}, bounds on the treewidth of the solution~\cite{FominLKPS20,KleinM14,KleinM12,MarxPP18}, and pattern coverage~\cite{FominLMPPS16,Nederlof20a}. The central theme of all these results is that planar graphs exhibit the ``{\em square root phenomenon}'': parameterized problems whose fastest parameterized algorithm run in time $f(k)n^{O(k)}$ or $2^{O(k)}$ on general graphs admit $f(k)n^{O(\sqrt{k})}$ or even $2^{O(\sqrt{k})}n^{O(1)}$ time algorithms when input is restricted to planar or  $H$-minor free graphs. 

Another central research direction within parameterized algorithms is the study of cut and cycle hitting problems such as {\sc Multiway Cut}, {\sc Multicut}, {\sc Odd Cycle Transversal} and {\sc Directed Feebdack Vertex Set}. The existence of $f(k)n^{O(1)}$ time algorithms for these problems were considered major open problems~\cite{DowneyF99}, and the algorithms developed for these problems~\cite{BousquetDT18,ChenLLOR08,MarxR14,Marx06,ReedSV04} and their generalizations~\cite{ChitnisCHPP16,ChitnisHM13,ChitnisCHM15,CyganKLPPSW21,CyganPPW13,GokeMM20,KratschPPW15,KimKPW21,LokshtanovMRS17,LokshtanovPSSZ20,MarxOR13,Wahlstrom17} have become standard textbook material~\cite{CyganFKLMPPS15}. 
Our paper is concerned with the following question: 
\begin{tcolorbox}
\begin{center}
Does the square root phenomenon apply to cut and cycle hitting problems?
\end{center}
\end{tcolorbox}
\smallskip
\noindent 
Despite substantial previous progress this question still remains largely unresolved. In fact, for most of these problems even an algorithm with running time $n^{O(\sqrt{k})}$  is not known, when input is restricted to planar or  $H$-minor free graphs. Here $k$ is solution size. In this paper, we make substantial progress on this question. 

\subsection{Previous work on cut and cycle hitting problems on planar graphs}
In most cut problems, we have an input graph $G$, a set of terminals $T\subseteq V(G)$, a set 
 $R$ of unordered pairs of vertices in $T$, called terminal pairs, and a positive integer $k$, and the objective is to test whether we can delete at most $k$ edges (or vertices), say $S$, so that ``certain demands on  terminal pairs in $R$ are met''.  When we demand that  for every terminal pair 
 $(t_1,t_2)\in R$, the vertices $t_1$ and $t_2$ lie in different connected components of $G-S$, this corresponds to \probEMultiC. When the set $R$ of pairs of terminals comprises all the pairs of vertices in the set $T$ of terminals, it corresponds to \probEMC (also known as {\sc Multiterminal Cut}). Thus, naturally cut problems have three parameters: (a) size of the set of terminal ($r)$ (or the the number of request pairs); (b) the solution size $k$; and (c) the combined parameter $r+k$. 

Dahlhaus et al.~\cite{DahlhausJPSY94} initiated an algorithmic study of  \probEMC, and showed that the problem is NP-complete on general graphs  when $|T|\geq 3$, and  on planar graphs when $|T|$ is part of the input.  They complemented the result on planar graphs by designing an algorithm with running time $O((4r)^r n^{2r-1}\log n)$. Klein and Marx, revisited this problem in $2012$, and designed an algorithm with running time 
$2^{O(r)} n^{O(\sqrt{r})}$~\cite{KleinM12}. Marx~\cite{Marx12} also showed that, this algorithm is essentially tight. That, is they showed that the problem is W[1]-hard. In fact, Marx~\cite{Marx12} showed that unless Exponential Time Hypothesis (ETH) is false, there is no algorithm solving the problem in time $f(r)n^{O(\sqrt{r})}$, for some function $f$ depending on $r$ alone.  de Verdi\`{e}re~\cite{Verdiere17} generlized this result to \probEMultiC and to graphs of genus $g$. In particular, 
he designed an algorithm with running time $(g+r)^{O(g+r)} n^{O(\sqrt{g^2+gr})}$~\cite{Verdiere17}; which is shown to be optimal in~\cite{Cohen-AddadVMM19}. Finally, Cohen-Addad et al.~\cite{Cohen-AddadVM21} designed $(1+\epsilon)$-approximation algorithm 
for \probEMultiC running in time $(g+r)^{O((g+r)^3)} (\frac{1}{\epsilon})^{O(\sqrt{g+r})}n \log n$. In the last two results $r=|R|$. 

While, \probEMC and  \probEMultiC, parameterized by $|T|$ and $|R|$, on planar graphs and graphs of bounded genus have been studied in literature, there is only one such result with respect to solution size $k$, as a parameter.  In a seminal result Pilipczuk et al.~\cite{PilipczukPSL18} designed a polynomial kernel for \probSteiner on graphs of bounded genus and utilized it together with a known algorithm to design $2^{O(\sqrt{k \log k})} n^{O(1)}$ time algorithm for \probSteiner on graphs of bounded genus. Using a duality between  \probSteiner  and \probEMC  on planar graphs, Pilipczuk et al.~\cite{PilipczukPSL18} also designed a polynomial kernel for \probEMC on planar graphs. Using this kernel they  designed  $2^{O(\sqrt{k} \log k)} n^{O(1)}$ time algorithm for \probEMC on planar graphs.


A subexponential time algorithm for \probFVS (deleting vertices to hit all cycles) on $H$-minor free graphs  follows from the the theory of Bidimensionality~\cite{demaine2005subexponential}. However, if we are only interested in hitting odd cycles (\probOCT), theory of  Bidimensionality~\cite{demaine2005subexponential} does not apply.  Lokshtanov et al.~\cite{LokshtanovSW12} gave the first $2^{O(\sqrt{k} \log k)}n^{O(1)}$ time randomized algorithm for \probOCT on planar graphs. Later, Jansen et al.~\cite{JansenPL19} designed deterministic polynomial kernels for \probOCT and \probVMC. Using this kernel the algorithm of Lokshtanov et al.~\cite{LokshtanovSW12}  can be made deterministic. Finally, we would also like to mention results on {\sc Subset TSP}.  Marx et al.~\cite{MarxPP18} gave $2^{O(\sqrt{k} \log k)} n^{O(1)}$ time algorithm for {\sc Subset TSP} on edge-weighted directed planar graphs. Earlier,  Klein and Marx~\cite{KleinM14} obtained a similar result for  {\sc Subset TSP}  on  undirected planar graphs.

\subsection{Our results}
In this paper we take a significant step forward in understanding the parameterized complexity of cut and cycle hitting problems on $H$-minor free graphs. 

\smallskip
\begin{tcolorbox}
\begin{center}
We provide an algorithmic framework that allows us to design the first $n^{O(\sqrt{k})}$ time algorithm for several cut and  cycle hitting problems on $H$-minor free graphs in a uniform way. This includes  the first $n^{O(\sqrt{k})}$ time algorithm for \probEMC, \probVMC, \probEB, and \probOCT on 
$H$-minor free graphs. 

\end{center}
\end{tcolorbox}

\smallskip

Algorithms running in time $n^{O(\sqrt{k})}$ are subexponential but they are not parameterized  subexponential algorithms; that is they are not of the form $2^{o(k)}n^{O(1)}$.   We incorporate 
a notion of candidate sets, a notion weaker than polynomial kernels, in our framework and design parameterized  subexponential algorithms for several aforementioned problems.  A parameterized problem is said to admit a {\em  kernel}, if there is a polynomial time algorithm (the degree of polynomial is independent of $k$), called a kernelization algorithm, that reduces the input instance  $(I,k)$ down to an instance $(I',k')$ with size bounded by a function $p(k)$ in $k$, while preserving the answer. This reduced instance is called a $p(k)$ kernel for the problem. If $p(k)$ is (quasi)-polynomial  then we say 
that a parameterized problem admits a (quasi)-polynomial kernel~\cite{CyganFKLMPPS15,fomin2019kernelization}. Randomized kernels of 
(quasi)-polynomial size is known for several cut problems~\cite{KratschW14,KratschW20,HolsK18,Wahlstrom20,corr/abs-2002-08825} on general graphs. For these problems, 
it is natural to first apply the kernelization algorithm and then run the $n^{O(\sqrt{k})}$ time algorithm on the obtained kernel to design parameterized  subexponential algorithms. There are two difficulties with this approach. 
\begin{enumerate}
\setlength{\itemsep}{-2pt}
\item  Even though we start with a graph $G$ that excludes some fixed graph $H$ as a minor, it is not necessary that the output graph $G'$ will also exclude $H$ as a minor, as we are running a kernelization algorithm for general graphs. 
\item  The reduced parameter $k'$ could become $k^{O(1)}$, and thus running $n^{O(\sqrt{k})}$ time  algorithm on kernel may not yield the desired result. 
\end{enumerate}
 
In most of the known kernelization algorithm $k'\leq k$ and hence, the second difficulty can be overcome. However, the first constraint is more challenging and requires making the known kernelization algorithm return an instance in the desired family of graphs. That is, if $\cal G$ is the family of input graphs, then we would like the kernelization algorithm also to return the output graph $G' \in {\cal G}$. We overcome this difficulty by observing that known kernelization algorithm for cut problems can be made to output  a 
``candidate set'' of size (quasi)-polynomial in $k$. A candidate set is a set $\mathsf{Cand} \subseteq V(G)$ (resp., $\mathsf{Cand} \subseteq E(G)$), satisfying that 
$(I,k)$ has a solution of size at most $k$ if and only if $I$ has a solution of size at most $k$ that is contained in $\mathsf{Cand}$. By incorporating candidate set in our framework, we obtain the following parameterized  subexponential algorithms. 

%


\begin{itemize}
\setlength{\itemsep}{-1pt}
\item In the \probOCT (resp., \probEB) problem input is a graph $G$ and integer $k$. The task is to determine whether there exists a set $S$ of at most $k$ vertices (resp., edges) so that $G - S$ is bipartite. 
We give a $2^{O(\sqrt{k}\log k)} \cdot n^{O(1)}$ time algorithm for \probOCT and for \probEB. 

\item In the \probVMC (resp., \probEMC) problem input is a graph $G$ a vertex set $T \subseteq V(G)$ of size $r$, and an integer $k$. The task is to determine whether there exists a set $S$ of at most $k$ vertices in $V(G) \setminus T$ (resp., edges in $E(G)$) such that no connected component of $G - S$ contains at least two vertices from $T$. We obtain a $2^{O(r \sqrt{k} \log k)} \cdot n^{O(1)}$ time algorithm for \probVMC and a $2^{O(\sqrt{k}\log^4 k)} \cdot n^{O(1)}$ time algorithm for \probEMC. Note that the running time of the algorithm for \probVMC depends on $r$ while the one for \probEMC does not. 

\item In the \probVMultiC (resp., \probEMultiC) problem input is a graph $G$, together with $r$ vertex pairs $(s_1, t_1), \ldots, (s_r, t_r)$. The task is to determine whether there exists a vertex subset $S \subseteq V(G) \setminus \{s_i, t_i ~:~ i \leq r\}$ (resp., edge subset $S \subseteq E(G)$) such that no connected component of $G - S$ contains both $s_i$ and $t_i$ for any $i \leq r$. We obtain a $2^{O((\sqrt{rk}+r) \log (rk))} \cdot n^{O(1)}$ time algorithm for \probVMultiC and a $2^{O((r+\sqrt{k})\log^4 (rk))} \cdot n^{O(1)}$ time algorithm for \probEMultiC.

\item In the \probGFVS (resp., \probGFES) problem input is a graph $G$, together with a function $\Lambda : V(G) \times V(G) \rightarrow \Sigma$ and an integer $k$, where $\Sigma$ is a group of size $g$ and $\Lambda$ satisfies $\Lambda(u,v) \times \Lambda(v,u) = \mathbf{1}$. The task is to determine whether there exists $S \subseteq V(G)$ (resp., $S \subseteq E(G)$) of size at most $k$ such that $G-S$ has no {\em non-null cycles}. Here a non-null cycle is a cycle $(v_0,v_1,\dots,v_m = v_0)$ in the graph satisfying $\prod_{i=1}^m \Lambda(v_{i-1},v_i) \neq \mathbf{1}$.
We obtain a $2^{O(g \sqrt{k}\log(gk))} \cdot n^{O(1)}$ time algorithm for \probGFVS and a $2^{O(\sqrt{k} \log g \log^4 k)} \cdot n^{O(1)}$ time algorithm for \probGFES.
\end{itemize}

For convenience, in the rest of the paper, we use the following abbreviations for the problem names: \probOCT (OCT), \probEB (EB), \probVMC (VMWC), \probEMC (EMWC), \probVMultiC (VMC), \probEMultiC (EMC), \probGFVS (GFVS), and \probGFES (GFES).

\subsection{Methods: algorithms via a contraction decomposition theorem}
One of the modern tools to design polynomial-time approximation schemes (PTASs) and FPT algorithms  on planar graphs or more generally on $H$-minor free graphs is to prove  strengthening of the classic Baker's layering technique~\cite{Baker94}. This generally yields, in what is known as (Vertex) Edge Decomposition Theorems~\cite{Baker94,eppstein2000diameter,DemaineHK05,DeVosDOSRSV04,DBLP:conf/soda/Dvorak18} (see~\cite{Panolan0Z19} for a detailed introduction). There are several problems for which an optimal solution with respect to subgraphs or minors of a graph is not larger than an optimal solution with respect to the graph itself. For those, (Vertex) Edge Decomposition Theorem~\cite{Baker94,eppstein2000diameter,DemaineHK05,DeVosDOSRSV04,DBLP:conf/soda/Dvorak18}  are much more relevant. 
However, it is easy to observe that classical problems such as {\sc Dominating Set}, {\sc Bisection} and {\sc Traveling Salesman Problem (TSP)}, are not  closed under taking subgraph or minors. Nevertheless, they are closed under another useful graph operation, namely, {\em edge contraction}. That is, if $G$ is a graph and $H$ is obtained from $G$ by a series of edge contractions, then the ``solution size\rq{}\rq{} of $H$ is not larger than that of $G$. Such problems naturally led to the following notion of {\em contraction decomposition}, which is most relevant to our studies in this paper.

\medskip 
\noindent 
[{\bf The notion of a Contraction  Decomposition Theorem (CDT)}]
{\sl Let $G$ be a graph (say, which belongs to a prespecified graph class $\cal C$) and let $k\in\mathbb{N}$. Then, the edge set of $G$ can be partitioned into $k + 1$ sets (possibly empty) in such a way that contracting any of these sets in $G$ yields a graph of treewidth at most $f(k)$. Moreover, such a partition, together with tree decompositions of width at most $f(k)$ of the respective graphs, can be found in polynomial time.}

\medskip
A CDT is known for planar graphs~\cite{Klein08,Klein06} (initially for a variation of contraction called compression), graphs of bounded genus~\cite{DemaineHM10}, $H$-minor free graphs~\cite{DemaineHK11} and unit disk graphs~\cite{Panolan0Z19}. In all of these works, the treewidth is bounded by a linear function of $k$. That is, $f(k)=O(k)$. However, for most known applications on these classes of graphs, any computable function $f$ of $k$ suffices.   In this paper, we prove a generalization of CDT on graph of bounded genus, or more generally, on $h$-almost-embeddable graph for an arbitrary but fixed constant $h$, and apply a non-trivial dynamic programming algorithm on Robertson-Seymour decomposition theorem for $H$-minor free graphs~\cite{robertson2003graph}. 


To describe our results and especially the new CDT, we only focus on planar graphs in the introduction.
Our generalization can be expressed in the context of planar graphs as follows.

\begin{lemma}\label{lem:exposition}
Let $G$ be a planar graph.
Then for any $p \in \mathbb{N}$, there exist disjoint sets $Z_1,\dots,Z_p \subseteq V(G)$ such that for every $i \in [p]$ and every $Z' \subseteq Z_i$, $\mathbf{tw}(G/(Z_i \backslash Z')) = O(p+|Z'|)$.
Furthermore, the sets $Z_1, \dots, Z_p$ can be computed in polynomial time. 
\end{lemma}

Here $G/V$ denote the graph obtained from $G$ by contracting each edge in the induced subgraph $G[V]$.
To see why the above lemma implies the contraction decomposition theorem, let $E_1,\dots,E_p$ be the edges in $G[Z_1],\dots,G[Z_p]$.
Since $Z_1, \dots, Z_p$ are disjoint, $E_1,\dots,E_p$ are also disjoint and thus there exists a partition $E_1',\dots,E_p'$ of the edges of $G$ such that $E_i \subseteq E_i'$.
Then we have $\mathbf{tw}(G/E_i') \leq \mathbf{tw}(G/E_i) = \mathbf{tw}(G/Z_i) = O(p)$.
A very important difference however is that, in our result, we still have control on the treewidth of the resulting graph, even if we do not contract all the edges of each $E_i$. 

To see an example where this difference is crucial, let us consider the \probEB problem on planar graphs, where the goal is to find a set $S$ of $k$ edges such that $G-S$ is bipartite. Using known techniques (see Lemma \ref{lem-candidate} below), we can find in polynomial time a set $\mathsf{Cand}$ of size $k^{O(1)}$ such that we only need to look inside $\mathsf{Cand}$ for the solution. Let $Z_1, \dots, Z_p$ be the sets obtained from applying Lemma \ref{lem:exposition} with $p = \sqrt{k}$, and $E_1,\dots,E_p$ be the sets of edges in $G[Z_1],\dots,G[Z_p]$. Because we are looking for a solution $S$ of size $k$ and $E_1,\dots,E_p$ are disjoint, it means that one of these sets, say $E_i$, contains less than $\sqrt{k}$ edges of $S$. As we are looking for a solution $S \subseteq \mathsf{Cand}$, there are at most $k^{O(\sqrt{k})}$ possible sets for $S_i :=S \cap E_i$. Therefore, by trying all the possibilities of $i$ and $S_i$, we can assume that the algorithm knows $S_i$. Denoting $Z'$ the set of vertices adjacent to edges in $S_i$, it means that we have $|Z'| = O(\sqrt{k})$ and thus our lemma implies that $\mathbf{tw}(G/(Z_i \setminus Z')) = O(\sqrt{k})$. A pretty standard DP argument shows that \probEB on a graph $H$ can be solved in time $2^{O(\mathbf{tw}(H))}n^{O(1)}$, and we will be able to adapt this argument to solve \probEB in $G$ in time $2^{O(\mathbf{tw}(G/(Z_i \setminus Z')))}n^{O(1)} = 2^{O(\sqrt{k})}n^{O(1)}$. Indeed, while some of the vertices in $G/(Z_i \setminus Z')$ can correspond to a large connected component $C$ of $G[Z_i \setminus Z']$, the fact that we have contracted only edges which are \textit{not} in the solution means that $C$ is bipartite, and in the DP we only have to guess which side of $C$ will belongs to which side of the final bipartition of $G-S$. 

A similar approach using only the contraction decomposition Theorem would fail, as it requires that one set completely avoids the solution. Therefore $p$ needs to be greater than $k$, and we obtain a decomposition of a graph after contraction of treewidth at least $k$, which cannot lead to a subexponential-time algorithm.

In order to prove Lemma \ref{lem:exposition} we will consider the vertex-face incidence (VFI) graph of $G$, which is a bipartite graph containing on one side the set of vertices of $G$ and on the other side the set of faces. Moreover, a vertex is adjacent to a face if it belongs to its boundary. Classically, if the VFI graph of a planar graph has diameter at most $p$, then $\mathbf{tw}(G) = O(p)$.
The \textit{deep} faces of $G$ refer to the faces farthest from the outer face $o$ of $G$, i.e., the faces have the maximum distance from $o$ in the VFI of $G$.
The following lemma is an essential part of our proof (which is a special case of Lemma~\ref{lem-twdiam2}).

\begin{lemma}\label{lem:exposition_planar}
    Let $G$ be a planar graph such that the VFI graph of $G$ has diameter $p$.
    For each deep face $f$ of $G$, let $\kappa(f)$ be a set of at most $q$ vertices on the boundary of $f$.
    Then there exists a tree decomposition of $G$ of width $O(p+q)$ where for every deep face $f$ of $G$, $\kappa(f)$ is contained in some bag of the decomposition. 
\end{lemma}

The key idea for proving the above lemma (without going into details) is the following.
We draw inside each deep face $f$ of $G$ a $|\kappa(f)| \times |\kappa(f)|$ grid and connecting the vertices in the first row of the grid to the vertices in $\kappa(f)$.
The resulting graph $G'$ is still planar.
Since the grids are drawn in the deep faces of $G$, the VFI graph of $G'$ has diameter $O(p+q)$, so there is a tree decomposition of $G'$ of width $O(p+q)$.
We then exploit the well-linked property of grids to show how to modify this tree decomposition (without significantly increasing its width) to force the vertices in each $\kappa(f)$ to be contained in the same bag.

With Lemma~\ref{lem:exposition_planar}, let us explain roughly how to obtain the sets $Z_1, \dots, Z_p$ of Lemma \ref{lem:exposition}. Let $G$ be a planar graph and assumed that the embedding is fixed. Consider now a BFS in the VFI graph of $G$ starting from the outer face. Let $L_0, L_1, \dots, L_m$ denote the different levels of the BFS. In particular, we will be interested in the odd levels $L_{2i + 1}= L'_i$ for $i \leq m/2$. As $L_0$ is the outer face and the VFI graph is bipartite, the $L'_i$ form a partition of the vertices of $G$. An interesting way to understand the $L'_i$ is by noting that $L'_1$ is exactly the boundary of the outer face of $G$, $L'_2$ is exactly the boundary of the outerface of $G - L'_1$ and so on.   

For every $i \in [p]$, we now define each $Z_i$ as the union of the $L'_j$ for $j = i \pmod p$. Let $Z'$ be a subset of vertices in $Z_i$. In order to understand what happens when we contract the edges of $G[Z_i \setminus Z']$, let us partition the graph into $G_{\leq} := G[L'_1 \cup \cdots \cup L'_i ]$ and $G_{>} := G[L'_{i+1} \cup \cdots \cup L'_{m/2} ]$ the rest. The idea is to apply some sort of inductive argument in $G_>$ and then combine the tree decomposition with the one of $G_{\leq}/(Z_i \setminus Z')$. First note that, because $G_{\leq}$ consists of the first $2i < 2p$ layers of the BFS in the VFI graph of $G$, we have that the diameter in the VFI graph of $G_{\leq}$ is $O(p)$. Moreover, we know $L'_i$ is the outerface of $G[L'_{i+1} \cup \dots \cup L'_{m/2} ]$ which means that $L'_i$ is a cut between $G_{\leq}$ and $G_>$. Moreover, it means that every connected component of $G_>$ sits inside one face of $G_{\leq}$ whose boundary is a cycle of $L'_i$. Let $C_1, \dots, C_r$ denote the connected components of $G_>$ and $\partial_1, \dots, \partial_r$ the set of cycles bounding the associated faces $f_1,\dots,f_r$ (we assume here for simplification that $f_1,\dots,f_r$ are different, and it is easy to see that they are all deep faces of $G$).
As all the $\partial_i$ are inside $L_i$, it means that in $G/(Z_i \backslash Z')$, each $\partial_j$ will be contracted into a cycle $\partial'_j$ of size $1 + |Z' \cap \partial_j| = O(|Z'|)$.

The idea of the proof is then to compute inductively a tree decomposition for each $C_j/(Z_i \backslash Z')$ as well as $G_\leq/(Z_i \backslash Z')$ and merge them together. In order to do that efficiently, we require that each $\partial'_j$ appears in one bag of the decomposition of $G_\leq/(Z_i \backslash Z')$ so that we can append the tree decomposition of $C_j/(Z_i \backslash Z')$ to this bag. In order to obtain such a tree decomposition of $G_\leq/(Z_i \backslash Z')$ we will use Lemma \ref{lem:exposition_planar}.
We let $\kappa(f_j)$ include, for every vertex $v$ in the contracted cycle $\partial'_j$, a vertex in the cycle $\partial_j$ corresponding to $v$.
Thus, $|\kappa(f_j)| = |\partial'_j| = O(|Z'|)$.
By Lemma \ref{lem:exposition_planar}, $G_\leq$ has a tree decomposition of width $O(p+|Z'|)$ in which every $\kappa(f_j)$ is contained in some bag.
This tree decomposition then induces an $O(p+|Z'|)$-width tree decomposition of $G_\leq/(Z_i \backslash Z')$ in which every $\partial'_j$ is contained in some bag.

So far we sketched our results in only the context of planar graphs.
Of course, our actual results are much more general, and require more technical proofs.
We prove Lemma~\ref{lem:exposition} for \textit{almost-embeddable} graphs, which is a class of graphs generalizing the class of bounded-genus graphs.
With this in hand, we solve our problems by doing dynamic programming on the famous ``almost-embeddable'' tree decomposition of an $H$-minor-free graph introduced by Robertson and Seymour \cite{robertson2003graph}.
Specifically, the generalized version of Lemma~\ref{lem:exposition} allows us to apply Baker's technique with a second-level dynamic programming to efficiently compute the DP table at each node of the Robertson-Seymour tree decomposition.

\paragraph{Organization.}
The rest of the paper is organized as follows.
In Section~\ref{sec:framework}, we present the generic framework for solving our problems on $H$-minor-free graphs, which all of our algorithms follow.
In Section~\ref{sec-keyproof}, we prove our main technical lemma (Lemma~\ref{lem-key}).
In Section~\ref{sec-app}, we apply our framework in Section~\ref{sec:framework} to solve each of our problems individually.
Finally, in Section~\ref{sec-conclusion}, we conclude the paper.

\section{Preliminaries}

\paragraph{Basic notations.}
Let $G$ be a graph.
We use $V(G)$ and $E(G)$ to denote the vertex set and the edge set of $G$, respectively.
For $V \subseteq V(G)$, we use $G[V]$ to denote the induced subgraph of $G$ on $V$ and use $G - V$ to denote the induced subgraph of $G$ on $V(G) \backslash V$.
The notation $G/V$ denotes the graph obtained from $G$ by contracting all edges in $G[V]$, or equivalently, contracting every connected component of $G[V]$ to a single vertex.
We denote by $N_G(V)$ the set of vertices in $V(G) \backslash V$ that are neighboring to at least one vertex in $V$.

\paragraph{Tree decomposition and treewidth.}
A \textit{tree decomposition} $\mathcal{T}$ of a graph $G$ is a tree $T$ where each node $t \in T$ is associated with a bag $\beta(t) \subseteq V(G)$ such that
\textbf{(i)} $\bigcup_{t \in T} \beta(t) = V(G)$, \textbf{(ii)} for any edge $(u,v) \in E(G)$, there exists $t \in T$ with $u,v \in \beta(t)$, and \textbf{(iii)} for any $v \in V(G)$, the nodes $t \in T$ with $v \in \beta(t)$ form a connected subset in $T$.
The \textit{width} of $\mathcal{T}$ is $\max_{t \in T} |\beta(t)| - 1$.
The \textit{treewidth} of a graph $G$, denoted by $\mathbf{tw}(G)$ is the minimum width of a tree decomposition of $G$.
It is sometimes more convenient to consider \textit{rooted} trees.
Throughout this paper, we always view the underlying tree of a tree decomposition as a rooted tree.
A tree decomposition is \textit{binary} if its underlying (rooted) tree is binary.
The following result is well-known.

\begin{lemma}[\cite{CyganFKLMPPS15}] \label{lem-tdecomp}
Given a graph $G$ with $|V(G)| = n$ and $\mathbf{tw}(G) = w$, a binary tree decomposition of $G$ of width $O(w)$ can be computed in $2^{O(w)} n^{O(1)}$ time.
\end{lemma}

Let $\mathcal{T}$ be a tree decomposition of a graph $G$, and $T$ be the underlying (rooted) tree of $\mathcal{T}$.
The \textit{adhesion} of a non-root node $t \in T$, denoted by $\sigma(t)$, is defined as $\sigma(t) = \beta(t) \cap \beta(t')$ where $t'$ is the parent of $t$.
For convenience, we also define the adhesion of the root of $T$ as the empty set.
The \textit{adhesion size} of $\mathcal{T}$ is the maximum size of the adhesion of a node $t \in T$.
For a node $t \in T$, we define the \textit{$\gamma$-set} of $t$, denoted by $\gamma(t)$, as $\gamma(t) = \bigcup_{s \in T_t} \beta(s)$ where $T_t$ is the subtree of $T$ rooted at $t$.
The \textit{torso} of $t$, denoted by $\mathsf{tor}(t)$, is the graph obtained from $G[\beta(t)]$ by making $\sigma(s)$ a clique for all children $s$ of $t$, i.e., adding edges between any two vertices $u,v \in \beta(t)$ such that $u,v \in \sigma(s)$ for some child $s$ of $t$.

\paragraph{Almost-embeddable graphs and graph minors.} 
The class of almost-embeddable graphs is a generalization of the class of bounded-genus graphs.
The formal definition of almost-embeddable graphs (Definition~\ref{def-almost}) will be given in Section~\ref{sec-keyproof}, because it is technical and we only need the formal definition in Section~\ref{sec-keyproof} when proving Lemma~\ref{lem-key}.
Roughly speaking, a graph $G$ is \textit{$h$-almost-embeddable} if it can be partitioned into an apex set $A \subseteq V(G)$ of size at most $h$, a subgraph $G_0$ that can be embedded in a surface of (Euler) genus $h$, and $h$ disjoint vortices $G_1,\dots,G_h$ which admit $h$-width path decompositions and can be attached to faces of (the embedded) $G_0$ in a well-structured way.
The \textit{almost-embeddable structure} of $G$ refers to the ``witness'' of the almost-embeddability of $G$ (i.e., the partition of $G$ into the apex set, the embeddable subgraph $G_0$, and the vortices, together with the way that the vortices are attached to $G_0$).

A graph $H$ is a \textit{minor} of a graph $G$ if $H$ can be obtained from $G$ by deleting vertices, deleting edges, and contracting edges.
A graph $G$ is \textit{$H$-minor-free} if $H$ is not a minor of $G$.
One of the most important results in graph minor theory by Robertson and Seymour \cite{robertson2003graph} states that every $H$-minor-free graph (for a fixed $H$) admits a tree decomposition with $O(1)$ adhesion size and $O(1)$-almost-embeddable torsos.

\begin{lemma}[Robertson-Seymour decomposition] \label{lem-rsdecomp}
Let $H$ be a fixed graph.
Then for some constant $h>0$ depending on $H$, every $H$-minor-free graph $G$ admits a tree decomposition $\mathcal{T}_\textnormal{RS}$ with adhesion size at most $h$ in which the torso of each node is $h$-almost-embeddable.
In addition, $\mathcal{T}_\textnormal{RS}$ has the property that $G[\gamma(t) \backslash \sigma(t)]$ is connected and $\sigma(t) = N_G(\gamma(t) \backslash \sigma(t))$ 
for all node $t$.
Such a tree decomposition (and the almost-embeddable structures of the torsos) can be computed in polynomial time.
\end{lemma}
\begin{proof}[Proof sketch]
The fact that any $H$-minor free graph $G$ admits a tree decomposition $\mathcal{T}_\textnormal{RS}$ of adhesion size at most $h$ whose torsos are $h$-almost-embeddable (for some constant $h>0$ depending on $H$) follows from the profound work of Robertson and Seymour \cite{robertson2003graph}; we call such a tree decomposition \textit{Robertson-Seymour} decomposition.
Several algorithms have been developed to compute a Robertson-Seymour decomposition (and the almost-embeddable structures of the torsos) in polynomial time \cite{demaine2005algorithmic,grohe2013simple,kawarabayashi2011simpler}.
To further make the tree decomposition $\mathcal{T}_\textnormal{RS}$ satisfy the additional property, we do some simple modifications on $\mathcal{T}_\textnormal{RS}$.
Roughly speaking, for each node $t$ where $G[\gamma(t) \backslash \sigma(t)]$ is disconnected, we split the subtree rooted at $t$ into multiple subtrees each of which corresponds to a connected component of $G[\gamma(t) \backslash \sigma(t)]$.
For each node $t$ where $\sigma(t) \neq N_G(\gamma(t) \backslash \sigma(t))$, we remove the vertices in $\sigma(t) \backslash N_G(\gamma(t) \backslash \sigma(t))$ from the bags of all nodes in the subtree rooted at $t$.
These modifications are standard, so we defer the details to Appendix~\ref{appx-proofrs}.
We can guarantee that after these modifications, $\mathcal{T}_\textnormal{RS}$ still has bounded adhesion size and almost-embeddable torsos.
Thus, the resulting $\mathcal{T}_\textnormal{RS}$ is what we want.
\end{proof}

\paragraph{Candidate sets for vertex/edge-deletion problems.}
All problems studied in this paper are \textit{vertex-deletion} (resp., \textit{edge-deletion}) problems, in which a feasible solution is a set of vertices (resp., edges) of the input graph such that after deleting them the resulting graph satisfies certain conditions.
Let $I$ be an instance of a parameterized vertex-deletion (resp., edge-deletion) problem with input graph $G$ and solution-size parameter $k \in \mathbb{N}$.
A \textit{candidate set} for $I$ is a set $\mathsf{Cand} \subseteq V(G)$ (resp., $\mathsf{Cand} \subseteq E(G)$) such that $I$ has a solution of size at most $k$ iff $I$ has a solution of size at most $k$ that is contained in $\mathsf{Cand}$.
It was known that all problems studied in this paper admit small candidate sets that can be computed using polynomial-time randomized algorithms with high success probability.

\begin{lemma}[small candidate set lemma] \label{lem-candidate}
One can compute candidate set for the following problems in polynomial time, by randomized algorithms, with success probability at least $1-1/2^n$ (where $n$ is the size of the input graph): 
candidate sets for EB and OCT instances of size $k^{O(1)}$, EMWC and GFES of size $k^{O(\log^3 k)}$, EMC of size $(r+k)^{O(\log^3(r+k))}$, VMWC of size $k^{O(r)}$, VMC of size $k^{O(\sqrt{r})}$, GFVS of size $k^{O(g)}$. 
\end{lemma}
\begin{proof}[Proof sketch]
Proof of this lemma is implicit in literature \cite{KratschW20,corr/abs-2002-08825}.
Known kernels for all these problems actually first compute a candidate set and then apply reduction rules (in fact, an appropriate torso operation) to reduce the graph size.
The candidate sets for EB, OCT, VMWC, VMC, and GFVS follow from the work \cite{KratschW20}, while the candidate sets for EMWC, EMC, and GFES follow from the work \cite{corr/abs-2002-08825}.
See Appendix~\ref{appx-proofcand} for a more detailed discussion about the references for each of these problems.
\end{proof}

\section{The generic framework}
\label{sec:framework}

All of our algorithms follow the same framework.
Before discussing the framework, let us first observe some common features of the problems to be studied.
\begin{enumerate}
    \item \textbf{Small candidate sets.}
    All problems we study admit small-size candidate sets that can be computed efficiently (Lemma~\ref{lem-candidate}).
    Our framework will make use of these small candidate sets.
    However, even without the candidate sets in Lemma~\ref{lem-candidate}, our framework still works and leads to algorithms with $n^{O(\sqrt{k})}$ kind running time, by using the trivial candidate set, i.e., the entire vertex set or edge set.
    \item \textbf{Tree-decomposition dynamic programming.}
    Given a tree decomposition $\mathcal{T}$ of the input graph of width $w$, all problems we study can be solved by applying standard dynamic programming (DP) on $\mathcal{T}$ in time depending on $w$.
    For example, consider the OCT problem on a graph $G$.
    Let $\mathcal{T}$ be a tree decomposition of $G$ of width $w$, and $T$ be the underlying (rooted) tree of $\mathcal{T}$.
    We can do DP on $\mathcal{T}$ by considering the nodes of $T$ from bottom to top.
    Specifically, we compute a DP table at each node $t \in T$ in which every entry corresponds to a sub-problem on $G[\gamma(t)]$ with a constraint on the ``state'' of the vertices in $\sigma(t)$ in the solution: which vertices in $\sigma(t)$ are in the OCT, which vertices are on one side of the remaining bipartite graph, and which vertices are on the other side; the entry stores an optimal solution for that sub-problem.
    Since $\beta(t) \leq w+1$, one can compute the DP table at $t$ in $2^{O(w)}$ time provided the DP tables of all children of $t$, by simply enumerating the states of the vertices in $\beta(t)$.
    In this way, we solve the OCT problem in $2^{O(w)} \cdot n^{O(1)}$ time.
    \item \textbf{Contraction-friendly for tree-decomposition DP.}
    On top of the solvability by tree-decomposition DP, the problems we study have a very important feature, which we call \textit{contraction-friendly}.
    That is, if we know that some parts of the graph is disjoint from the solution, then we can contract these parts and do DP on a tree decomposition $\mathcal{T}$ of the graph \textit{after the contraction}, in time depending on the width of $\mathcal{T}$.
    For example, consider again the OCT problem on a graph $G$.
    Suppose we somehow know that a set $X \subseteq V(G)$ of vertices are not in the solution OCT, so $G[X]$ must be bipartite.
    We then contract the connected components of $G[X]$ and compute a tree decomposition $\mathcal{T}$ for the resulting graph $G/X$ of width $w = O(\mathbf{tw}(G/X))$.
    Let $T$ be the underlying tree of $\mathcal{T}$.
    The bag $\beta(t)$ of each node $t \in T$ contains at most $w+1$ vertices of $G/X$, each of which either corresponds to a single vertex in $G$ or a connected component of $G[X]$.
    How can we do DP on $\mathcal{T}$ in time depending on $w$?
    We do it in a similar fashion as above.
    At each node $t \in T$, we compute a DP table in which every entry corresponds to a sub-problem on the pre-image of $(G/X)[\gamma(t)]$ in $G$ with a constraint on the state of the vertices in the pre-image of $\sigma(t)$ in the solution.
    However, the pre-image of $\sigma(t)$ can have a large size, because a connected component of $G[X]$ might be large.
    The key observation here is that each connected component of $G[X]$ can have only \textit{two} possible states in the solution.
    Indeed, each connected component of $G[X]$, as a connected bipartite graph, has a \textit{unique} bipartite structure, i.e., can be \textit{uniquely} partitioned into two independent sets $\varGamma_1$ and $\varGamma_2$.
    Since $X$ is disjoint from the solution OCT, $\varGamma_1$ and $\varGamma_2$ must belong to opposite sides of the remaining bipartite graph (and hence there are only two possibilities).
    Therefore, the pre-image of $\sigma(t)$ can have in total $2^{O(|\sigma(t)|)} = 2^{O(w)}$ possible states in the solution, and thus the size of the DP table at $t$ is $2^{O(w)}$.
    For the same reason, the pre-image of $\beta(t)$ can have $2^{O(|\beta(t)|)} = 2^{O(w)}$ possible states in the solution; by enumerating these states, one can compute the DP table at $t$ in $2^{O(w)}$ time (provided the DP tables of all children of $t$).
    In this way, the entire problem can be solved in $2^{O(w)} \cdot n^{O(1)}$ time.
\end{enumerate}

Our algorithmic framework can be applied to any vertex/edge-deletion problems on $H$-minor-free graphs that have the above features.
Essentially, if we are given a candidate set $\mathsf{Cand}$ of the problem and we are able to do DP in $f(\vec{r},w) \cdot n^{O(1)}$ time on a $w$-width tree decomposition of the graph after contracting some part that is disjoint from the solution, then our framework gives an algorithm for the problem that runs in $|\mathsf{Cand}|^{O(\sqrt{k})} \cdot f(\vec{r},\sqrt{k}) \cdot n^{O(1)}$ time; here $\vec{r}$ denotes the parameter(s) of the problem other than the solution-size parameter $k$ (e.g., the number of terminals in VMWC/EMWC, the group size in GFVS/GFES, etc.).
Note that any vertex-deletion (resp., edge-deletion) problem admits a trivial candidate set of size $O(n)$ (resp., $O(n^2)$), i.e., the entire vertex set $V(G)$ (resp., edge set $E(G)$).
Therefore, even without the small candidate sets in Lemma~\ref{lem-candidate}, our framework can lead to algorithms with $f(\vec{r},\sqrt{k}) \cdot n^{O(\sqrt{k})}$ running time.

Next, we discuss our framework in detail.
Consider a problem instance with graph $G$ which is $H$-minor-free and solution-size parameter $k \in \mathbb{N}$.
First of all, we use Lemma~\ref{lem-candidate} to compute a small candidate set $\mathsf{Cand} \subseteq V(G)$ or $\mathsf{Cand} \subseteq E(G)$ for the instance.
By the definition of a candidate set, we now only need to decide whether there exists a solution inside $\mathsf{Cand}$ of size at most $k$.
Since $G$ is $H$-minor-free, by Lemma~\ref{lem-rsdecomp}, we can compute in polynomial time a tree decomposition $\mathcal{T}_\textnormal{RS}$ of $G$ with adhersion size at most $h$ in which each torso is $h$-almost-embeddable (for some constant $h>0$), together with the almost-embeddable structures of the torsos.
Let $T_\textnormal{RS}$ be the underlying (rooted) tree of $\mathcal{T}_\textnormal{RS}$.
By Lemma~\ref{lem-rsdecomp}, $\mathcal{T}_\textnormal{RS}$ also satisfies the property that for all $t \in T_\textnormal{RS}$, $G[\gamma(t) \backslash \sigma(t)]$ is connected and $\sigma(t) = N_G(\gamma(t) \backslash \sigma(t))$.

In a high level, our algorithm solves the problem using dynamic programming on the tree decomposition $\mathcal{T}_\textnormal{RS}$.
However, the width of $\mathcal{T}_\textnormal{RS}$ could be large, so doing DP in a standard way does not work.
Fortunately, $\mathcal{T}_\textnormal{RS}$ has some useful properties.
First, the adhesion size of $\mathcal{T}_\textnormal{RS}$ is bounded by $h$, i.e., $|\sigma(t)| \leq h$ for all $t \in T_\textnormal{RS}$.
This means, though the width of $\mathcal{T}_\textnormal{RS}$ can be large, the DP table to be computed at each node $t \in T_\textnormal{RS}$ only has a constant size.
Therefore, the main challenge actually occurs in how to implement a single step of the DP procedure, that is, how to compute the DP table of a node $t \in T_\textnormal{RS}$, provided the DP tables of all children of $t$.
To this end, we exploit the second good property of $\mathcal{T}_\textnormal{RS}$: the torso of each node $t \in T_\textnormal{RS}$ is $h$-almost-embeddable.
In order to use this property, we establish the following structural lemma about almost-embeddable graphs, which serves as a technical core of our framework.
\begin{lemma} \label{lem-key}
Let $G$ be an $h$-almost-embeddable graph, for some constant $h$.
Then for any $p \in \mathbb{N}$, there exist disjoint sets $Z_1,\dots,Z_p \subseteq V(G)$ such that for any $i \in [p]$ and $Z' \subseteq Z_i$, $\mathbf{tw}(G/(Z_i \backslash Z')) = O(p+|Z'|)$, where the constant hidden in $O(\cdot)$ depends on $h$.
Furthermore, if the almost-embeddable structure of $G$ is given, then $Z_1,\dots,Z_p$ can be computed in polynomial time.
\end{lemma}

To see how powerful Lemma~\ref{lem-key} is, let us first consider a special case where the tree $T_\textnormal{RS}$ only has one node, the root $\mathsf{rt}$.
That is, the graph $G = G[\beta(\mathsf{rt})] = \mathsf{tor}(\mathsf{rt})$ itself is $h$-almost-embeddable.
We observe that Lemma~\ref{lem-key} almost directly solves our problems in this case (and in particular solves the problems on bounded-genus graphs).
Again, we use the OCT problem as an example.
Set $p = \lfloor \sqrt{k} \rfloor$ and we compute the disjoint sets $Z_1,\dots,Z_p \subseteq V(G)$ in Lemma~\ref{lem-key}.
We apply (a variant of) Baker's technique on $Z_1,\dots,Z_p$.
Specifically, we observe that $G$ has an OCT of size at most $k$ iff for some $i \in [p]$ and $Z' \subseteq Z_i \cap \mathsf{Cand}$ with $|Z'| \leq k/p$, $G$ has an OCT of size at most $k$ that is \textit{disjoint} with $Z_i \backslash Z'$.
The ``if'' part is obvious.
To see the ``only if'' part, suppose $G$ has an OCT of size at most $k$.
Then by the definition of candidate sets, $G$ has an OCT $V_\mathsf{oct}$ of size at most $k$ that is contained in $\mathsf{Cand}$.
Since $Z_1,\dots,Z_p$ are disjoint, there exists some $i \in [p]$ such that $|V_\mathsf{oct} \cap Z_i| \leq k/p$.
Define $Z' = V_\mathsf{oct} \cap Z_i$.
Since $V_\mathsf{oct} \subseteq \mathsf{Cand}$, we have $Z' \subseteq Z_i \cap \mathsf{Cand}$.
Also, we have $V_\mathsf{oct} \cap (Z_i \backslash Z') = \emptyset$, which implies the ``only if'' part.
Therefore, to solve the problem, it now suffices to test, for every pair $(i,Z')$ where $i \in [p]$ and $Z' \subseteq Z_i \cap \mathsf{Cand}$ with $|Z'| \leq k/p$, whether $G$ has an OCT of size at most $k$ that is disjoint from $Z_i \backslash Z'$.
How many pairs $(i,Z')$ are there to be considered?
Since $|\mathsf{Cand}| = k^{O(1)}$ for OCT, the total number of such pairs is bounded by $p \cdot k^{O(k/p)} = k^{O(\sqrt{k})}$.
To solve the problem for each pair $(i,Z')$, we recall the ``contraction-friendly'' feature of OCT (discussed at the beginning of this section).
As we know that $Z_i \backslash Z'$ is disjoint from the solution, we can contract the connected components of $G[Z_i \backslash Z']$ and do DP on a tree decomposition of the resulting graph $G/(Z_i \backslash Z')$.
By Lemma~\ref{lem-key}, $\mathbf{tw}(G/(Z_i \backslash Z')) = O(p+|Z'|) = O(\sqrt{k})$, and thus the DP procedure takes $2^{O(\sqrt{k})} \cdot n^{O(1)}$ time.
In this way, we solve the OCT problem in $k^{O(\sqrt{k})} \cdot n^{O(1)}$ time for this special case.

The general case where $|T_\textnormal{RS}| > 1$ is more involved.
Recall that we want to do a single step of DP on $\mathcal{T}_\textnormal{RS}$: computing the DP table of a node $t \in T_\text{RS}$ provided the DP table of the children of $t$.
Each entry of this DP table corresponds to a sub-problem on $G[\gamma(t)]$ with some constraint on the configuration of $\sigma(t)$ in the solution.
More formally, each such sub-problem can be represented by a pair $(\lambda,l)$ where $\lambda$ is a configuration of $\sigma(t)$ (which is problem-specific) and $l \in [k]$ is an upper bound on the solution size (we denote this sub-problem as $\mathsf{Prob}_{\lambda,l}$).
The sub-problem $\mathsf{Prob}_{\lambda,l}$ asks if there exists a solution for $G[\gamma(t)]$ of size at most $l$ that is contained in $\mathsf{Cand}$ and is \textit{compatible} with the configuration $\lambda$ of $\sigma(t)$; the corresponding entry of the DP table stores the YES/NO answer to the sub-problem $\mathsf{Prob}_{\lambda,l}$.
Thus, computing the DP table of $t$ is equivalent to solving a set of sub-problems on $G[\gamma(t)]$.
To this end, our basic idea is similar to the above, i.e., applying Baker's technique on $G[\gamma(t)]$, contracting the part that is disjoint from the solution, and doing a second-level DP on a tree decomposition of the graph after the contraction.
However, Lemma~\ref{lem-key} can only give us a decomposition on $\beta(t)$ because only $\mathsf{tor}(t)$ is almost-embeddable, while what we need is a decomposition on $\gamma(t)$ in order to apply Baker's technique.
Therefore, we first need to establish a corollary of Lemma~\ref{lem-key}, which extends a decomposition on $\beta(t)$ obtained by Lemma~\ref{lem-key} to a ``decomposition'' on $\gamma(t)$.
Consider a node $t \in T_\textnormal{RS}$ and a set $C$ of children of $t$.
Let $U_C = \bigcup_{s \in C} (\gamma(s) \backslash \sigma(s))$.
We define $G_t^C$ as the graph obtained from $G[\gamma(t) \backslash U_C]$ by making $\sigma(s)$ a clique for all $s \in C$.
In other words, $G_t^C$ is obtained from $G[\gamma(t) \backslash U_C]$ by adding, for all $s \in C$ and all $u,v \in \sigma(s)$ that are not adjacent in $G[\gamma(t) \backslash U_C]$, an edge $(u,v)$.
With this definition, we have $G_t^C = \mathsf{tor}(t)$ if $C$ consists of all children of $t$ and $G_t^C = G[\gamma(t)]$ if $C = \emptyset$.
Lemma~\ref{lem-key} implies the following corollary.

\begin{corollary} \label{cor-extention}
For any node $t \in T_\textnormal{RS}$ and $p \in [n]$, one can compute in polynomial time $Y_1,\dots,Y_p \subseteq \gamma(t) \backslash \sigma(t)$ satisfying the following conditions. \\
\textbf{(i)} For any child $s$ of $t$ and any $i \in [p]$, either $\gamma(s) \backslash \sigma (s) \subseteq Y_i$ or $\gamma(s) \cap Y_i = \emptyset$. \\
\textbf{(ii)} For any vertex $v \in \gamma(t)$, there are at most $h$ indices $i \in [p]$ such that $v \in Y_i$. \\
\textbf{(iii)} For any $i \in [p]$ and $Y' \subseteq Y_i$, $\mathbf{tw}(G_t^{C_i}/(Y_i \backslash Y')) = O(p+|Y'|)$, where $C_i$ is the set of children $s$ of $t$ satisfying $\gamma(s) \cap Y_i = \emptyset$.
\end{corollary}
\begin{proof}
Since $\mathsf{tor}(t)$ is $h$-almost-embeddable, we can use Lemma~\ref{lem-key} to compute the disjoint sets $Z_1^*,\dots,Z_p^* \subseteq \beta(t)$ 
satisfying $\mathbf{tw}(\mathsf{tor}(t)/(Z_i^* \backslash Z')) = O(p+|Z'|)$ for any $i \in [p]$ and $Z' \subseteq Z_i^*$.
Setting $Z_i = Z_i^* \backslash \sigma(t)$ for $i \in [p]$, we have $Z_i \subseteq \beta(t) \backslash \sigma(t)$ and for any $Z' \subseteq Z_i$,
\begin{equation*}
    \mathbf{tw}(\mathsf{tor}(t)/(Z_i \backslash Z')) = \mathbf{tw}(\mathsf{tor}(t)/(Z_i^* \backslash (Z' \cup \sigma(t))) = O(p+|Z' \cup \sigma(t)|) = O(p+|Z'|).
\end{equation*}
We then obtain $Y_1,\dots,Y_p$ by ``extending'' $Z_1,\dots,Z_p$ to $\gamma(t) \backslash \sigma(t)$ as follows.
We define $Y_i$ as the union of $Z_i$ and $\gamma(s) \backslash \sigma (s)$ for all children $s$ of $t$ such that $\sigma(s) \cap Z_i \neq \emptyset$.
By this construction, it is clear that $Y_1,\dots,Y_p$ satisfy condition \textbf{(i)}.
To see that condition \textbf{(ii)} is also satisfied, observe that for each child $s$ of $t$ and each $i \in [p]$, $\gamma(s) \backslash \sigma (s) \subseteq Y_i$ only if $\sigma(s) \cap Z_i \neq \emptyset$.
Since $|\sigma(s)| \leq h$ and $Z_1,\dots,Z_p$ are disjoint, $\gamma(s) \backslash \sigma (s) \subseteq Y_i$ for at most $h$ indices $i \in [p]$.
It follows that condition \textbf{(ii)} holds for any $v \in \gamma(t) \backslash \beta(t)$.
For $v \in \beta(t)$, we have $v \in Y_i$ iff $v \in Z_i$ and hence condition \textbf{(ii)} also holds.

The rest of the proof is dedicated to verifying condition \textbf{(iii)}.
This part carefully uses of the facts that $\mathbf{tw}(\mathsf{tor}(t)/(Z_i \backslash Z')) = O(p+|Z'|)$, $G[\gamma(t) \backslash \sigma(t)]$ is connected, and $\sigma(t) = N_G(\gamma(t) \backslash \sigma(t))$.
The proof is somehow tedious but technically not difficult, so the reader can feel free to skip this part.

Let $i \in [p]$ and $Y' \subseteq Y_i$.
We want to show $\mathbf{tw}(G_t^{C_i}/(Y_i \backslash Y')) = O(p+|Y'|)$.
For $u,v \in \beta(t)$, we say the pair $(u,v)$ is \textit{bad} if $u,v \in Y_i \backslash Y'$ and $u,v$ belong to different connected components of $G_t^{C_i}[Y_i \backslash Y']$, and is \textit{good} otherwise.
Note that if $u,v \in \beta(t)$ are connected by an edge in $G$, then $(u,v)$ is good because $u$ and $v$ are also connected by an edge in $G_t^{C_i}[Y_i \backslash Y']$.
Let $G'$ be the graph obtained from $G_t^{C_i}$ by adding edges $(u,v)$ for all good pair $(u,v)$ of vertices in $\beta(t)$ such that $(u,v) \in E(\mathsf{tor}(t))$.
We notice that $G_t^{C_i}/(Y_i \backslash Y')$ is a subgraph of $G'/(Y_i \backslash Y')$, because $G_t^{C_i}$ is a subgraph of $G'$ and two vertices belong to the same connected component of $G_t^{C_i}[Y_i \backslash Y']$ iff they belong to the same connected component of $G'[Y_i \backslash Y']$.
So it suffices to show $\mathbf{tw}(G'/(Y_i \backslash Y')) = O(p+|Y'|)$.
Define $Z' = (Y' \cap \beta(t)) \cup (\bigcup_{s \in C'} Z_i \cap \sigma(s))$, where $C'$ is the set of children $s$ of $t$ satisfying $(\gamma(s) \backslash \sigma(s)) \cap Y' \neq \emptyset$.
We have the following properties about $Z'$.

\medskip
\noindent
\textit{Claim~1. $Z_i \backslash Z' \subseteq Y_i \backslash Y'$.}

\medskip
\noindent
\textit{Proof of Claim~1.}
By construction, we have $Z_i \subseteq Y_i$.
Furthermore, we have $Y' \cap Z_i \subseteq Y' \cap \beta(t) \subseteq Z'$.
Therefore, $Z_i \backslash Z' \subseteq Z_i \backslash (Y' \cap Z_i) = Z_i \backslash Y' \subseteq Y_i \backslash Y'$.
\hfill $\lhd$

\medskip
\noindent
\textit{Claim~2. $\mathbf{tw}(\mathsf{tor}(t)/(Z_i \backslash Z')) = O(p+|Y'|)$.}

\medskip
\noindent
\textit{Proof of Claim~2.}
We have $|C'| \leq |Y'|$ because the sets $\gamma(s) \backslash \sigma(s)$ are disjoint for all children $s$ of $t$.
It follows that $|Z'| \leq (h+1) \cdot |Y'| = O(|Y'|)$.
Also, we have $Z' \subseteq Z_i$ since $Y' \cap \beta(t) \subseteq Y_i \cap \beta(t) = Z_i$.
Therefore, by the property of $Z_1,\dots,Z_p$, $\mathbf{tw}(\mathsf{tor}(t)/(Z_i \backslash Z')) = O(p+|Z'|) = O(p+|Y'|)$.
\hfill $\lhd$


\medskip
\noindent
\textit{Claim~3. For any $u,v \in Z_i \backslash Z'$ that $(u,v)$ is an edge of $\mathsf{tor}(t)$, $(u,v)$ is also an edge of $G'$.}

\medskip
\noindent
\textit{Proof of Claim~3.}
It suffices to show that $(u,v)$ is good.
As observed before, if $u,v$ are connected by an edge in $G$, then $(u,v)$ is good.
So assume $(u,v)$ is not an edge of $G$.
But $(u,v)$ is an edge of $\mathsf{tor}(t)$, thus $u,v \in \sigma(s)$ for some child $s$ of $t$.
It follows that $u,v \in \sigma(s) \cap Z_i$.
If $(\gamma(s) \backslash \sigma(s)) \cap Y' \neq \emptyset$, then $s \in C'$ and $\sigma(s) \cap Z_i \subseteq Z'$, which contradicts with the facts that $u,v \in \sigma(s) \cap Z_i$ and $u,v \notin Z'$.
Therefore, $(\gamma(s) \backslash \sigma(s)) \cap Y' = \emptyset$.
On the other hand, since $u,v \in \sigma(s) \cap Z_i$, we have $\sigma(s) \cap Z_i \neq \emptyset$.
By the construction of $Y_i$, this implies $\gamma(s) \backslash \sigma(s) \subseteq Y_i$.
As a result, $\gamma(s) \backslash \sigma(s) \subseteq Y_i \backslash Y'$.
By assumption, $G[\gamma(s) \backslash \sigma(s)]$ is connected and $\sigma(s) = N_G(\gamma(s) \backslash \sigma(s))$.
Thus, $u$ and $v$ belong to the same connected component of $G[(\gamma(s) \backslash \sigma(s)) \cup \{u,v\}]$.
Because $(\gamma(s) \backslash \sigma(s)) \cup \{u,v\} \subseteq Y_i \backslash Y'$, $G[(\gamma(s) \backslash \sigma(s)) \cup \{u,v\}]$ is a subgraph of $G[Y_i \backslash Y']$, which is in turn a subgraph of $G_t^{C_i}[Y_i \backslash Y']$.
This implies $u,v$ belong to the same connected component of $G_t^{C_i}[Y_i \backslash Y']$ and hence $(u,v)$ is good.
\hfill $\lhd$

\medskip
\noindent
Note that Claim~1 and Claim~3 imply that the vertices in each connected component of $\mathsf{tor}(t)[Z_i \backslash Z']$ are contained in the same connected component of $G'[Y_i \backslash Y']$.
With this in hand, we shall show that a ``large'' induced subgraph of $G'/(Y_i \backslash Y')$ is a minor of $\mathsf{tor}(t)/(Z_i \backslash Z')$.
For convenience, we write $G_1 = G'/(Y_i \backslash Y')$ and $G_2 = \mathsf{tor}(t)/(Z_i \backslash Z')$.
Let $\pi_1: V(G') \rightarrow V(G_1)$ and $\pi_2: \beta(t) \rightarrow V(G_2)$ be the natural quotient maps (which map each vertex to its corresponding vertex in the contracted graph).
Define $B = \pi_1(\beta(t))$.

\medskip
\noindent
\textit{Claim~4. $G_1[B]$ is a minor of $G_2$. In particular, $\mathbf{tw}(G_1[B]) \leq \mathbf{tw}(G_2) = O(p+|Y'|)$.}

\medskip
\noindent
\textit{Proof of Claim~4.}
Since the vertices in each connected component of $\mathsf{tor}(t)[Z_i \backslash Z']$ are contained in the same connected component of $G'[Y_i \backslash Y']$, we have $\pi_1(v) = \pi_1(v')$ for any $v,v' \in \beta(t)$ such that $\pi_2(v) = \pi_2(v')$.
Thus, there exists a unique map $f: V(G_2) \rightarrow B$ such that $f \circ \pi_2 = \pi_1|_{\beta(t)}$.
We show that \textbf{(1)} $G_2[f^{-1}(b)]$ is connected for all $b \in B$ and \textbf{(2)} for each edge $(b,b')$ of $G_1[B]$, there exist $v \in f^{-1}(b)$ and $v' \in f^{-1}(b')$ such that $(v,v')$ is an edge in $G_2$.
Note that these two properties imply that $G_1[B]$ is a minor of $G_2$.

To see \textbf{(1)}, consider a vertex $b \in B$.
Since $G_1 = G'/(Y_i \backslash Y')$, $\pi_1^{-1}(b)$ is either a single vertex in $\beta(t) \backslash (Y_i \backslash Y')$ or (the vertex set of) a connected component of $G'[Y_i \backslash Y']$.
If $\pi_1^{-1}(b)$ is a single vertex in $\beta(t) \backslash (Y_i \backslash Y')$, then $|f^{-1}(b)| = 1$ and thus $G_2[f^{-1}(b)]$ is connected.
So assume $\pi_1^{-1}(b)$ is a connected component of $G'[Y_i \backslash Y']$.
Consider two vertices $v,v' \in f^{-1}(b)$; we want to show $v$ and $v'$ are connected by a path in $G_2[f^{-1}(b)]$.
Let $u \in \pi_2^{-1}(v)$ and $u' \in \pi_2^{-1}(v')$.
Now $u,u' \in \pi_1^{-1}(b)$, so $u$ and $u'$ are connected by a path $(u,a_1,\dots,a_q,u')$ in $G'[\pi_1^{-1}(b)]$.
Define $J = \{j \in [q]: a_j \in \beta(t)\}$ and suppose $J = \{j_1,\dots,j_{q'}\}$ where $j_1<\cdots<j_{q'}$.
We claim that $(u,a_{j_1},\dots,a_{j_{q'}},u')$ is a path connecting $u$ and $u'$ in $\mathsf{tor}(t)$.
Consider $u$ and $a_{j_1}$.
Since $a_1,\dots,a_{j_1-1} \notin \beta(t)$ and $(a_1,\dots,a_{j_1-1})$ is a path in $G'$, we must have $a_1,\dots,a_{j_1-1} \in \gamma(s) \backslash \sigma(s)$ for some child $s$ of $t$.
Because $u$ is neighboring to $a_1$ and $a_{j_1}$ is neighboring to $a_{j_1-1}$ in $G'$, we have $u,a_{j_1} \in \sigma(s)$ and hence $(u,a_{j_1})$ is an edge in $\mathsf{tor}(t)$.
By the same argument, we deduce any two consecutive vertices in the sequence $(u,a_{j_1},\dots,a_{j_{q'}},u')$ are connected by an edge in $\mathsf{tor}(t)$.
Thus, $(u,a_{j_1},\dots,a_{j_{q'}},u')$ is a path in $\mathsf{tor}(t)$.
Since $a_1,\dots,a_q \in \pi_1^{-1}(b)$, we have $\pi_2(a_{j_1}),\dots,\pi_2(a_{j_{q'}}) \in f^{-1}(b)$.
As a result, $(v,\pi_2(a_{j_1}),\dots,\pi_2(a_{j_{q'}}),v')$ is a path connecting $v$ and $v'$ in $G_2[f^{-1}(b)]$.

To see \textbf{(2)}, consider an edge $(b,b')$ in $G_1[B]$.
As mentioned before, each of $\pi_1^{-1}(b)$ and $\pi_1^{-1}(b')$ is either a single vertex in $\beta(t) \backslash (Y_i \backslash Y')$ or (the vertex set of) a connected component of $G'[Y_i \backslash Y']$.
We first notice that $\pi_1^{-1}(b)$ and $\pi_1^{-1}(b')$ cannot both be connected components of $G'[Y_i \backslash Y']$, because the images of the connected components of $G'[Y_i \backslash Y']$ under $\pi_1$ form an independent set in $G_1$ but $(b,b')$ is an edge in $G_1$.
Without loss of generality, assume $\pi_1^{-1}(b)$ is a single vertex in $\beta(t) \backslash (Y_i \backslash Y')$, which implies $|f^{-1}(b)| = 1$.
Let $v$ be the only vertex in $f^{-1}(b)$.
If $\pi_1^{-1}(b')$ is also a single vertex in $\beta(t) \backslash (Y_i \backslash Y')$, then $(v,v')$ is an edge in $G_2$ where $v'$ is the only vertex in $f^{-1}(b')$.
If $\pi_1^{-1}(b')$ is a connected component of $G'[Y_i \backslash Y']$, then there must be a \textit{witness vertex} $u' \in \pi_1^{-1}(b')$ such that $(\pi_1^{-1}(b),u')$ is an edge in $G'$.
If $u' \in \beta(t)$, we are done because $(\pi_1^{-1}(b),u')$ is also an edge in $\mathsf{tor}(t)$ and thus $(v,v')$ is an edge in $G_2$ for $v' = \pi_2(u') \in f^{-1}(b)$.
So assume $u' \notin \beta(t)$.
Then $u' \in \gamma(s) \backslash \sigma(s)$ for some child $s$ of $t$.
Since $\pi_1^{-1}(b) \in \beta(t)$ and $(\pi_1^{-1}(b),u')$ is an edge in $G'$, we must have $\pi_1^{-1}(b) \in \sigma(s)$.
Furthermore, because $\pi_1^{-1}(b')$ is a connected component of $G'[Y_i \backslash Y']$ and $\pi_1^{-1}(b')$ contains at least one vertex in $\beta(t)$ (as $b' \in B$), $\pi_1^{-1}(b')$ must contain a vertex $u^* \in \sigma(s)$ in order to connecting $u'$ and the vertex in $\beta(t)$.
Observe that $(\pi_1^{-1}(b),u^*)$ is an edge in $\mathsf{tor}(t)$, since $\pi_1^{-1}(b),u^* \in \sigma(s)$.
As a result, $(v,v')$ is an edge in $G_2$ for $v' = \pi_2(u^*) \in f^{-1}(b')$.

Once we have conditions \textbf{(1)} and \textbf{(2)}, we see that $G_1[B]$ is a minor of $G_2$.
Therefore, $\mathbf{tw}(G_1[B]) \leq \mathbf{tw}(G_2)$.
By Claim~2, $\mathbf{tw}(G_2) = O(p+|Y'|)$ and thus $\mathbf{tw}(G_1[B]) = O(p+|Y'|)$.
\hfill $\lhd$

\medskip
Now we are ready to show $\mathbf{tw}(G_1) = O(p+|Y'|)$, i.e., $\mathbf{tw}(G'/(Y_i \backslash Y')) = O(p+|Y'|)$.
By Claim~4, $\mathbf{tw}(G_1[B]) = O(p+|Y'|)$.
Let $\mathcal{T}_0$ be a tree decomposition of $G_1[B]$ of width $w = O(p+|Y'|)$, and $T_0$ the underlying (rooted) tree of $\mathcal{T}_0$.
We now modify $\mathcal{T}_0$ to obtain a tree decomposition of $G_1$ as follows.
Consider the vertices in $V(G_1) \backslash B$.
For each $v \in V(G_1) \backslash B$, $\pi_1^{-1}(v)$ is either \textbf{(1)} a connected component of $G'[Y_i \backslash Y]$ that does not contain any vertex in $\beta(t)$ or \textbf{(2)} a single vertex in $Y' \backslash \beta(t)$.
We say $v$ is a type-1 vertex for case \textbf{(1)}, and a type-2 vertex for case \textbf{(2)}.
For each type-1 vertex $v \in V(G_1) \backslash B$, since $G'[\pi_1^{-1}(v)]$ is connected and $\pi_1^{-1}(v) \cap \beta(t) = \emptyset$, we have $\pi_1^{-1}(v) \subseteq \gamma(s) \backslash \sigma(s)$ for some child $s$ of $t$.
Thus, $N_{G'}(\pi_1^{-1}(v)) \cap \beta(t) \subseteq \sigma(s)$.
We now claim that $N_{G'}(\pi_1^{-1}(v)) \cap \beta(t)$ forms a clique in $G'$.
First, any vertex $u \in N_{G'}(\pi_1^{-1}(v))$ is not in $Y_i \backslash Y'$, for otherwise $u$ and $\pi_1^{-1}(v)$ belong to the same connected component of $G'[Y_i \backslash Y]$, which implies $u \in \pi_1^{-1}(v)$, contradicting with the fact $u \in N_{G'}(\pi_1^{-1}(v))$.
It follows that $(u,u')$ is good for all $u,u' \in N_{G'}(\pi_1^{-1}(v)) \cap \beta(t)$.
Second, $N_{G'}(\pi_1^{-1}(v)) \cap \beta(t)$ forms a clique in $\mathsf{tor}(t)$ because $N_{G'}(\pi_1^{-1}(v)) \cap \beta(t) \subseteq \sigma(s)$.
Therefore, $N_{G'}(\pi_1^{-1}(v)) \cap \beta(t)$ forms a clique in $G'$, by the construction of $G'$.
Note that $N_{G_1}(v) \cap B = \pi_1(N_{G'}(\pi_1^{-1}(v)) \cap \beta(t))$, because $\pi_1^{-1}(v')$ is a single vertex in $N_{G'}(\pi_1^{-1}(v))$ for all $v' \in N_{G_1}(v)$.
This implies that $N_{G_1}(v) \cap B$ forms a clique in $G_1[B]$.
Thus, there is a node $t_v \in T_0$ such that $N_{G_1}(v) \cap B \subseteq \beta(t_v)$.
We now add a new node $t_v^*$ to $T_0$ with bag $\beta(t_v^*) = (N_{G_1}(v) \cap B) \cup \{v\}$ as a child of $t_v$.
We do this for all type-1 vertices $v \in V(G_1) \backslash B$.
After this, we add all type-2 vertices in $V(G_1) \backslash B$ to the bag of every node of $T_0$ (including the newly added ones).
As one can easily verify, this results in a tree decomposition of $G_1$.
Finally, it suffices to bound the width of this tree decomposition.
Before adding the type-2 vertices to the bags, the size of each bag is at most $w$ (recall that $w$ is the width of the original $\mathcal{T}_0$).
Since the number of type-2 vertices is at most $|Y' \backslash \beta(t)|$, the width of the final tree decomposition is bounded by $w+|Y' \backslash \beta(t)| = O(p+|Y'|)$.
\end{proof}

With Corollary~\ref{cor-extention} in hand, we are ready to apply Baker's technique to compute the DP table of $t$ (provided the DP tables of the children of $t$).
Recall that computing the DP table of $t$ is equivalent to solving a set of sub-problems $\mathsf{Prob}_{\lambda,l}$ on $G[\gamma(t)]$, where $\mathsf{Prob}_{\lambda,l}$ asks if there exists a solution (contained in $\mathsf{Cand}$) of size at most $l \in [k]$ that is compatible with the configuration $\lambda$ of $\sigma(t)$.
Now we define another set of sub-problems on $G[\gamma(t)]$ as follows.
For a configuration $\lambda$ of $\sigma(t)$, a number $l \in [k]$, and a set $Y \subseteq \gamma[t]$, we define $\mathsf{Prob}_{\lambda,l,Y}$ as a sub-problem on $G[\gamma(t)]$ that asks if there exists a solution (contained in $\mathsf{Cand}$) of size at most $l$ that is compatible with the configuration $\lambda$ and in addition is \textit{disjoint from} $Y$ (for edge-deletion problems, a solution is disjoint from $Y$ if none of the edges in the solution is incident to a vertex in $Y$).
It is clear that the answer to $\mathsf{Prob}_{\lambda,l}$ is YES if the answer to $\mathsf{Prob}_{\lambda,l,Y}$ is YES for some $Y \subseteq \gamma[t]$.
Next, we are going to reduce the task of solving the sub-problems $\mathsf{Prob}_{\lambda,l}$ to solving a set of sub-problems $\mathsf{Prob}_{\lambda,l,Y}$, by applying Corollary~\ref{cor-extention}.
Set $p = \lfloor \sqrt{k} \rfloor$ and compute the sets $Y_1,\dots,Y_p \subseteq \gamma(t) \backslash \sigma(t)$ in Corollary~\ref{cor-extention}.
Then we construct a set $\varPi$ of pairs $(i,Y')$ where $i \in [p]$ and $Y' \subseteq Y_i$ satisfying three conditions: \textbf{(1)} $|\varPi| = |\mathsf{Cand}|^{O(\sqrt{k})}$, \textbf{(2)} $|Y'| = O(k/p)$ for all $(i,Y') \in \varPi$, and \textbf{(3)} the answer to $\mathsf{Prob}_{\lambda,l}$ is YES iff the answer to $\mathsf{Prob}_{\lambda,l,Y_i \backslash Y'}$ is YES for some $(i,Y') \in \varPi$.
To this end, we consider vertex-deletion and edge-deletion problems separately.

For vertex-deletion problems, $\mathsf{Cand} \subseteq V(G)$.
We simply define $\varPi$ as the set of all pairs $(i,Y')$ where $i \in [p]$ and $Y' \subseteq Y_i \cap \mathsf{Cand}$ satisfying $|Y'| \leq hk/p$.
Clearly, $|\varPi| \leq p \cdot |\mathsf{Cand}|^{hk/p} = |\mathsf{Cand}|^{O(\sqrt{k})}$.
Also, it is obvious that $\varPi$ satisfies condition \textbf{(2)}.
To see $\varPi$ satisfies condition \textbf{(3)}, it suffices to verify the ``only if'' part (as the ``if'' part is obvious).
Suppose the answer to $\mathsf{Prob}_{\lambda,l}$ is YES and consider a solution $V_\mathsf{sol} \subseteq \gamma(t) \cap \mathsf{Cand}$ to $\mathsf{Prob}_{\lambda,l}$. 
We have $|V_\mathsf{sol}| \leq l$.
Due to property \textbf{(ii)} of Corollary~\ref{cor-extention}, each vertex in $\gamma(t)$ is contained in at most $h$ $Y_i$'s.
Therefore, there exists $i \in [p]$ such that $|V_\mathsf{sol} \cap Y_i| \leq hl/p \leq hk/p$.
Setting $Y' = V_\mathsf{sol} \cap Y_i$, we now have $(i,Y') \in \varPi$ and $V_\mathsf{sol} \cap (Y_i \backslash Y') = \emptyset$.
Thus, $V_\mathsf{sol}$ is also a solution to $\mathsf{Prob}_{\lambda,l,Y_i \backslash Y'}$, which implies that the answer to $\mathsf{Prob}_{\lambda,l,Y_i \backslash Y'}$ is YES.

For edge-deletion problems, $\mathsf{Cand} \subseteq E(G)$.
Let $\mathsf{Cand}' \subseteq V(G)$ be the set of vertices of $G$ that are incident to the edges in $\mathsf{Cand}$.
We have $|\mathsf{Cand}'| \leq 2|\mathsf{Cand}|$.
We define $\varPi$ as the set of all pairs $(i,Y')$ where $i \in [p]$ and $Y' \subseteq Y_i \cap \mathsf{Cand}'$ satisfying $|Y'| \leq 2hk/p$.
Clearly, $|\varPi| \leq p \cdot |\mathsf{Cand}'|^{2hk/p} = |\mathsf{Cand}|^{O(\sqrt{k})}$ and $\varPi$ satisfies condition \textbf{(2)}.
To see $\varPi$ satisfies condition \textbf{(3)}, again it suffices to verify the ``only if'' part.
Suppose the answer to $\mathsf{Prob}_{\lambda,l}$ is YES and consider a solution $E_\mathsf{sol} \subseteq E(G[\gamma(t)]) \cap \mathsf{Cand}$ to $\mathsf{Prob}_{\lambda,l}$.
We have $|E_\mathsf{sol}| \leq l$.
Let $V_\mathsf{sol} \subseteq \gamma(t) \cap \mathsf{Cand}'$ consist of vertices incident to the edges in $E_\mathsf{sol}$; so we have $|V_\mathsf{sol}| \leq 2l$.
By the same argument as in the vertex-deletion problems, we see there exist $i \in [p]$ such that $|V_\mathsf{sol} \cap Y_i| \leq 2hl/p \leq 2hk/p$.
Setting $Y' = V_\mathsf{sol} \cap Y_i$, we now have $(i,Y') \in \varPi$ and $V_\mathsf{sol}$ is disjoint from $Y_i \backslash Y'$ (hence $E_\mathsf{sol}$ is disjoint from $Y_i \backslash Y'$).
Thus, $E_\mathsf{sol}$ is also a solution to $\mathsf{Prob}_{\lambda,l,Y_i \backslash Y'}$, which implies that the answer to $\mathsf{Prob}_{\lambda,l,Y_i \backslash Y'}$ is YES.

By the above argument, to solve a sub-problem $\mathsf{Prob}_{\lambda,l}$, it suffices to solve the sub-problems $\mathsf{Prob}_{\lambda,l,Y_i \backslash Y'}$ for all $(i,Y') \in \varPi$.
To this end, we recall the ``contraction-friendly'' feature of our problems.
Since $\mathsf{Prob}_{\lambda,l,Y_i \backslash Y'}$ looks for a solution that is disjoint from $Y_i \backslash Y'$, the ``contraction-friendly'' feature allows us to contract the connected components of $Y_i \backslash Y'$ and do DP on a tree decomposition of the resulting graph $G[\gamma(t)]/(Y_i \backslash Y')$.
If the treewidth of $G[\gamma(t)]/(Y_i \backslash Y')$ was small, we can solve $\mathsf{Prob}_{\lambda,l,Y_i \backslash Y'}$ efficiently.
Unfortunately, we are not able to bound the treewidth of $G[\gamma(t)]/(Y_i \backslash Y')$.
However, by property \textbf{(iii)} of Corollary~\ref{cor-extention}, we have $\mathbf{tw}(G_t^{C_i}/(Y_i \backslash Y')) = O(p+|Y'|) = O(\sqrt{k})$.
So the idea here is to begin with an $O(\sqrt{k})$-width tree decomposition of $G_t^{C_i}/(Y_i \backslash Y')$ and modify it to obtain a tree decomposition of $G[\gamma(t)]/(Y_i \backslash Y')$ that ``almost'' has width $O(\sqrt{k})$ and has some other good properties.
For convenience of exposition, in what follows, we do not distinguish the vertices in $\gamma(t) \backslash (Y_i \backslash Y')$ with their images in $G[\gamma(t)]/(Y_i \backslash Y')$, and similarly the vertices in $V(G_t^{C_i}) \backslash (Y_i \backslash Y')$ with their images in $G_t^{C_i}/(Y_i \backslash Y')$.
Note that $\sigma(t) \subseteq \gamma(t) \backslash (Y_i \backslash Y')$ and $\gamma(s) \subseteq \gamma(t) \backslash (Y_i \backslash Y')$ for all $s \in C_i$ (by the definition of $C_i$); so they can be viewed as sets of vertices in $G[\gamma(t)]/(Y_i \backslash Y')$.

\begin{lemma} \label{lem-specialtd}
For any $(i,Y') \in \varPi$, one can construct in $2^{O(\sqrt{k})} n^{O(1)}$ time a tree decomposition $\mathcal{T}^*$ of $G[\gamma(t)]/(Y_i \backslash Y')$ with underlying (rooted) tree $T^*$ which has the following properties: \\
\textbf{(i)} each node of $T^*$ has at most 3 children; \\
\textbf{(ii)} the root $\mathsf{rt}$ of $T^*$ has only one child $\mathsf{rt}'$, where $\sigma(\mathsf{rt}') = \sigma(t)$ and $ \beta(\mathsf{rt}) = \sigma(t)$; \\
\textbf{(iii)} for any node $t^* \in T^*$, either $\beta(t^*) = O(\sqrt{k})$ or $t^*$ is a leaf of $T^*$ satisfying $\beta(t^*) = \gamma(s)$ and $\sigma(t^*) = \sigma(s)$ for some $s \in C_i$.
\end{lemma}
\begin{proof}
Consider a pair $(i,Y') \in \varPi$.
We write $G_1 = G_t^{C_i}/(Y_i \backslash Y')$, $G_1' = G[V(G_t^{C_i})]/(Y_i \backslash Y')$, $G_2 = G[\gamma(t)]/(Y_i \backslash Y')$.
We first notice that $G_t^{C_i}[Y_i \backslash Y'] = G[Y_i \backslash Y']$, because $\sigma(s) \cap Y_i = \emptyset$ for all $s \in C_i$.
Furthermore, $G[V(G_t^{C_i})]$ is a subgraph of $G_t^{C_i}$ and the two graphs share the same vertex set.
It follows that $G_1'$ is a subgraph of $G_1$ which has the same vertex set as $G_1$.
On the other hand, we have $V(G_t^{C_i}) \subseteq \gamma(t)$, so $G_1'$ is an induced subgraph of $G_2$.
Therefore, $V(G_1)$ can be viewed as a subset of $V(G_2)$ and we have $V(G_2) = V(G_1) \cup (\bigcup_{s \in C_i} \gamma(s))$; for convenience, we do not distinguish the vertices in $V(G_1)$ with their corresponding vertices in $V(G_2)$.


Since $\mathbf{tw}(G_1) = O(p+|Y'|) = O(\sqrt{k})$ by property \textbf{(iii)} of Corollary~\ref{cor-extention}, we can use Lemma~\ref{lem-tdecomp} to compute in $2^{O(\sqrt{k})} n^{O(1)}$ time a binary tree decomposition $\mathcal{T}^*$ for $G_1$ of width $O(\sqrt{k})$.
Let $T^*$ be the underlying (rooted) tree of $\mathcal{T}^*$.
We are going to modify $\mathcal{T}^*$ to a tree decomposition of $G_2$ satisfying the three desired properties.
In the first step of modification, we add the vertices in $\sigma(t)$ to the bags of all nodes of $T^*$.
Then we create a new root $\mathsf{rt}$ for $T^*$ with bag $\beta(\mathsf{rt}) = \sigma(t)$ and let the orginal root $\mathsf{rt}'$ of $T^*$ be the only child of $\mathsf{rt}$.
Clearly, the resulting $\mathcal{T}^*$ is still a binary tree decomposition of $G_1$ of width $O(\sqrt{k})$, because $|\sigma(t)| \leq h$.
In the second step of modification, we add some new nodes (with bags) to $T^*$ as follows.
For each $s \in C_i$, $\sigma(s)$ forms a clique in $G_t^{C_i}$ and $\sigma(s) \cap Y_i = \emptyset$ by the definition of $C_i$.
Thus, $\sigma(s)$ also forms a clique in $G_1$.
So there exists a node $t^* \in T^*$ such that $\sigma(s) \subseteq \beta(t^*)$; we call $t^*$ the \textit{support node} of $s$ (if there exist multiple $t^* \in T^*$ satisfying $\sigma(s) \subseteq \beta(t^*)$, we arbitrarily pick one as the support node of $s$).
Note that different nodes in $C_i$ can have the same support node in $T^*$.
Now consider a node $t^* \in T^*$.
Suppose $t^*$ is the support node of $s_1,\dots,s_q \in C_i$.
We define a binary tree $T_q$ with bags as follows.
Define $T_0$ as the empty tree.
For $j \in [q]$, define $T_j$ as the tree consisting of a root node with bag $\bigcup_{j'=1}^j \sigma(s_{j'})$ whose left subtree is $T_{j-1}$ and right subtree is a single node with bag $\gamma(s_j)$.
Note that $T_q$ has exactly $q$ leaves in which the $j$-th leftmost leaf has bag $\gamma(s_j)$ and adhesion $\sigma(s_j)$.
Also note that the bag of any internal node of $T_q$ is a subset of $\beta(t^*)$, and thus is of size $O(\sqrt{k})$.
We then add the tree $T_q$ to $T^*$ as a subtree of $t^*$, and call $T_q$ the \textit{heavy subtree} of $t^*$.
We do this for every $t^* \in T^*$.
This completes the modification of $\mathcal{T}^*$.

We now verify that the resulting $\mathcal{T}^*$ is indeed a tree decomposition of $G_2$ with the three desired properties.
Call a node in $T^*$ \textit{new} if it is newly added to $T^*$ during our modification and \textit{old} if it is originally in $T^*$.
Properties \textbf{(i)} and \textbf{(ii)} follow directly from our construction.
To see property \textbf{(iii)}, we observe a one-to-one correspondence between the nodes in $C_i$ and the new leaves in $T^*$.
Indeed, each $s \in C_i$ corresponds to a leaf $s^*$ in the heavy subtree of its support node in $T^*$, where we have $\beta(s^*) = \gamma(s)$ and $\sigma(s^*) = \sigma(s)$.
It is clear that this correspondence is one-to-one.
Consider a node $s^* \in T^*$.
If $s^*$ is an old node or is an internal node in a heavy subtree, then $|\beta(s^*)| = O(\sqrt{k})$.
If $s^*$ is a new leaf (i.e., a leaf of a heavy subtree), then it corresponds to some $s \in C_i$ and we have $\beta(s^*) = \gamma(s)$ and $\sigma(s^*) = \sigma(s)$.
So property \textbf{(iii)} holds.
To show $\mathcal{T}^*$ is a tree decomposition of $G_2$, consider an edge $(u,v)$ of $G_2$.
If $u,v \in V(G_1)$, then $(u,v)$ is an edge of $G_1$ and thus $u,v \in \beta(t^*)$ for some old node $t^*$ in $T^*$, since the original $\mathcal{T}^*$ is a tree decomposition of $G_1$.
Otherwise, one of $u$ and $v$ must be contained in $\gamma(s) \backslash \sigma(s)$ for some $s \in C_i$, say $u \in \gamma(s) \backslash \sigma(s)$.
Then $v \in \gamma(s)$ as $\sigma(s) = N_G(\gamma(s) \backslash \sigma(s))$.
Therefore, $u,v \in \beta(s^*) = \gamma(s)$, where $s^*$ is the new leaf in $T^*$ corresponding to $s$.
Now it suffices to show that for any vertex $v \in V(G_2)$, the nodes whose bags containing $v$ form a connected subset in $T^*$.
Call a node in $T^*$ a $v$-\textit{node} if its bag contains $v$.
If $v \in \gamma(s) \backslash \sigma(s)$ for some $s \in C_i$, then the only $v$-node in $T^*$ is the new leaf corresponding to $s$.
Otherwise, we have $v \in V(G_1)$.
Since the original $\mathcal{T}^*$ is a tree decomposition of $G_1$, the old $v$-nodes form a connected subset in $T^*$ (which we denote by $X$).
If $t^*$ is an old $v$-node, then $v$ may also be contained in (the bags of the nodes of) the heavy subtree of $t^*$.
By our construction, if $v$ is contained in a heavy subtree $T_0$, then the $v$-nodes in $T_0$ form a connected subset that contains the root of $T_0$.
Therefore, the $v$-nodes in each heavy subtree form a connected subset in $T^*$ that is adjacent to $X$.
It follows that the $v$-nodes form a connected subset in $T^*$.
\end{proof}

For a pair $(i,Y') \in \varPi$, to solve the sub-problems $\mathsf{Prob}_{\lambda,l,Y_i \backslash Y'}$, we compute the tree decomposition $\mathcal{T}^*$ in the above lemma and try to do DP on $\mathcal{T}^*$.
If the width of $\mathcal{T}^*$ was exactly $O(\sqrt{k})$, we are good.
But $T^*$ can have some ``heavy'' leaves, whose bags might be large.
However, by property \textbf{(iii)} of Lemma~\ref{lem-specialtd}, each leaf $t^* \in T^*$ with $\beta(t^*)$ not bounded by $O(\sqrt{k})$ satisfies $\gamma(t^*) = \beta(t^*) = \gamma(s)$ and $\sigma(t^*) = \sigma(s)$ for some $s \in C_i$.
Therefore, the DP table of $t^*$ to be computed is nothing but the DP table of $s$ (in the DP procedure on $\mathcal{T}_\text{RS}$).
Recall that when we reach $t$ in the DP procedure on $\mathcal{T}_\text{RS}$, the DP tables of all children $s$ of $t$ are already computed.
In other words, we have the DP tables of these leaves \textit{for free}.
As such, when doing DP on $\mathcal{T}^*$, we do not need to worry about the heavy leaves. 
The property \textbf{(i)} of Lemma~\ref{lem-specialtd} allows us to do DP more conveniently on $\mathcal{T}^*$ (for some of our problems).
Finally, once the DP on $\mathcal{T}^*$ is done, the DP table of the child $\mathsf{rt}'$ of the root $\mathsf{rt}$ of $T^*$ encodes the answers to all the sub-problems $\mathsf{Prob}_{\lambda,l,Y_i \backslash Y'}$, because $\sigma(\mathsf{rt}') = \sigma(t)$ and $\gamma(\mathsf{rt}')$ consists of all vertices of $G[\gamma(t)]/(Y_i \backslash Y')$, by property \textbf{(ii)} of Lemma~\ref{lem-specialtd}.
After all pairs $(i,Y') \in \varPi$ are considered, we solve all $\mathsf{Prob}_{\lambda,l,Y_i \backslash Y'}$, which in turn solves all $\mathsf{Prob}_{\lambda,l}$ and computes the DP table of $t$.
For a specific problem, if we are able to do DP in $f(\vec{r},w) \cdot n^{O(1)}$ time on a $w$-width tree decomposition of the graph after contracting some part disjoint from the solution, then solving each $\mathsf{Prob}_{\lambda,l,Y_i \backslash Y'}$ takes $f(\vec{r},\sqrt{k}) \cdot n^{O(1)}$ time, and solving $\mathsf{Prob}_{\lambda,l,Y_i \backslash Y'}$ for all $(i,Y') \in \varPi$ takes $|\mathsf{Cand}|^{O(\sqrt{k})} \cdot f(\vec{r},\sqrt{k}) \cdot n^{O(1)}$ time since $|\varPi| = |\mathsf{Cand}|^{O(\sqrt{k})}$.
As a result, the main DP procedure on $\mathcal{T}_\text{RS}$ (i.e., the entire algorithm) takes $|\mathsf{Cand}|^{O(\sqrt{k})} \cdot f(\vec{r},\sqrt{k}) \cdot n^{O(1)}$ time.

This completes the exposition of the generic framework of our algorithms.
The only missing piece now is the DP procedures for various problems, which are standard (and mostly known in literature).
We shall give detailed discussions for each of our algorithms individually in Section~\ref{sec-app}.

\section{Proof of Lemma~\ref{lem-key}} \label{sec-keyproof}
This section is dedicated to proving Lemma~\ref{lem-key}.
In Section~\ref{sec-surface}, we first review some basic notions about surface-embedded graphs and establish two technical lemmas (Lemma~\ref{lem-degree} and Lemma~\ref{lem-twdiam2}) which will be used in the proof of Lemma~\ref{lem-key}.
These two lemmas (especially Lemma~\ref{lem-twdiam2}) might be of independent interest.
Then in Section~\ref{sec-Z1Zp}, we construct the sets $Z_1,\dots,Z_p$ in Lemma~\ref{lem-key}.
Finally, in Section~\ref{sec-boundingtw}, we prove the sets $Z_1,\dots,Z_p$ we construct satisfy the desired property.

\subsection{Results for surface-embedded graphs} \label{sec-surface}
We begin with introducing some basic notions about surface-embedded graphs.
Let $\varSigma$ be a connected orientable surface.
A $\varSigma$-embedded graph is represented as a pair $(G,\eta)$ where $G$ is the graph and $\eta$ is an embedding of $G$ to $\varSigma$.
For any subgraph $G'$ of $G$, $\eta$ induces an embedding of $G'$ to $\varSigma$; for convenience, we usually use the same notation ``$\eta$'' to denote this subgraph embedding.
A \textit{face} of $(G,\eta)$ refers to (the closure of) a connected component of $\varSigma \backslash \eta(G)$, where $\eta(G)$ is the image of $G$ on $\varSigma$ under the embedding $\eta$.
We denote by $F_\eta(G)$ the set of faces of $(G,\eta)$.
The \textit{boundary} of a face $f \in F_\eta(G)$, denoted by $\partial f$, is the subgraph of $G$ consisting of all vertices and edges that are incident to $f$ (under the embedding $\eta$).
Note that $f$ itself is a face in its boundary subgraph $(\partial f, \eta)$, i.e., $f \in F_\eta(\partial f)$.
\begin{lemma} \label{lem-degree}
Let $f \in F_\eta(G)$ be a face of a $\varSigma$-embedded graph $(G,\eta)$.
Then for any face $f' \in F_\eta(\partial f)$ of $(\partial f,\eta)$ such that $f' \neq f$, the graph $\partial f'$ has $O(g)$ connected components and the maximum degree of $\partial f'$ is $O(g)$, where $g = \mathsf{gns}(\varSigma)$.
\end{lemma}
\begin{proof}
We first figure out what the faces of $(\partial f',\eta)$ are.
As mentioned before, $f'$ itself is a face of $(\partial f',\eta)$.
In addition, since $\partial f'$ is a subgraph of $\partial f$, there must be another face $f_0 \in F_\eta(\partial f')$ such that $f \subseteq f_0$.
By definition, each edge of $\partial f'$ is incident to $f'$.
Also, each edge of $\partial f'$ is incident to $f_0$, because it is incident to $f$ in $(\partial f,\eta)$.
Note that one edge can be incident to at most two faces in a surface-embedded graph, and thus the edges of $\partial f'$ are only incident to $f'$ and $f_0$.
It follows that $f'$ and $f_0$ are the only two faces of $(\partial f',\eta)$, i.e., $F_\eta(\partial f') = \{f',f_0\}$, because any face of $(\partial f',\eta)$ must be incident to some edge of $\partial f'$.

Based on this observation, we further argue that $\partial f'$ has no vertex of degree 0 or 1.
Similarly to the edges, each vertex of $\partial f'$ is incident to $f'$ (by definition) and also incident to $f_0$ since it is incident to $f$ in $(\partial f,\eta)$.
But vertices of degree 0 or 1 can only be incident to one face in a surface-embedded graph.
Thus, every vertex of $\partial f'$ has degree at least 2.

Now we are ready to prove the lemma.
Let $\#_V$, $\#_E$, $\#_F$, $\#_C$ be the numbers of vertices, edges, faces, connected components of $(\partial f',\eta)$.
By Euler's formula, $-\#_V+\#_E-\#_F+\#_C = O(g)$, where $g$ is the genus of $\varSigma$.
We have shown that $\#_F = |F_\eta(\partial f')| = 2$.
Therefore, $(\#_E+\#_C) - \#_V = O(g)$.
Since every vertex of $\partial f'$ has degree at least 2, we have $\#_E \geq \#_V$, which implies $\#_C = O(g)$.
It suffices to show that the maximum degree of $\partial f'$ is $O(g)$.
Since $(\#_E+\#_C) - \#_V = O(g)$, we have $2 \#_E - 2 \#_V = O(g)$.
Note that $2 \#_E - 2 \#_V = \sum_{v \in V(\partial f')} (\mathsf{deg}(v)-2)$.
We have $\mathsf{deg}(v)-2 \geq 0$ for all $v \in V(\partial f')$, as each vertex of $\partial f'$ has degree at least 2.
Thus, $\mathsf{deg}(v)-2 \leq 2 \#_E - 2 \#_V = O(g)$ for all $v \in V(\partial f')$.
As a result, the maximum degree of $\partial f'$ is $O(g)$.
\end{proof}

The \textit{vertex-face incidence} (VFI) graph of $(G,\eta)$ is a bipartite graph with vertex set $V(G) \cup F_\eta(G)$ and edges connecting every pair $(v,f) \in V(G) \times F_\eta(G)$ such that $v$ is incident to $f$ (or equivalently, $v$ is a vertex in $\partial f$).
For a subset $F \subseteq F_\eta(G)$ of faces, let $F^+ \subseteq V(G) \cup F_\eta(G)$ denote the subset consisting of all $f \in F$ and all $v \in V(G)$ that are on the boundary of some face in $F$, i.e., $F^+ = (\bigcup_{f \in F} V(\partial f)) \cup F$.
We then define the VFI graph of $(G,\eta)$ \textit{restricted to} $F$ as the induced subgraph $G^*[F^+]$, where $G^*$ is the VFI graph of $(G,\eta)$.
We notice the following simple fact.
\begin{fact} \label{fact-connected}
For any $F \subseteq F_\eta(G)$, if $\bigcup_{f \in F} f$ is connected (as a subspace of $\varSigma$), then the VFI graph of $(G,\eta)$ restricted to $F$ is connected.
In particular, the VFI graph of $(G,\eta)$ is connected.
\end{fact}
\begin{proof}
Let $F \subseteq F_\eta(G)$ such that $\bigcup_{f \in F} f$ is connected.
We denote by $G_F^*$ the VFI graph of $(G,\eta)$ restricted to $F$.
First, we notice that any two faces $f,f' \in F$ satisfying $f \cap f' \neq \emptyset$ are in the same connected component of $G_F^*$.
Indeed, if $f \cap f' \neq \emptyset$, then $V(\partial f) \cap V(\partial f') \neq \emptyset$ and thus $f \rightarrow v \rightarrow f'$ is a path from $f$ to $f'$ in $G_F^*$ for any $v \in V(\partial f) \cap V(\partial f')$.
It follows that all $f \in F$ belong to the same connected component of $G_F^*$, because $\bigcup_{f \in F} f$ is connected.
Note that $F$ is a dominating set of $G_F^*$ by the definition of $G_F^*$.
Thus, $G_F^*$ is connected.
Finally, since the VFI graph of $(G,\eta)$ is nothing but $G_{F_\eta(G)}^*$ and $\bigcup_{f \in F_\eta(G)} f = \varSigma$ is connected, the VFI graph of $(G,\eta)$ is connected.
\end{proof}

Let $G^*$ be the VFI graph of $(G,\eta)$.
For two vertices $v,v' \in V(G)$, the \textit{vertex-face distance} between $v$ and $v'$ in $(G,\eta)$ is defined as the shortest-path distance between $v$ and $v'$ in $G^*$.
In the same way, we can define the \textit{vertex-face distance} between a vertex and a face of $(G,\eta)$, or between two faces of $(G,\eta)$.
The \textit{vertex-face diameter} of $(G,\eta)$, denoted by $\mathsf{diam}^*(G,\eta)$, is the maximum vertex-face distance between vertices/faces of $(G,\eta)$ (or equivalently, the graph diameter of $G^*$).
The following lemma follows easily from the work \cite{eppstein2000diameter} of Eppstein.
\begin{lemma} \label{lem-twdiam}
Let $(G,\eta)$ be a $\varSigma$-embedded graph where the genus of the surface $\varSigma$ is $O(1)$.
Then we have $\mathbf{tw}(G) = O(\mathsf{diam}^*(G,\eta))$.
\end{lemma}
\begin{proof}
We obtain another graph from $G$ by adding new vertices and edges as follows.
For each face $f \in F_\eta(G)$, we add a new vertex $v_f$ and edges connecting $v_f$ with all vertices in $\partial f$.
Let $G'$ be the resulting graph.
We notice that the embedding $\eta$ of $G$ can be extended to an embedding of $G'$ to $\varSigma$.
Indeed, we can embed each $v_f$ in the interior of the face $f$ and draw the edges connecting $v_f$ to the vertices in $\partial f$ inside $f$.
Thus, $G'$ is a bounded-genus graph.
Furthermore, the VFI graph $G^*$ of $G$ can be viewed as a subgraph of $G'$ which has the same vertex set as $G'$, if we identify each $f \in F_\eta(G)$ with the vertex $v_f$ of $G'$.
As such, the diameter of $G'$ is smaller than or equal to the diameter of $G^*$, where the latter is $\mathsf{diam}^*(G,\eta)$.
By the work \cite{eppstein2000diameter} of Eppstein, the treewidth of a bounded-genus graph is linear in its diameter.
So we have $\mathbf{tw}(G) \leq \mathbf{tw}(G') = O(\mathsf{diam}^*(G,\eta))$.
\end{proof}

The notion of vertex-face distance/diameter can be generalized to the situation where the faces of $(G,\eta)$ are weighted.
Let $w: F_\eta(G) \rightarrow \mathbb{N}$ be a weight function on the faces of $(G,\eta)$.
Recall that $G^*$ is the VFI graph of $(G,\eta)$.
Consider a simple path $\pi = (a_0,a_1,\dots,a_m)$ in $G^*$, and write $F_\pi = F_\eta(G) \cap \{a_0,a_1,\dots,a_m\}$.
The \textit{cost} of $\pi$ under the weight function $w$ is defined as $m+\sum_{f \in F_\pi} w(f)$, i.e., the length of $\pi$ plus the total weights of the faces that $\pi$ goes through.
We define the \textit{$w$-weighted vertex-face distance} between vertices/faces $a,a' \in V(G) \cup F_\eta(G)$ in $(G,\eta)$ as the minimum cost of a path connecting $a$ and $a'$ in $G^*$ under the weight function $w$.
Note that $w$-weighted vertex-face distances satisfy the triangle inequality, although they do not necessarily form a metric because the $w$-weighted vertex-face distance from a face $f \in F_\eta(G)$ to itself is $w(f)$ rather than 0.
The \textit{$w$-weighted vertex-face diameter} of $(G,\eta)$, denoted by $\mathsf{diam}_w^*(G,\eta)$, is the maximum $w$-weighted vertex-face distance between vertices/faces of $(G,\eta)$.
Clearly, when $w$ is the zero function, the $w$-weighted vertex-face distance/diameter coincides with the ``unweighted'' vertex-face distance/diameter defined before.
We prove the following important lemma, which can be viewed as a generalization of Lemma~\ref{lem-twdiam}.
\begin{lemma} \label{lem-twdiam2}
Let $(G,\eta)$ be a $\varSigma$-embedded graph where the genus of the surface $\varSigma$ is $O(1)$ and $\kappa: F_\eta(G) \rightarrow 2^{V(G)}$ be a map satisfying $\kappa(f) \subseteq V(\partial f)$.
Then we have $\mathbf{tw}(G^\kappa) = O(\mathsf{diam}_{w_\kappa}^*(G,\eta))$, where $G^\kappa$ is the graph obtained from $G$ by making $\kappa(f)$ a clique for all $f \in F_\eta(G)$ and $w_\kappa: F_\eta(G) \rightarrow \mathbb{N}$ is a weighted function on the faces of $(G,\eta)$ defined as $w_\kappa(f) = |\kappa(f)|$.
\end{lemma}
\begin{proof}
We first construct another graph $G'$ from $G$ by adding, for each face $f \in F_\eta(G)$, a new vertex $v_f$ and edges connecting $v_f$ with all vertices in $\kappa(f)$.
Also, we extend the embedding $\eta$ of $G$ to an embedding of $G'$ to $\varSigma$ by drawing $v_f$ in the interior of the face $f$ and drawing the edges connecting $v_f$ to the vertices in $\kappa(f)$ inside $f$.
For convenience, we still use the notation $\eta$ to denote the embedding of $G'$ to $\varSigma$.
Note that the degree of each new vertex $v_f$ is $w_\kappa(f)$.
Next, we replace each $v_f$ with a grid as follows.
Consider a vertex $v_f$.
Let $e_1,\dots,e_{w_\kappa(f)}$ be the edges adjacent to $v_f$ sorted in clockwise order around $v_f$ under the embedding $\eta$.
Now we replace $v_f$ with a $w_\kappa(f) \times w_\kappa(f)$ grid $\varGamma_f$ and let the edge $e_i$ connect to the $i$-th grid vertex in the top row of $\varGamma_f$ for $i \in [w_\kappa(f)]$.
We do this for every $v_f$ and let $G''$ be the resulting graph.
The way we replace $v_f$ with $\varGamma_f$ allows us to modify the embedding of $v_f$ and $e_1,\dots,e_{w_\kappa(f)}$ to an embedding of $\varGamma_f$ and the new $e_1,\dots,e_{w_\kappa(f)}$ inside $f$.
Specifically, we take a small neighborhood $X$ of $v_f$ in $\varSigma$ that is homeomorphism to a disk.
We ``clear'' the drawing of $v_f$ and $e_1,\dots,e_{w_\kappa(f)}$ in the interior of $X$, and then draw $\varGamma_f$ in the interior of $X$ with curves connecting the $i$-th grid vertex in the top row of $\varGamma_f$ with the intersection point of (the image of) $e_i$ and the boundary of $X$.
In other words, we can extend the embedding $\eta$ of $G$ to an embedding of $G''$ to $\varSigma$, which we still denote by $\eta$ for convenience.

\medskip
\noindent
\textit{Claim~1. $\mathsf{diam}^*(G'',\eta) = O(\mathsf{diam}_{w_\kappa}^*(G,\eta))$.}

\medskip
\noindent
\textit{Proof of Claim~1.}
Since the VFI graph of $(G'',\eta)$ is connected, it suffices to show that the vertex-face distance between any two vertices $u,v \in V(G'')$ in $(G'',\eta)$ is $O(\mathsf{diam}_{w_\kappa}^*(G,\eta))$.
There are two possibilities for the vertex $u$ (resp., $v$): it is either a vertex in $G$ or a vertex in $\varGamma_f$ for some $f \in F_\eta(G)$.
We first consider the case where both $u$ and $v$ are vertices in $G$ (the other cases are similar).
Let $(u,f_0,v_1,f_1\cdots,v_q,f_q,v)$ be the path in the VFI of $(G,\eta)$ with the minimum cost under the weight function $w_\kappa$, where $f_0,f_1\dots,f_{q} \in F_\eta(G)$ and $v_1,\dots,v_q \in V(G)$.
So we have $2q+2+\sum_{i=0}^q w_\kappa(f_i) \leq \mathsf{diam}_{w_\kappa}^*(G,\eta)$.
For convenience, we write $v_0 = u$ and $v_{q+1} = v$.
We claim that the vertex-face distance between $v_{i-1}$ and $v_{i}$ in $(G'',\eta)$ is at most $2w_\kappa(f_{i-1}) + 2$.
Since $v_{i-1}$ (resp., $v_i$) is incident to $f_{i-1}$ in $(G,\eta)$ and $\eta$ embeds the grid $\varGamma_{f_{i-1}}$ inside $f_{i-1}$, there exists a grid vertex $v_{i-1}'$ (resp., $v_i'$) in $\varGamma_{f_{i-1}}$ whose vertex-face distance to $v_{i-1}$ (resp., $v_i$) in $(G'',\eta)$ is 2.
Note that the vertex-face distance between any two grid vertices in $\varGamma_{f_{i-1}}$ is at most $2w_\kappa(f_{i-1})-2$.
Thus, the vertex-face distance between $v_{i-1}$ and $v_{i}$ in $(G'',\eta)$ is at most $2w_\kappa(f_{i-1}) + 2$.
This further implies the the vertex-face distance between $u$ and $v$ in $(G'',\eta)$ is $O(q+\sum_{i=0}^q w_\kappa(f_i))$, which is $O(\mathsf{diam}_{w_\kappa}^*(G,\eta))$.
The case where $u$ and/or $v$ is a vertex in some grid $\varGamma_f$ can be handled similarly.
Specifically, if $u$ is in $\varGamma_f$ and $v$ is in $G$ (resp., a vertex in $\varGamma_{f'}$), then we can show that the vertex-face distance between $u$ and $v$ in $(G'',\eta)$ is linear in the $w_\kappa$-weighted vertex-face distance from $f$ to $v$ (resp., $f'$) in $(G,\eta)$.
As a result, $\mathsf{diam}^*(G'',\eta) = O(\mathsf{diam}_{w_\kappa}^*(G,\eta))$.
\hfill $\lhd$

\medskip
By Claim~1 and Lemma~\ref{lem-twdiam}, we have $\mathbf{tw}(G'') = O(\mathsf{diam}_{w_\kappa}^*(G,\eta))$.
Let $\mathcal{T}$ be a tree decomposition of $G''$ of width $w = O(\mathsf{diam}_{w_\kappa}^*(G,\eta))$, and $T$ be the underlying tree of $\mathcal{T}$.
Our next plan is to modify $\mathcal{T}$ to obtain an $O(w)$-width tree decomposition of a graph that contains $G^\kappa$ as a minor.
To this end, we first introduce the notion of \textit{central} nodes.
Let $t \in T$ be a node.
Then $T - \{t\}$ is a forest.
We say a tree $T'$ in the forest $T - \{t\}$ \textit{contains} a vertex $v \in V(G'')$ if $v \in \beta(t')$ for some $t' \in T'$.
Note that every vertex in $V(G'') \backslash \beta(t)$ is contained in exactly one tree in the forest $T - \{t\}$, because $\mathcal{T}$ is a tree decomposition.
For a subset $V \subseteq V(G'')$ of vertices in $G''$, we denote by $\mu(t,V)$ the maximum number of vertices in $V \backslash \beta(t)$ a tree in $T - \{t\}$ contains.
We say $t$ is \textit{central} for $V$ if $\mu(t,V) \leq |V|/2$, i.e., each tree in the forest $T - \{t\}$ contains at most $|V|/2$ vertices in $V \backslash \beta(t)$.

\medskip
\noindent
\textit{Claim~2. For any $V \subseteq V(G'')$, there exists a node in $T$ that is central for $V$.}

\medskip
\noindent
\textit{Proof of Claim~2.}
We begin from an arbitrary node $t \in T$ and keep walking in $T$ until we reach a central node for $V$ in the following way.
If $t$ is a central node for $V$, we are done.
Otherwise, there exists a (unique) tree $T'$ in the forest $T - \{t\}$ which contains more than $|V|/2$ vertices in $V \backslash \beta(t)$; let $t' \in T'$ be the node adjacent to $t$.
We then go from $t$ to $t'$.
To see this walk always terminates, it suffices to show that when we go from $t$ to $t'$, we cannot go back to $t$ in the next step.
Indeed, if we go from $t$ to $t'$, then the tree in $T - \{t\}$ containing $t'$ contains more than $|V|/2$ vertices in $V \backslash \beta(t)$, which implies the tree in $T - \{t'\}$ containing $t$ contains less than $|V|/2$ vertices in $V \backslash \beta(t)$.
So we will not go back to $t$ from $t'$ in the next step.
\hfill $\lhd$

\medskip
Let $\hat{\kappa}(f) \subseteq V(\varGamma_f)$ consist of the vertices in the last row of the grid $\varGamma_f$.
By Claim~2, for each $f \in F_\eta(G)$, there exists a central node $t_f \in T$ for $\hat{\kappa}(f)$.
We ``drag'' the vertices in $\hat{\kappa}(f) \backslash \beta(t_f)$ to the node $t_f$ for each $f \in F_\eta(G)$ as follows.
Consider a face $f \in F_\eta(G)$ and a node $t \in T$.
We define a set $\alpha_f(t) \subseteq \hat{\kappa}(f) \backslash \beta(t)$ as follows.
As argued before, every vertex in $\hat{\kappa}(f) \backslash \beta(t)$ is contained in exactly one tree in the forest $T - \{t\}$.
Let $\alpha_f(t)$ consist of all $v \in \hat{\kappa}(f) \backslash \beta(t)$ such that $t_f \notin T_v$ where $T_v$ is the tree in $T - \{t\}$ containing $v$.
We then add all vertices in $\bigcup_{f \in F_\eta(G)} \alpha_f(t)$ to the bag $\beta(t)$ for all $t \in T$.
We show that after this modification, $\mathcal{T}$ is still a tree decomposition of $G''$ of width $O(w)$ satisfying an additional condition: for any $f \in F_\eta(G)$, there is a node whose bag contains $\hat{\kappa}(f)$.
For convenience of exposition, for each node $t \in T$, we write $\beta^*(t)$ as the bag of $t$ after the modification and $\beta(t)$ as the bag of $t$ before the modification, i.e., $\beta^*(t) = \beta(t) \cup (\bigcup_{f \in F_\eta(G)} \alpha_f(t))$.
It is easy to verify that for every $f \in F_\eta(G)$ and every $v \in \hat{\kappa}(f) \backslash \beta(t_f)$, the nodes $t \in T$ such that $v \in \alpha_f(t)$ form a path from $t_f$ to a node adjacent to the connected set $\{t \in T: v \in \beta(t)\}$.
Thus, for every vertex $v \in V(G'')$, the nodes $t \in T$ such that $v \in \beta^*(t)$ form a connected set in $T$.
This implies that $\mathcal{T}$ is a tree decomposition of $G''$ after the modification, because we only add vertices to the bags in the modification.
Also, it is clear that $\mathcal{T}$ satisfies the additional condition, since $\hat{\kappa}(f) \backslash \beta(t_f) \subseteq \alpha_f(t_f)$ and hence $\hat{\kappa}(f) \subseteq \beta^*(t_f)$.
Therefore, it suffices to show that the width of $\mathcal{T}$ is $O(w)$, i.e., $|\beta^*(t)| = O(w)$ for all $t \in T$.
Consider a node $t \in T$.
Recall that $\mu(t,V)$ denotes the maximum number of vertices in $V \backslash \beta(t)$ a tree in $T - \{t\}$ can contain.
The key observation is the following.

\medskip
\noindent
\textit{Claim~3. $w_\kappa(f) - \mu(t,\hat{\kappa}(f)) = O(|\beta(t) \cap V(\varGamma_f)|)$ for any $f \in F_\eta(G)$.}

\medskip
\noindent
\textit{Proof of Claim~3.}
Suppose the forest $T - \{t\}$ consists of $q$ trees $T_1,\dots,T_q$.
As argued before, each vertex in $\hat{\kappa}(f) \backslash \beta(t)$ is contained in exactly one of $T_1,\dots,T_q$.
We say two vertices $v,v' \in \hat{\kappa}(f) \backslash \beta(t)$ are \textit{separated} if $v$ and $v'$ are contained in different trees in $\{T_1,\dots,T_q\}$, and two sets $K,K' \subseteq \hat{\kappa}(f)$ are \textit{separated} if $v$ and $v'$ are separated for any $v \in K \backslash \beta(t)$ and $v' \in K' \backslash \beta(t)$.
We first show that one can partition $\hat{\kappa}(f)$ into two separated subsets $K$ and $K'$ such that both of them are of size at least $(w_\kappa(f) - \mu(t,\hat{\kappa}(f)))/2$.
If $|\hat{\kappa}(f) \cap \beta(t)| \geq (w_\kappa(f) - \mu(t,\hat{\kappa}(f)))/2$, then we simply let $K$ be a subset of $\hat{\kappa}(f) \cap \beta(t)$ of size $(w_\kappa(f) - \mu(t,\hat{\kappa}(f)))/2$ and $K' = \hat{\kappa}(f) \backslash K$.
We have $|K'| = w_\kappa(f) - |K| = (w_\kappa(f) + \mu(t,\hat{\kappa}(f)))/2$ and it is clear that $K$ and $K'$ are separated as $K \backslash \beta(t) = \emptyset$.
If $|\hat{\kappa}(f) \cap \beta(t)| < (w_\kappa(f) - \mu(t,\hat{\kappa}(f)))/2$, then $\hat{\kappa}(f) \backslash \beta(t) \geq (w_\kappa(f) - \mu(t,\hat{\kappa}(f)))/2$.
Let $i \in [q]$ be the smallest index such that $T_1,\dots,T_i$ contains at least $(w_\kappa(f) - \mu(t,\hat{\kappa}(f)))/2$ vertices in $\hat{\kappa}(f) \backslash \beta(t)$.
Then we define $K$ as the set of vertices in $\hat{\kappa}(f) \backslash \beta(t)$ that are contained in $T_1,\dots,T_i$ and define $K' = \hat{\kappa}(f) \backslash K$.
Clearly, $K$ and $K'$ are separated and $|K| \geq (w_\kappa(f) - \mu(t,\hat{\kappa}(f)))/2$.
To see $|K'| \geq (w_\kappa(f) - \mu(t,\hat{\kappa}(f)))/2$, observe that $T_1,\dots,T_{i-1}$ contains less than $(w_\kappa(f) - \mu(t,\hat{\kappa}(f)))/2$ vertices in $\hat{\kappa}(f) \backslash \beta(t)$.
Since $T_i$ can contain at most $\mu(t,\hat{\kappa}(f))$ vertices in $\hat{\kappa}(f) \backslash \beta(t)$, we have $|K| \leq (w_\kappa(f) + \mu(t,\hat{\kappa}(f)))/2$ and thus $|K'| = w_\kappa(f) - |K| \geq (w_\kappa(f) - \mu(t,\hat{\kappa}(f)))/2$.

Next, we show that $\min\{|K|,|K'|\} \leq |\beta(t) \cap V(\varGamma_f)|$, which directly proves the claim because we have $\min\{|K|,|K'|\} \geq (w_\kappa(f) - \mu(t,\hat{\kappa}(f)))/2$.
We first observe that for any path $\pi$ in $G''$ connecting a vertex $v \in K$ and a vertex $v' \in K'$, $\beta(t)$ must contain at least one vertex on $\pi$.
Indeed, because $\mathcal{T}$ is a tree decomposition of $G''$ (before the modification) and $\pi$ is a connected subgraph of $G''$, the nodes whose bags contain at least one vertex on $\pi$ form a connected subset in $T$.
This connected subset must contain $t$, since either $\{v,v'\} \cap \beta(t) \neq \emptyset$ or $v$ and $v'$ are vertices in $\hat{\kappa}(f) \backslash \beta(t)$ that are contained in different trees in $T - \{t\}$ (as $K$ and $K'$ are separated).
With this observation, we are ready to show $\min\{|K|,|K'|\} \leq |\beta(t) \cap V(\varGamma_f)|$.
We have $\min\{|K|,|K'|\} \leq |\hat{\kappa}(f) \backslash \beta(t)| \leq w_\kappa(f)$.
If $|\beta(t) \cap V(\varGamma_f)| \geq w_\kappa(f)$, we are done.
So assume $|\beta(t) \cap V(\varGamma_f)| < w_\kappa(f)$.
We say $\beta(t)$ \textit{hits} a row/column of the grid $\varGamma_f$ if there is a vertex in the row/column that is contained in $\beta(t)$.
We claim that either all the columns of $\varGamma_f$ containing the vertices in $K$ are hit by $\beta(t)$ or all the columns of $\varGamma_f$ containing the vertices in $K'$ are hit by $\beta(t)$, which implies $|\beta(t) \cap V(\varGamma_f)| \geq \min\{|K|,|K'|\}$.
Suppose this is not true, then there exist $v \in K$ and $v' \in K'$ such that the columns $C$ and $C'$ containing $v$ and $v'$ are not hit by $\beta(t)$.
Since $\varGamma_f$ has $w_\kappa(f)$ rows and $|\beta(t) \cap V(\varGamma_f)| < w_\kappa(f)$, there is at least one row $R$ of $\varGamma_f$ that is not hit by $\beta(t)$.
Consider the path $\pi$ in $\varGamma_f$ connecting $v$ and $v'$ which goes through $C$, $R$, and $C'$.
Since $C,R,C'$ are not hit by $\beta(t)$, none of the vertices on $\pi$ is contained in $\beta(t)$.
This contradicts with the fact that $\beta(t)$ contains at least one vertex on each path in $G''$ connecting a vertex in $K$ and a vertex in $K'$.
Therefore, $\min\{|K|,|K'|\} \leq |\beta(t) \cap V(\varGamma_f)|$.
\hfill $\lhd$

\medskip
With Claim~3 in hand, we now show $|\beta^*(t)| = O(|\beta(t)|) = O(w)$.
We shall show that $|\alpha_f(t)| = O(|\beta(t) \cap V(\varGamma_f)|)$ for all $f \in F_\eta(G)$, which implies $|\beta^*(t)| = O(|\beta(t)|)$ because $|\beta^*(t)| \leq |\beta(t)| + \sum_{f \in F_\eta(G)} |\alpha_f(t)|$ and $\sum_{f \in F_\eta(G)} |\beta(t) \cap V(\varGamma_f)| \leq |\beta(t)|$.
If $\mu(t,\hat{\kappa}(f)) \leq w_\kappa(f)/2$, then by Claim~3 we have $w_\kappa(f) = O(|\beta(t) \cap V(\varGamma_f)|)$ and thus $|\alpha_f(t)| \leq w_\kappa(f) = O(|\beta(t) \cap V(\varGamma_f)|)$.
So assume $\mu(t,\hat{\kappa}(f)) > w_\kappa(f)/2$ and let $T'$ be the tree in the forest $T - \{t\}$ that contains $\mu(t,\hat{\kappa}(f))$ vertices in $\hat{\kappa}(f) \backslash \beta(t)$.
Observe that $t_f \in T'$ for otherwise $\mu(t_f,\hat{\kappa}(f)) \geq \mu(t,\hat{\kappa}(f)) > w_\kappa(f)/2$.
Therefore, $\alpha_f(t)$ is the set of vertices in $\hat{\kappa}(f) \backslash \beta(t)$ that are not contained in $T'$.
So we have
\begin{equation*}
    |\alpha_f(t)| = |\hat{\kappa}(f) \backslash \beta(t)| - \mu(t,\hat{\kappa}(f)) \leq w_\kappa(f) - \mu(t,\hat{\kappa}(f)) = O(|\beta(t) \cap V(\varGamma_f)|),
\end{equation*}
where the last equality follows from Claim~3.

Now we have seen that after the modification, $\mathcal{T}$ is a tree decomposition of $G''$ of width $O(w)$ satisfying that for any $f \in F_\eta(G)$, there is a node whose bag contains $\hat{\kappa}(f)$.
Equivalently, $\mathcal{T}$ is an $O(w)$-width tree decomposition of the graph $G^*$ obtained from $G''$ by making $\hat{\kappa}(f)$ a clique for all $f \in F_\eta(G)$.
Note that $G^*$ contains $G^\kappa$ as a minor.
Indeed, if we contract the edges on each column of the grid $\varGamma_f$ and the edges connecting the vertices in $\kappa(f)$ with the vertices in the first row of $\varGamma_f$ for all $f \in F_\eta(G)$, then we obtain $G^\kappa$ from $G^*$.
As a result, $\mathbf{tw}(G^\kappa) = O(w) = O(\mathsf{diam}_{w_\kappa}^*(G,\eta))$, which completes the proof.
\end{proof}


\subsection{Constructing the sets $Z_1,\dots,Z_p$} \label{sec-Z1Zp}
Now we begin the proof of Lemma~\ref{lem-key}.
Let us recall the formal definition of almost-embeddable graphs.
Since this definition is somehow involved, we first need to introduce some basic notions.
A \textit{facial disk} of a $\varSigma$-embedded graph $(G,\eta)$ is a (topological) disk in $\varSigma$ whose interior is disjoint from the image of $(G,\eta)$ in $\varSigma$.
Clearly, a facial disk of $(G,\eta)$ is contained in some face of $(G,\eta)$.
For a disk $D$ in $\varSigma$, we define $V(G) \cap_\eta D \subseteq V(G)$ as the set of vertices whose images under $\eta$ lie in $D$.
If $D$ is a facial disk of $(G,\eta)$, then the images of the vertices in $V(G) \cap_\eta D$ under $\eta$ are all on the boundary of $D$; in this case, the vertices in $V(G) \cap_\eta D$ can be sorted in clockwise or counterclockwise order along the boundary of $D$.
A \textit{path decomposition} is a tree decomposition in which the underlying tree is a path.
\begin{definition}[almost-embeddable graphs] \label{def-almost}
For an integer $h > 0$, a graph $G$ is $h$-\textbf{almost-embeddable} if there exists $A \subseteq V(G)$ with $|A| \leq h$ such that we can write $G-A = G_0 \cup G_1 \cup \cdots \cup G_h$ and there exist an embedding $\eta$ of $G_0$ in a surface $\varSigma$ of genus $h$, $d$ disjoint facial disks $D_1,\dots,D_h$ of $(G_0,\eta)$, and $h$ pairs $(\tau_1,\mathcal{P}_1),\dots,(\tau_h,\mathcal{P}_h)$ such that the following conditions hold for all $i \in [h]$.
\begin{itemize}
    \item $G_1,\dots,G_h$ are mutually disjoint. 
    \item $V(G_0) \cap_\eta D_i = V(G_0) \cap V(G_i)$.
    Set $q_i = |V(G_0) \cap_\eta D_i| = |V(G_0) \cap V(G_i)|$.
    \item $\tau_i = (v_{i,1},\dots,v_{i,q_i})$ is a permutation of the vertices in $V(G_0) \cap_\eta D_i$ that is compatible with the clockwise or counterclockwise order along the boundary of $D_i$;
    \item $\mathcal{P}_i$ is a path decomposition of $G_i$ of width at most $h$ whose underlying path $\pi_i = (u_{i,1},\dots,u_{i,q_i})$ is of length $q_i - 1$ and satisfies $v_{i,j} \in \beta(u_{i,j})$ for all $j \in [q_i]$.
\end{itemize}
We call $A$ the \textbf{apex set}, $G_0$ the \textbf{embeddable skeleton}, and $G_1,\dots,G_h$ the \textbf{vortices} attached to the disks $D_1,\dots,D_h$.
Also, we call each pair $(\tau_i,\mathcal{P}_i)$ the \textbf{witness pair} of the vortex $G_i$, for $i \in [h]$.
These components together are called the \textbf{almost-embeddable structure} of $G$.
\end{definition}

Let $G$ be an $h$-almost-embeddable graph with apex set $A \subseteq V(G)$, embeddable skeleton $G_0$ with an embedding $\eta$ to a genus-$h$ surface $\varSigma$, and vortices $G_1,\dots,G_h$ attached to disjoint facial disks $D_1,\dots,D_h$ in $(G_0,\eta)$ with witness pairs $(\tau_1,\mathcal{P}_1),\dots,(\tau_h,\mathcal{P}_h)$.
Let $p \in \mathbb{N}$ be the given number.
In this section, we construct the disjoint sets $Z_1,\dots,Z_p \subseteq V(G)$ in Lemma~\ref{lem-key}, and then in the next section we show these sets satisfy the desired property.

In our construction, the sets $Z_1,\dots,Z_p$ will all be subsets of $V(G_0)$.
In the first step, we add some ``virtual'' edges to $G_0$.
Consider an index $i \in [h]$.
Suppose $\tau_i = (v_{i,1},\dots,v_{i,q_i})$.
By definition, $v_{i,1},\dots,v_{i,q_i}$ are the vertices of $(G_0,\eta)$ that lie on the boundary of $D_i$, sorted in clockwise or counterclockwise order.
For convenience, we write $v_{i,0} = v_{i,q_i}$.
We then add the edges $(v_{i,j-1},v_{i,j})$ for all $j \in [q_i]$ to $G_0$, and call them \textit{virtual} edges.
Furthermore, we draw these virtual edges inside (the interior of) the disk $D_i$ without crossing (this is possible because $v_1,\dots,v_{q_i}$ are sorted along the boundary of $D_i$).
The images of these virtual edges then enclose a disk $D_i'$ in $D_i$.
We do this for all indices $i \in [h]$.
Let $G_0'$ denote the resulting graph after adding the virtual edges.
Since $D_1,\dots,D_h$ are disjoint facial disks in $(G_0,\eta)$, the images of the virtual edges do not cross each other or cross the original edges in $(G_0,\eta)$.
Therefore, the drawing of the virtual edges extends $\eta$ to an embedding of $G_0'$ to $\varSigma$; for simplicity, we still use the notation $\eta$ to denote this embedding.
Note that the disks $D_1',\dots,D_h'$ are faces of $(G_0',\eta)$, which we call \textit{vortex faces}.

Now fix an arbitrary point $x_\mathsf{out} \in \varSigma$ that is disjoint from the image of $G_0'$ under $\eta$, or equivalently, in the interior of a face of $(G_0',\eta)$.
For any $\varSigma$-embedded graph whose image is disjoint from $x_\mathsf{out}$, we call the face containing $x_\mathsf{out}$ the \textit{outer} face of the graph.
Let $o \in F_\eta(G_0')$ be the outer face of $(G_0',\eta)$.
For any vertex $v \in V(G_0') = V(G_0)$, the vertex-face distance between $o$ and $v$ in $(G_0',\eta)$ is an odd number (since $o$ is a face and $v$ is a vertex).
We define $L_i \subseteq V(G_0')$ as the subset consisting of all vertices $v \in V(G_0')$ such that the vertex-face distance between $o$ and $v$ in $(G_0',\eta)$ is $2i-1$.
Let $m \in \mathbb{N}$ be the largest number such that $L_m \neq \emptyset$.
By Fact~\ref{fact-connected}, the VFI graph of $(G_0',\eta)$ is connected, and hence $L_1,\dots,L_m$ forms a partition of $V(G_0')$.
We call $L_i$ the \textit{$i$-th layer} of $(G_0',\eta)$.
For convenience, we write $L_{i^-}^{i+} = \bigcup_{i=i^-}^{i^+} L_i$ for two indices $i^-,i^+ \in [m]$ that $i^- \leq i^+$, and write $L_{\geq i} = L_i^m$ and $L_{> i} = L_{i+1}^m$ for $i \in [m]$.
Define $\ell: V(G_0') \rightarrow [m]$ as the function which maps each $v \in V(G_0')$ to the unique index $i \in [m]$ such that $v \in L_i$.
One can easily see the following properties of $L_1,\dots,L_m$.

\begin{fact} \label{fact-diff1}
If two vertices $u,v \in V(G_0')$ are incident to a common face of $(G_0',\eta)$, $|\ell(u) - \ell(v)| \leq 1$.
Also, if two vertices $u,v \in V(G_0')$ are neighboring in $G_0'$, $|\ell(u) - \ell(v)| \leq 1$.
\end{fact}
\begin{proof}
Let $u,v \in V(G_0')$.
If $u$ and $v$ are incident to a common face of $(G_0',\eta)$, then the vertex-face distance between $u$ and $v$ in $(G_0',\eta)$ is $2$.
Note that the vertex-face distance between $o$ and $u$ (resp., $v$) in $(G_0',\eta)$ is $2\ell(u)-1$ (resp., $2\ell(v)-1$).
By the triangle inequality, we have $|(2\ell(u)-1) - (2\ell(v)-1)| \leq 2$, which implies $|\ell(u) - \ell(v)| \leq 1$.
If $u$ and $v$ are neighboring in $G_0'$, then $u$ and $v$ are incident to a common face of $(G_0',\eta)$.
Indeed, there exists a face of $(G_0',\eta)$ incident to the edge $(u,v)$, and hence incident to both $u$ and $v$.
Therefore, $|\ell(u) - \ell(v)| \leq 1$.
\end{proof}

\begin{fact} \label{fact-sameface}
Let $o_i$ be the outer face of $(G_0'[L_{\geq i}],\eta)$.
The following statements are true. \\
\textbf{(i)} $o_i$ is incident to all vertices in $L_i$ but no vertex in $L_{> i}$. \\
\textbf{(ii)} $o_i$ is the outer face of $(G_0'[L'],\eta)$ for any $L'$ that $L_i \subseteq L' \subseteq L_{\geq i}$. \\
\textbf{(iii)} The image of any vertex $v \in L_1^{i-1}$ in $\varSigma$ under $\eta$ lies in the interior of $o_i$.
\end{fact}
\begin{proof}
We prove the statements by induction on $i$.
The base case is $i=1$.
We have $L_{\geq 1} = V(G_0')$, and hence $o_1 = o$ (recall that $o$ is the outer face of $(G_0',\eta)$).
Note that $o$ is incident to exactly the vertices whose vertex-face distance to $o$ is $1$, i.e., the vertices in $L_1$.
Therefore, $o_1$ is incident to all vertices in $L_1$ but no vertex in $L_{\geq 2}$.
Now consider a set $L'$ that $L_1 \subseteq L' \subseteq L_{\geq 1}$.
Since $(G_0'[L'],\eta)$ is a subgraph of $(G_0'[L_{\geq 1}],\eta) = (G_0',\eta)$, the outer face $o_{L'}$ of $(G_0'[L'],\eta)$ contains $o$.
On the other hand, $(\partial o,\eta)$ is a subgraph of $(G_0'[L_1],\eta)$, which is in turn a subgraph of $(G_0'[L'],\eta)$.
So $o_{L'}$ is contained in the outer face of $(\partial o,\eta)$, which is clearly $o$.
Thus, we have $o_{L'} = o = o_1$.
So statements \textbf{(i)} and \textbf{(ii)} hold for $i=1$, while statement \textbf{(iii)} also holds since $L_1^0 = \emptyset$.

Suppose the statements hold for all $i < j$, and we want to show they hold for $i = j$.
Consider the outer face $o_j$ of $(G_0'[L_{\geq j}],\eta)$.
To verify statement \textbf{(i)}, we first show that $o_j$ is incident to all vertices in $L_j$.
We have $o_{j-1} \subseteq o_j$ as $(G_0'[L_{\geq j}],\eta)$ is a subgraph of $(G_0'[L_{\geq j-1}],\eta)$.
Let $v \in L_j$ be a vertex.
The vertex-face distance between $o$ and $v$ in $(G_0',\eta)$ is $2j-1$, so there is a shortest path $\pi = (o,v_1,f_1,\dots,v_{j-1},f_{j-1},v)$ in the VFI of $(G_0',\eta)$ connecting $o$ and $v$ where $v_1,\dots,v_{j-1} \in V(G_0')$ and $f_1,\dots,f_{j-1} \in F_\eta(G_0')$.
Since $\pi$ is a shortest path, the vertex-face distance between $o$ and $v_{j-1}$ in $(G_0',\eta)$ is $2j-3$, i.e., $v_{j-1} \in L_{j-1}$.
Note that $v_{j-1}$ is incident to $f_{j-1}$.
Also, $v_{j-1}$ is incident to $o_{j-1}$, by our induction hypothesis.
Therefore, if we remove $v_{j-1}$ from $(G_0'[L_{\geq j-1}],\eta)$, then $f_{j-1}$ will be ``merged'' to the outer face of the resulting graph.
It follows that $f_{j-1} \subseteq o_j$.
Since $v$ is incident to $f_{j-1}$, it is also incident to $o_j$.
This shows $o_j$ is incident to all vertices in $L_j$.
Next, we show that $o_j$ is incident to no vertex in $L_{> j}$.
Consider a vertex $v \in L_{> j}$ and a face $f \in F_\eta(G_0')$ of $(G_0',\eta)$ incident to $v$.
By Fact~\ref{fact-diff1}, for any vertex $v' \in V(G_0')$ incident to $f$, $|\ell(v') - \ell(v)| \leq 1$ and hence $\ell(v') \geq j$.
In other words, all vertices incident to $f$ are contained in $L_{\geq j}$.
Therefore, $f$ is also a face of $(G_0'[L_{\geq j}],\eta)$.
Note that $f \neq o$ since $o$ is only incident to the vertices in $L_1$.
It follows that $x_\mathsf{out} \notin f$ and thus $f$ is not the outer face of $(G_0'[L_{\geq j}],\eta)$.
This shows $o_j$ is incident to no vertex in $L_{> j}$.
So statement \textbf{(i)} follows.

Next, we verify statement \textbf{(ii)}.
Since $(G_0'[L'],\eta)$ is a subgraph of $(G_0'[L_{\geq j}],\eta)$, the outer face $o_{L'}$ of $(G_0'[L'],\eta)$ contains $o_j$.
On the other hand, because $o_j$ is not incident to any vertex in $L_{> j}$, $(\partial o_j,\eta)$ is a subgraph of $(G_0'[L_j],\eta)$, which is in turn a subgraph of $(G_0'[L'],\eta)$.
So $o_{L'}$ is contained in the outer face of $(\partial o_j,\eta)$, which is clearly $o_j$.
Thus, $o_{L'} = o_j$ and statement \textbf{(ii)} holds.

Finally, we verify statement \textbf{(iii)}.
By our induction hypothesis, the image of any $v \in L_1^{j-2}$ lies in the interior of $o_{j-1}$ and hence in the interior of $o_j$ as $o_{j-1} \subseteq o_j$.
So it suffices to consider the vertices in $L_{j-1}$.
Again, by the induction hypothesis, all vertices in $L_{j-1}$ are incident to $o_{j-1}$ and hence their images are contained in $o_j$.
To further show that they are contained in the interior of $o_j$, observe that the image of any vertex in $V(G_0') \backslash L_{\geq j}$ must lie in the interior of some face of $(G_0'[L_{\geq j}],\eta)$, simply because $\eta$ is an embedding.
Therefore, the images of the vertices in $L_{j-1}$ lie in the interior of $o_j$, proving statement \textbf{(iii)}.
\end{proof}

We say a face $f \in F_\eta(G_0')$ \textit{hits} a layer $L_i$ if $f$ is incident to some vertex in $L_i$.
By Fact~\ref{fact-diff1}, any face of $(G_0',\eta)$ can hit at most two layers, which must be consecutive.
Recall the vortex faces $D_1',\dots,D_h'$.
If some layer $L_i$ is hit by some vortex face, we say $L_i$ is a \textit{bad} layer.
Since a face can hit at most two layers and there are $h$ vortex faces, the number of bad layers is at most $2h$.
Now set $p' = p+2h$.
We say a number $q \in [p']$ is \textit{bad} if it is congruent to the index of a bad layer modulo $p'$, i.e., $q \equiv i \pmod{p'}$ for some $i \in [m]$ such that $L_i$ is a bad layer, and is \textit{good} otherwise.
As there are at most $2h$ bad layers, there are at most $2h$ bad numbers in $[p']$.
So we can always find $p$ good numbers $q_1,\dots,q_p \in [p']$.
We then construct the sets $Z_1,\dots,Z_p$ by simply defining $Z_i$ as the union of all layers whose indices are congruent to $q_i$ modulo $p'$, i.e., $Z_i = \bigcup_{j=0}^{\lfloor(m-q_i)/p'\rfloor} L_{jp'+q_i}$.
Note that $Z_i \subseteq V(G_0') = V(G_0) \subseteq V(G)$ for $i \in [p]$ and $Z_1,\dots,Z_p$ are disjoint.
Furthermore, it is clear that $Z_1,\dots,Z_p$ can be computed in polynomial time given the almost-embeddable structure of $G$.



\subsection{Bounding the treewidth of $G/(Z_i \backslash Z')$} \label{sec-boundingtw}
To prove Lemma~\ref{lem-key}, it now suffices to show the sets $Z_1,\dots,Z_p$ constructed in the previous section satisfy that for any $i \in [p]$ and $Z' \subseteq Z_i$, $\mathbf{tw}(G/(Z_i \backslash Z')) = O(p+|Z'|)$.
Without loss of generality, it suffices to consider the case $i = 1$, as the constructions of $Z_1,\dots,Z_p$ are the same.
Let $Z' \subseteq Z_1$ be an arbitrary subset.
Our goal is to show $\mathbf{tw}(G/(Z_1 \backslash Z')) = O(p+|Z'|)$.

Recall that $Z_1 = \bigcup_{j=0}^{\lfloor(m-q)/p'\rfloor} L_{jp'+q}$ for some good number $q \in [p']$.
For simplicity of exposition, let us define $L_i = \emptyset$ for any integer $i \in (-\infty,0] \cup [m+1,+\infty)$ and write $i_j = (j-2)p' + q$ for any $j \in \mathbb{N}$.
Then we can write $Z_1 = \bigcup_{j=1}^{m'} L_{i_j}$, where $m' = \lfloor(m-q)/p'\rfloor + 2$.
We first observe a relation between the treewidth of $G/(Z_1 \backslash Z')$ and the treewidth $G_0'/(Z_1 \backslash Z')$.

\begin{lemma} \label{lem-relation}
$\mathbf{tw}(G/(Z_1 \backslash Z')) = O(h \cdot \mathbf{tw}(G_0'/(Z_1 \backslash Z'))+h)$.
\end{lemma}
\begin{proof}
We first notice that $G[Z_1 \backslash Z'] = G_0[Z_1 \backslash Z'] = G_0'[Z_1 \backslash Z']$.
We have $G-A = G_0 \cup (\bigcup_{i=1}^h G_i)$.
The vertices in the intersection of $G_0$ and the vortices $G_1,\dots,G_h$ are all on the boundaries of the vortex faces $D_1',\dots,D_h'$.
Hence, these vertices all lie in the bad layers of $G_0'$, while $Z_1$ is the union of the good layers $L_{j_1},\dots,L_{j_{m'}}$.
This implies $Z_1 \subseteq V(G_0) \backslash (\bigcup_{i=1}^h V(G_i))$ and thus $G[Z_1 \backslash Z'] = G_0[Z_1 \backslash Z']$.
To see $G_0[Z_1 \backslash Z'] = G_0'[Z_1 \backslash Z']$, recall that $G_0'$ is obtained from $G_0$ by adding some virtual edges.
The virtual edges are all on the boundaries of the vortex faces $D_1',\dots,D_h'$, and thus are disjoint from $Z_1$.
So we have $G_0[Z_1 \backslash Z'] = G_0'[Z_1 \backslash Z']$.
Based on this fact, we see that $G_0/(Z_1 \backslash Z')$ is a subgraph of $G_0'/(Z_1 \backslash Z')$ which has the same vertex set as $G_0'/(Z_1 \backslash Z')$, and $G_0/(Z_1 \backslash Z')$ is also a subgraph of $G/(Z_1 \backslash Z')$.
Therefore, the vertex set of $G_0'/(Z_1 \backslash Z')$ is a subset of the vertex set of $G/(Z_1 \backslash Z')$.

Now let $\mathcal{T}^*$ be a tree decomposition of $G_0'/(Z_1 \backslash Z')$ of width $w$ and $T^*$ be the underlying tree of $\mathcal{T}^*$.
We are going to modify $\mathcal{T}^*$ to a tree decomposition of $G/(Z_1 \backslash Z')$ of width $O(hw + h)$, which proves the claim.
To this end, we apply the same argument as in \cite{demaine2005subexponential} (Lemma 5.8).
For convenience, we do not distinguish the vertices in $V(G) \backslash (Z_1 \backslash Z')$ with their images in $G/(Z_1 \backslash Z')$.
For each vertex $v$ of $G_0'/(Z_1 \backslash Z')$, we use $T^*(v)$ to denote the subset of nodes in $T^*$ whose bags contain $v$, which is connected as $\mathcal{T}^*$ is a tree decomposition.
Consider a vortex $G_i$ with the witness pair $(\tau_i,\mathcal{P}_i)$.
Suppose $\sigma_i = (v_{i,1},\dots,v_{i,q_i})$.
Then $\mathcal{P}_i$ is a path decomposition of $G_i$ with path $P = (u_{i,1},\dots,u_{i,q_i})$ such that $v_{i,j} \in \beta(u_{i,j})$ for all $j \in [q_i]$.
We then add the bag $\beta(u_{i,j})$ of $\mathcal{P}_i$ to the bags of all vertices in $T^*(v_{i,j})$.
We do this for all vortices $G_1,\dots,G_h$.
After that, we add the apex set $A$ to the bags of all vertices in $T^*$.
It is easy to verify that $\mathcal{T}^*$ (after the modification) is a tree decomposition of $G/(Z_1 \backslash Z')$.
Indeed, the bags of $\mathcal{T}^*$ cover every edge of $G/(Z_1 \backslash Z')$: the edges in $G_0/(Z_1 \backslash Z')$ are covered by the original bags of $\mathcal{T}^*$, the edges in each vortex $G_i$ are covered by the bags of the path decomposition $\mathcal{P}_i$ (which are added to the bags of the corresponding vertices in $T^*$), and the edges adjacent to the apex set $A$ are also covered because we add $A$ to all bags of $\mathcal{T}^*$.
Furthermore, for any vertex $v$ of $G/(Z_1 \backslash Z')$, the nodes whose bags containing $v$ are connected in $T^*$; this follows from the fact that $(v_{i,1},\dots,v_{i,q_i})$ forms a path in $G_0'/(Z_1 \backslash Z')$ (which consists of virtual edges) for all $i \in [h]$ and thus $\bigcup_{j=j^-}^{j^+} T^*(v_{i,j})$ is connected in $T^*$ for any $j^-,j^+ \in [q_i]$.
Finally, we observe that the width of $\mathcal{T}^*$ is $O(hw+h)$.
Consider a node $t^* \in T^*$.
Originally, the size of the bag $\beta(t^*)$ is at most $w+1$.
If a vortex $G_i$ contains $c_i$ vertices in (the original) $\beta(t^*)$, then we added $c_i$ bags of the path decomposition $\mathcal{P}_i$ to $\beta(t^*)$.
Since the vortices $G_1,\dots,G_h$ are disjoint, each vertex in $\beta(t^*)$ can be contained in at most one vortex, which implies $\sum_{i=1}^h c_i \leq w+1$.
Therefore, we added at most $w+1$ bags of the path decompositions $\mathcal{P}_1,\dots,\mathcal{P}_h$ to $\beta(t^*)$, each of which has size $O(h)$.
So after this step, the size of $\beta(t^*)$ is $O(hw)$.
Then after we added the apex set $A$ to $\beta(t^*)$, the size $\beta(t^*)$ is $O(hw+h)$ because $|A| \leq h$.
\end{proof}

With the above lemma in hand, it now suffices to show $\mathbf{tw}(G_0'/(Z_1 \backslash Z')) = O(p+|Z'|)$.
We shall construct explicitly a tree decomposition of $G_0'/(Z_1 \backslash Z')$ of width $O(p+|Z'|)$.
To this end, we first need to define a support tree $T_\mathsf{supp}$ as follows.
Roughly speaking, $T_\mathsf{supp}$ is a tree that interprets the containment relation between the connected components of $G_0'[L_{> i_1}],\dots,G_0'[L_{> i_{m'}}]$.
The depth of $T_\mathsf{supp}$ is $m'$.
The root (i.e., the node in the 0-th level) of $T_\mathsf{supp}$ is a dummy node.
For all $j \in [m']$, the nodes in the $j$-th level of $T_\mathsf{supp}$ are one-to-one corresponding to the connected components of $G_0'[L_{> i_j}]$.
The parent of the nodes in the first level is just the root.
The parents of the nodes in the lower levels are defined as follows.
Consider a node $t \in T_\mathsf{supp}$ in the $j$-th level for $j \geq 2$, and let $C_t$ be the connected component of $G_0'[L_{> i_j}]$ corresponding to $t$.
Since $G_0'[L_{> i_j}]$ is a subgraph of $G_0'[L_{> i_{j-1}}]$, $C_t$ is contained in a unique connected component of $G_0'[L_{> i_{j-1}}]$, which corresponds to a node $t'$ in the $(j-1)$-th level.
We then define the parent of $t$ as $t'$.
For each node $t \in T_\mathsf{supp}$, we associate to $t$ a set $V_t \subseteq V(G_0')$ defined as follows.
If $t$ is the root of $T_\mathsf{supp}$, $V_t = \emptyset$.
Suppose $t$ is in the $j$-th level for $j \in [m']$ and is corresponding to the connected component $C_t$ of $G_0'[L_{> i_j}]$.
Then let $V_t$ consist of all vertices $v \in V(C_t)$ satisfying $i_j < \ell(v) \leq i_{j+1}$.
We notice the following simple fact.

\begin{fact} \label{fact-support}
The support tree $T_\mathsf{supp}$ and the associated sets $V_t$ satisfy the following properties. \\
\textbf{(i)} $\{V_t\}_{t \in T_\mathsf{supp}}$ is a partition of $V(G_0')$. \\
\textbf{(ii)} The vertices in each connected component of $G_0'[Z_1 \backslash Z']$ are contained in the same $V_t$. \\
\textbf{(iii)} If two vertices $v \in V_t$ and $v' \in V_{t'}$ are neighboring in $G_0'$ and $t \neq t'$, then either $t$ is the parent of $t'$ or $t'$ is the parent of $t$.
\end{fact}
\begin{proof}
To show property \textbf{(i)}, we first observe that every vertex $v \in V(G_0')$ belongs to $V_t$ for some $t \in T_\mathsf{supp}$.
Indeed, there exists some $j \in [m']$ such that $i_j < \ell(v) \leq i_{j+1}$.
Then $v \in L_{> i_j}$ and thus $v$ is contained in some connected component $C$ of $G_0'[L_{> i_j}]$.
Let $t \in T_\mathsf{supp}$ be the node in the $j$-th level corresponding to $C$.
We have $v \in V_t$.
Next, we show $V_t \cap V_{t'} = \emptyset$ if $t \neq t'$.
Suppose $t$ (resp., $t'$) is in the $j$-th (resp., $j'$-th) level.
If $j < j'$, then $V_t \cap V_{t'} = \emptyset$, as $\ell(v) \leq i_{j+1} \leq i_{j'} < \ell(v')$ for all $v \in V_t$ and $v' \in V_{t'}$.
Similarly, we have $V_t \cap V_{t'} = \emptyset$ if $j > j'$.
So it suffices to consider the case $j = j'$.
In this case, since $t \neq t'$, $t$ and $t'$ correspond to different connected components of $G_0'[L_{> i_j}]$ which contain the vertices in $V_t$ and $V_{t'}$ respectively.
Hence, $V_t \cap V_{t'} = \emptyset$.
This shows $\{V_t\}_{t \in T_\mathsf{supp}}$ is a partition of $V(G_0')$.

Next, we verify property \textbf{(ii)}.
Consider a connected component $C$ of $G_0'[Z_1 \backslash Z']$.
Since $Z_1 = \bigcup_{j=1}^{m'} L_{i_j}$ and by Fact~\ref{fact-diff1} the layers $L_{i_1},\dots,L_{i_{m'}}$ are non-adjacent in $G_0'$ (i.e., there is no edge in $G_0'$ connecting two layers $L_{i_j}$ and $L_{i_{j'}}$ for $j \neq j'$), the vertices in $C$ must be contained in the same layer $L_{i_j}$ for some $j \in [m']$.
Also, since $C$ is connected, it must belong to some connected component of $G_0'[L_{> i_{j-1}}]$, which corresponds to a node $t \in T_\mathsf{supp}$ in the $(j-1)$-th level.
By definition, the vertices in $C$ are contained in $V_t$.

Finally, we show property \textbf{(iii)}.
Now let $v \in V_t$ and $v' \in V_{t'}$ be two vertices neighboring in $G_0'$ and assume $t \neq t'$.
Suppose $t$ (resp., $t'$) is in the $j$-th (resp., $j'$-th) level.
By Fact~\ref{fact-diff1}, $|\ell(v) - \ell(v')| \leq 1$, which implies $|j - j'| \leq 1$.
If $j = j'$, then $V_t$ and $V_{t'}$ belong to different connected components of $G_0'[L_{> i_j}]$ and thus $v$ and $v'$ cannot be neighboring in $G_0'$.
So we have either $j = j'+1$ or $j' = j+1$.
Without loss of generality, assume $j = j'+1$.
Let $t^* \in T_\mathsf{supp}$ be the parent of $t$, and we claim that $t^* = t'$.
Indeed, both $t^*$ and $t'$ are in the $j'$-th level of $T_\mathsf{supp}$.
If $t^* \neq t'$, then $t^*$ and $t'$ correspond to different connected components $C_{t^*}$ and $C_{t'}$ of $G_0'[L_{> i_j}]$ respectively.
Note that $v$ is contained in $C_{t^*}$ and $v'$ is contained in $C_{t'}$, which contradicts with the fact that $v$ and $v'$ are neighboring.
Thus $t^* = t'$.
\end{proof}

Now let us consider the graphs $G_0'[V_t]/(V_t \cap (Z_1 \backslash Z'))$ for $t \in T_\mathsf{supp}$.
For convenience, we write $J_t = G_0'[V_t]/(V_t \cap (Z_1 \backslash Z'))$.
By properties \textbf{(i)} and \textbf{(ii)} of Fact~\ref{fact-support}, each $J_t$ is actually an induced subgraph of $G_0'/(Z_1 \backslash Z')$ and these induced subgraphs are disjoint and cover all vertices of $G_0'/(Z_1 \backslash Z')$.
Furthermore, by property \textbf{(iii)} of Fact~\ref{fact-support}, the induced subgraph $J_t$ is only adjacent to the induced subgraphs at the parent and the children of $t$.
Based on this observation, our next plan is the following.
We shall construct, for each node $t \in T_\mathsf{supp}$, a tree decomposition $\mathcal{T}_t^*$ of $J_t$ of width $O(p+|Z'|)$ with a good property: for each child $s$ of $t$, there is a bag of $\mathcal{T}_t^*$ which contains all vertices of $J_t$ that are neighboring to $J_s$ in $G_0'/(Z_1 \backslash Z')$.
By exploiting this good property, we can then ``glue'' the tree decompositions $\{\mathcal{T}_t^*\}_{t \in T_\mathsf{supp}}$ along the edges of $T_\mathsf{supp}$ to obtain a tree decomposition for $G_0'/(Z_1 \backslash Z')$ of width $O(p+|Z'|)$.
In what follows, we first describe the gluing step and then show how to construct the tree decompositions $\{\mathcal{T}_t^*\}_{t \in T_\mathsf{supp}}$.

Suppose now we already have the tree decompositions $\{\mathcal{T}_t^*\}_{t \in T_\mathsf{supp}}$.
Let $T_t^*$ be the underlying (rooted) tree of $\mathcal{T}_t^*$.
We are going to glue the trees $T_t^*$ together to obtain a new tree $T^*$ and then give each node in $T^*$ a bag so that it becomes a tree decomposition $\mathcal{T}^*$ of $G_0'/(Z_1 \backslash Z')$.
Consider a non-root node $s \in T_\mathsf{supp}$ with parent $t$.
By the good property of the tree decomposition $\mathcal{T}_t^*$, there is a node $t^* \in T_t^*$ whose bag $\beta(t^*)$ contains all vertices of $J_t$ that are neighboring to $J_s$ in $G_0'/(Z_1 \backslash Z')$; we call $t^*$ the \textit{portal} of $s$.
We add an edge to connect the root of $T_s^*$ and $t^*$.
We do this for all non-root nodes in $T_\mathsf{supp}$.
After that, we glue all trees $T_t^*$ together and obtain the new tree $T^*$.
Next, we associate to each node $s^* \in T^*$ a bag $\beta(s^*)$ as follows.
Consider a node $s^* \in T^*$ and suppose $s^*$ originally belongs to $T_s^*$ for $s \in T_\mathsf{supp}$.
If $s$ is the root, we simply define $\beta(s^*)$ as the bag of $s^*$ in the tree decomposition $\mathcal{T}_s^*$.
If $s$ is not the root, let $t$ be the parent of $s$ in $T_\mathsf{supp}$ and $t^* \in T_t^*$ be the portal of $s$.
We define $\beta(s^*)$ as the union of the bag of $s^*$ in $\mathcal{T}_s^*$ and the bag of $t^*$ in $\mathcal{T}_t^*$.
Let $\mathcal{T}^*$ be the tree $T^*$ with the associated bags.

\begin{fact}
$\mathcal{T}^*$ is a tree decomposition of $G_0'/(Z_1 \backslash Z')$ of width $O(p+|Z'|)$.
\end{fact}
\begin{proof}
By construction, each bag of $\mathcal{T}^*$ is either a bag or the union of two bags of some tree decompositions in $\{\mathcal{T}_t^*\}_{t \in T_\mathsf{supp}}$.
By assumption, the width of each $\mathcal{T}_t^*$ is $O(p+|Z'|)$ and hence the size of the bags of $\mathcal{T}^*$ is also $O(p+|Z'|)$.
So it suffices to show that $\mathcal{T}^*$ is a tree decomposition of $G_0'/(Z_1 \backslash Z')$.
We first show that every edge $e$ of $G_0'/(Z_1 \backslash Z')$ is covered by some bag of $\mathcal{T}^*$.
As observed before, either $e$ belongs to the induced subgraph $J_t$ for some $t \in T_\mathsf{supp}$ or $e$ connects a vertex in $J_s$ for some $s \in T_\mathsf{supp}$ and a vertex in $J_t$ where $t$ is the parent of $s$ in $T_\mathsf{supp}$.
In the first case, since $\mathcal{T}_t^*$ is a tree decomposition of $J_t$, there is a node $t^* \in \mathcal{T}_t^*$ whose bag contains the two endpoints of $e$.
Note that the bag of $t^*$ in $\mathcal{T}^*$ contains the the bag of $t^*$ in $\mathcal{T}_t^*$, and hence also contains the two endpoints of $e$.
In the second case, let $t^* \in T_t^*$ be the portal of $s$.
Then the bag of $t^*$ in $\mathcal{T}_t^*$ contains the endpoint of $e$ in $J_t$, because it contains all vertices of $J_t$ that are neighboring to $J_s$ in $G_0'/(Z_1 \backslash Z')$.
Also, there is a node $s^* \in T_s^*$ whose bag in $\mathcal{T}_s^*$ contains the endpoint of $e$ in $J_s$.
The bag of $s^*$ in $\mathcal{T}^*$, which is by definition the union of the bag of $s^*$ in $\mathcal{T}_s^*$ and the bag of $t^*$ in $\mathcal{T}_t^*$, contains both endpoints of $e$.
Finally, we show that for any vertex $v$ of $G_0'/(Z_1 \backslash Z')$, the nodes of $T^*$ whose bags contain $v$ are connected in $T^*$.
Suppose $v \in V_t$ for $t \in T_\mathsf{supp}$.
Observe that $v$ is contained in the bags of two types of nodes in $T^*$.
The first type are the nodes which originally belong to $T_t^*$ and whose bags in $\mathcal{T}_t^*$ contain $v$.
The second type are all nodes which originally belong to $T_s^*$ for some child $s$ of $t$ such that the bag of the portal of $s$ in $\mathcal{T}_t^*$ contains $v$.
The nodes of the first type are connected in $T^*$ because they are connected in $T_t^*$.
The nodes of the second type form some connected parts each of which consists of the nodes originally belonging to $T_s^*$ for some child $s$ of $t$.
The part corresponding to $T_s^*$ is adjacent to the portal of $s$, which is a node of the first type since its bag in $\mathcal{T}_t^*$ contains $v$.
In sum, the nodes of the first type are connected and each connected part formed by the nodes of the second type is adjacent to a node of the first type.
Therefore, the nodes whose bags contain $v$ are connected in $T^*$.
\end{proof}

We have seen above that given the tree decompositions $\{\mathcal{T}_t^*\}_{t \in T_\mathsf{supp}}$, one can construct a tree decomposition of $G_0'/(Z_1 \backslash Z')$ of width $O(p+|Z'|)$, which implies $\mathbf{tw}(G_0'/(Z_1 \backslash Z')) = O(p+|Z'|)$ and further implies $\mathbf{tw}(G/(Z_1 \backslash Z')) = O(p+|Z'|)$ by Lemma~\ref{lem-relation}.
Therefore, to complete the proof, we only need to show the existence of the tree decompositions $\{\mathcal{T}_t^*\}_{t \in T_\mathsf{supp}}$.
Recall that $\mathcal{T}_t^*$ is required to be a tree decomposition of $J_t = G_0'[V_t]/(V_t \cap (Z_1 \backslash Z'))$ of width $O(p+|Z'|)$ such that for each child $s$ of $t$ in $T_\mathsf{supp}$, there is a bag of $\mathcal{T}_t^*$ which contains $N_{t,s} \subseteq V_t$, the set of all vertices of $J_t$ that are neighboring to $J_s$ in $G_0'/(Z_1 \backslash Z')$.
Note that this requirement is equivalent to saying that $\mathcal{T}_t^*$ is a tree decomposition of $J_t'$ of width $O(p+|Z'|)$, where $J_t'$ is the graph obtained from $J_t$ by making $N_{s,t}$ a clique for all children $s$ of $t$ (because each clique in a graph is contained in some bag of its tree decomposition).
So it now suffices to prove $\mathbf{tw}(J_t') = O(p+|Z'|)$ for all $t \in T_\mathsf{supp}$.

Fix $j \in [m']$ and let us show that $\mathbf{tw}(J_t') = O(p+|Z'|)$ for all nodes $t$ in the $j$-th level of $T_\mathsf{supp}$.
Set $i^- = i_j+1$ and $i^+ = i_{j+1}$.
Recall that the sets $V_t$ for the nodes $t$ in the $j$-th level correspond to the vertices of $G_0'$ in the layers $L_{i^-},\dots,L_{i^+}$.
Consider the graph $(G_0'[L_{i^-}^{i^+}],\eta)$.
We say a face $f \in F_\eta(G_0'[L_{i^-}^{i^+}])$ of $(G_0'[L_{i^-}^{i^+}],\eta)$ is \textit{deep} if it is not contained in the outer face of $(G_0'[L_{i^+}],\eta)$, and is \textit{shallow} otherwise.
For each deep face $f \in F_\eta(G_0'[L_{i^-}^{i^+}])$, we denote by $\mathcal{C}_f$ the set of connected components of $G_0'[Z_1 \backslash Z']$ that intersect $\partial f$.
For each pair $(f,C)$ where $f \in F_\eta(G_0'[L_{i^-}^{i^+}])$ is a deep face and $C \in \mathcal{C}_f$, we arbitrarily choose a vertex $u_{f,C}$ in the intersection of $\partial f$ and $C$.
Let $U_f = \{u_{f,C}: C \in \mathcal{C}_f\}$.
We define a function $\kappa: F_\eta(G_0'[L_{i^-}^{i^+}]) \rightarrow 2^{L_{i^-}^{i^+}}$ as
\begin{equation*}
    \kappa(f) = \left\{
    \begin{array}{ll}
        U_f \cup (Z' \cap V(\partial f)) & \text{if } f \text{ is deep}, \\
        \emptyset & \text{if } f \text{ is shallow}.
    \end{array}
    \right.
\end{equation*}
Clearly, we have $\kappa(f) \subseteq V(\partial f)$.
Consider the graph obtained from $G_0'[L_{i^-}^{i^+}]$ by making $\kappa(f)$ a clique for all $f \in F_\eta(G_0'[L_{i^-}^{i^+}])$, which we denote by $(G_0'[L_{i^-}^{i^+}])^\kappa$.
We observe the following relation between the graph $(G_0'[L_{i^-}^{i^+}])^\kappa$ and the graphs $J_t'$.
\begin{fact}
$J_t'$ is a minor of $(G_0'[L_{i^-}^{i^+}])^\kappa$ for any node $t$ in the $j$-th level of $T_\mathsf{supp}$.
\end{fact}
\begin{proof}
By definition, $J_t = G_0'[V_t]/(V_t \cap (Z_1 \backslash Z'))$.
So there is a quotient map $\pi: V_t \rightarrow V(J_t) = V(J_t')$, which is surjective.
The vertex set of $(G_0'[L_{i^-}^{i^+}])^\kappa$ is $L_{i^-}^{i^+}$ and $V_t \subseteq L_{i^-}^{i^+}$.
So it suffices to show that for any $u,v \in V(J_t')$, if $(u,v)$ is an edge in $J_t'$, then there exist $\hat{u} \in \pi^{-1}(u)$ and $\hat{v} \in \pi^{-1}(v)$ such that $(\hat{u},\hat{v})$ is an edge in $(G_0'[L_{i^-}^{i^+}])^\kappa$; if this is true, then $J_t'$ is a minor of $(G_0'[L_{i^-}^{i^+}])^\kappa[V_t]$ and hence a minor of $(G_0'[L_{i^-}^{i^+}])^\kappa$.
If $(u,v)$ is an edge in $J_t$, then there exist $\hat{u} \in \pi^{-1}(u)$ and $\hat{v} \in \pi^{-1}(v)$ such that $(\hat{u},\hat{v})$ is an edge in $G_0'[V_t]$ and hence an edge in $(G_0'[L_{i^-}^{i^+}])^\kappa$.
So suppose $(u,v)$ is not an edge in $J_t$.
Then we must have $u,v \in N_{s,t}$ for some child $s$ of $t$.
Suppose $C$ is the connected component of $G_0'[L_{>i^+}]$ corresponding to the node $s$.
Since $C$ is connected, its image in $\varSigma$ under $\eta$ should be contained in a face $f \in F_\eta(G_0'[L_{i^-}^{i^+}])$ of $(G_0'[L_{i^-}^{i^+}],\eta)$.
Furthermore, $f$ must be a deep face, because (the image of) $C$ is not contained in the outer face of $(G_0'[L_{\geq i^+}])$ which is also the outer face of $(G_0'[L_{i^+}])$ by statement \textbf{(ii)} of Fact~\ref{fact-sameface}.
We claim that there exist $\hat{u} \in \pi^{-1}(u)$ and $\hat{v} \in \pi^{-1}(v)$ such that $\hat{u},\hat{v} \in \kappa(f)$, which implies $(\hat{u},\hat{v})$ is an edge in $(G_0'[L_{i^-}^{i^+}])^\kappa$.
Without loss of generality, it suffices to show the existence of $\hat{u}$.
There are two cases to be considered: $\pi^{-1}(u)$ is a single vertex in $V_t \backslash (Z_1 \backslash Z')$ and $\pi^{-1}(u)$ is a connected component of $G_0'[V_t \cap (Z_1 \backslash Z')]$.
In the first case, we simply let $\hat{u}$ be the only vertex in $\pi^{-1}(u)$.
we must have $\hat{u} \in L_{i^+} \subseteq Z_1$ since the vertices in $L_{i^-},\dots,L_{i^+-1}$ are not adjacent to any vertex in $V_s$ by Fact~\ref{fact-diff1}.
But $\hat{u} \notin Z_1 \backslash Z'$, which implies $\hat{u} \in Z'$.
Furthermore, $\hat{u} \in \partial f$ because $\hat{u}$ is adjacent to some vertex in $V_s$ (and hence adjacent to some vertex in $C$) while the image of $C$ in $\varSigma$ is contained in $f$.
It follows that $\hat{u} \in Z_1 \cap \partial f \subseteq \kappa(f)$.
In the second case, $\pi^{-1}(u)$ must intersect $\partial f$, again because $\pi^{-1}(u)$ is adjacent to some vertex in $V_s$.
By property \textbf{(ii)} of Fact~\ref{fact-support}, a connected component of $G_0'[V_t \cap (Z_1 \backslash Z')]$ is also a connected component of $G_0'[Z_1 \backslash Z']$.
Thus, we have $\pi^{-1}(u) \in \mathcal{C}_f$ (recall that $\mathcal{C}_f$ is the set of all connected components of $G_0'[Z_1 \backslash Z']$ which intersect $\partial f$).
By our construction of the set $U_f$, there exists a vertex $u_{f,\pi^{-1}(u)} \in U_f$ which lies in the intersection of $\partial f$ and $\pi^{-1}(u)$.
Setting $\hat{u} = u_{f,\pi^{-1}(u)}$, we have $\hat{u} \in \pi^{-1}(u)$ and $\hat{u} \in U_f \subseteq \kappa(f)$.
\end{proof}

By the above observation, we have $\mathbf{tw}(J_t') \leq \mathbf{tw}((G_0'[L_{i^-}^{i^+}])^\kappa)$ for any node $t$ in the $j$-th level of $T_\mathsf{supp}$.
So it suffices to show $\mathbf{tw}((G_0'[L_{i^-}^{i^+}])^\kappa) = O(p+|Z'|)$.
Let $w_\kappa: F_\eta(G_0'[L_{i^-}^{i^+}]) \rightarrow \mathbb{N}$ be the weight function on the faces of $(G_0'[L_{i^-}^{i^+}],\eta)$ defined as $w_\kappa(f) = |\kappa(f)|$.
Recall that $\mathsf{diam}_{w_\kappa}^*(G_0'[L_{i^-}^{i^+}],\eta)$ is the $w_\kappa$-weighted vertex-face diameter of $(G_0'[L_{i^-}^{i^+}],\eta)$.
By Lemma~\ref{lem-twdiam2}, we have $\mathbf{tw}((G_0'[L_{i^-}^{i^+}])^\kappa) = O(\mathsf{diam}_{w_\kappa}^*(G_0'[L_{i^-}^{i^+}],\eta))$.
In order to bound $\mathsf{diam}_{w_\kappa}^*(G_0'[L_{i^-}^{i^+}],\eta)$, we first need to bound the value of the weight function $w_\kappa$.
Clearly, $w_\kappa(f) = 0$ if $f$ is shallow.
The following observation helps bound $w_\kappa(f)$ for deep faces $f \in F_\eta(G_0'[L_{i^-}^{i^+}])$.

\begin{fact} \label{fact-deepface}
Let $f \in F_\eta(G_0'[L_{i^-}^{i^+}])$ be a deep face of $(G_0'[L_{i^-}^{i^+}],\eta)$.
Then $\partial f$ intersects at most $O(|Z'|)$ connected components of $G_0'[Z_1 \backslash Z']$.
\end{fact}
\begin{proof}
Let $f \in F_\eta(G_0'[L_{i^-}^{i^+}])$ be a deep face.
Since $Z_1 = \bigcup_{j=1}^{m'} L_{i_j}$, we have $Z_1 \cap L_{i^-}^{i^+} = L_{i^+}$.
So it suffices to show that $\partial f$ intersects $O(|Z'|)$ connected components of $G_0'[L_{i^+} \backslash Z']$.
Let $o_{i^+} \in F_\eta(G_0'[L_{i^+}])$ be the outer face of $(G_0'[L_{i^+}],\eta)$.
By the assumption that $f$ is deep, we have $f \nsubseteq o_{i^+}$.
Note that $o_{i^+}$ is also a face of $(\partial o_{i^+},\eta)$, which is clearly the outer face.
Because $\partial o_{i^+}$ is a subgraph of $G_0'[L_{i^+}]$ which is in turn a subgraph of $G_0'[L_{i^-}^{i^+}]$, $f$ must be contained in some face $f' \in F_\eta(\partial o_{i^+})$ of $(\partial o_{i^+},\eta)$ other than $o_{i^+}$.
By statement \textbf{(i)} of Fact~\ref{fact-sameface}, all vertices in $L_{i^+}$ are incident to $o_{i^+}$ and thus $V(\partial o_{i^+}) = L_{i^+}$.
Since $f \subseteq f'$ and $f'$ is a face of $(\partial o_{i^+},\eta)$, any vertex of $\partial o_{i^+}$ incident to $f$ must be on the boundary of $f'$, which implies $V(\partial f) \cap V(\partial o_{i^+}) \subseteq V(\partial f')$, i.e., $V(\partial f) \cap L_{i^+} \subseteq V(\partial f')$.
Therefore, it suffices to show that $\partial f'$ intersects $O(|Z'|)$ connected components of $G_0'[L_{i^+} \backslash Z']$.
As $f'$ is a face of $(\partial o_{i^+},\eta)$ and $f' \neq o_{i^+}$, by Lemma~\ref{lem-degree}, $\partial f'$ has $O(h)$ connected components and the maximum degree of $\partial f'$ is $O(h)$.
It is well-known that after removing $\alpha$ vertices from a graph of maximum degree $\Delta$, the number of connected components of the graph increases by at most $\alpha (\Delta - 1)$ (for completeness, we include a proof for this below).
Thus, $\partial f' - Z'$ has $O(h+ |Z'|(h-1))$ connected components, i.e., $O(|Z'|)$ connected components.
Note that $\partial f' - Z'$ is a subgraph of $G_0'[L_{i^+} \backslash Z']$, so every connected component of $\partial f' - Z'$ is contained in a connected component of $G_0'[L_{i^+} \backslash Z']$.
It follows that $\partial f' - Z'$ intersects $O(|Z'|)$ connected components of $G_0'[L_{i^+} \backslash Z']$, and hence $\partial f'$ intersects $O(|Z'|)$ connected components of $G_0'[L_{i^+} \backslash Z']$.

Finally, we prove the statement that after removing $\alpha$ vertices from a graph $G$ of maximum degree $\Delta$, the number of connected components of $G$ increases by at most $\alpha (\Delta - 1)$.
It suffices to consider the case $\alpha = 1$, because removing vertices can never increase the maximum degree of the graph.
Let $v$ be the vertex to be removed from $G$ and $C$ be the connected component of $G$ containing $v$.
Because $C$ is connected, every connected component of $C - \{v\}$ must contain a neighbor of $v$.
But $v$ can have at most $\Delta$ neighbors.
Therefore, $C - \{v\}$ can have at most $\Delta$ connected components.
In other words, after removing $v$, $C$ splits to at most $\Delta$ connected components.
Hence, the number of connected components of $G$ increases by at most $\Delta - 1$.
\end{proof}

The above observation implies that for any deep face $f \in F_\eta(G_0'[L_{i^-}^{i^+}])$, $|U_f| = |\mathcal{C}_f| = O(|Z'|)$ and thus $w_\kappa(f) = O(|Z'|)$.
With this in hand, we can finally bound $\mathsf{diam}_{w_\kappa}^*(G_0'[L_{i^-}^{i^+}],\eta)$, which then completes the entire proof of Lemma~\ref{lem-key}.

\begin{fact}
$\mathsf{diam}_{w_\kappa}^*(G_0'[L_{i^-}^{i^+}],\eta) = O(p+|Z'|)$.
\end{fact}
\begin{proof}
Let $o_{i^-}$ be the outer face of $(G_0'[L_{\geq i^-}],\eta)$, which is also the outer face of $(G_0'[L_{i^-}^{i^+}],\eta)$ by statement \textbf{(ii)} of Fact~\ref{fact-sameface}.
We first show that for any vertex $v \in L_{i^-}^{i^+}$, the $w_\kappa$-weighted vertex-face distance between $o_{i^-}$ and $v$ in $(G_0'[L_{i^-}^{i^+}],\eta)$ is $O(p)$.
Suppose $v \in L_i$.
So the vertex-face distance between the outer face $o$ of $(G_0',\eta)$ and $v$ in $(G_0',\eta)$ is $2i-1$.
Let $\pi = (o,v_1,f_1,\dots,v_{i-1},f_{i-1},v)$ be a shortest path in the VFI graph of $(G_0',\eta)$ connecting the outer face $o$ of $(G_0',\eta)$ and $v$, where $v_1,\dots,v_{i-1} \in V(G_0')$ and $f_1,\dots,f_{i-1} \in F_\eta(G_0')$.
Since $\pi$ is a shortest path, we have $v_j \in L_j$ for all $j \in [i-1]$.
We claim that $f_{i^-},\dots,f_{i-1} \in F_\eta(G_0'[L_{i^-}^{i^+}])$ and $f_{i^-+1},\dots,f_{i-1}$ are all shallow faces.
Let $j \in \{i^-+1,\dots,i-1\}$.
Then $f_j$ is incident to $v_{j-1} \in L_{j-1}$ and $v_j \in L_j$.
By Fact~\ref{fact-diff1}, all vertices incident to $f_j$ must lie in $L_{j-1}$ or $L_j$, which implies $V(\partial f_j) \subseteq L_{j-1} \cup L_j \subseteq L_{i^-}^{i^+}$.
Therefore, $f_j$ is also a face of $(G_0'[L_{i^-}^{i^+}],\eta)$, i.e., $f_j \in F_\eta(G_0'[L_{i^-}^{i^+}])$.
To see $f_j$ is a shallow face, we need to show that $f_j \subseteq o_{i^+}$ where $o_{i^+}$ is the outer face of $(G_0'[L_{i^+}],\eta)$.
Since $G_0'[L_{i^+}]$ is a subgraph of $G_0'$, $f_j$ must be contained in some face of $(G_0'[L_{i^+}],\eta)$.
In addition, as $f_j$ is incident to $v_{j-1}$ and $v_{j-1}$ is in the interior of $o_{i^+}$ by statement \textbf{(iii)} of Fact~\ref{fact-sameface}, $f_j$ cannot be contained in any face of $(G_0'[L_{i^+}],\eta)$ other than $o_{i^+}$.
Thus, $f_j \subseteq o_{i^+}$.
This shows $f_{i^-},\dots,f_{i-1}$ are all shallow faces of $(G_0'[L_{i^-}^{i^+}],\eta)$.
By statement \textbf{(i)} of Fact~\ref{fact-sameface}, $o_{i^-}$ is incident to $v_{i^-}$ as $v_{i^-} \in L_{i^-}$.
Thus, we obtain a path $(o_{i^-},v_{i^-},f_{i^-},\dots,v_{i-1},f_{i-1},v)$ connecting $o_{i^-}$ and $v$ in $(G_0'[L_{i^-}^{i^+}],\eta)$.
The cost of this path is $(2i-2i^-+1)+ w_\kappa(o_{i^-}) +\sum_{j=i^-}^{i-1} w_\kappa(f_j)$.
Since $f_{i^-},\dots,f_{i-1}$ are all shallow, $w_\kappa(f_{i^-}) = \cdots = w_\kappa(f_{i-1}) = 0$.
Also, $o_{i^-}$ is shallow because $o_{i^-} \subseteq o_{i^+}$, so $w_\kappa(o_{i^-}) = 0$.
Therefore, the cost of the path is $2i-2i^-+1 = O(p)$.
It follows that the $w_\kappa$-weighted vertex-face distance between $o_{i^-}$ and any $v \in L_{i^-}^{i^+}$ in $(G_0'[L_{i^-}^{i^+}],\eta)$ is $O(p)$.
By the triangle inequality, this further implies the $w_\kappa$-weighted vertex-face distance between any two vertices $v,v' \in L_{i^-}^{i^+}$ in $(G_0'[L_{i^-}^{i^+}],\eta)$ is $O(p)$.
By Fact~\ref{fact-connected}, every $f \in F_\eta(G_0'[L_{i^-}^{i^+}])$ is adjacent to some $v \in L_{i^-}^{i^+}$ in the VFI graph of $(G_0'[L_{i^-}^{i^+}],\eta)$.
Therefore, the $w_\kappa$-weighted vertex-face distance between two faces $f,f' \in F_\eta(G_0'[L_{i^-}^{i^+}])$ in $(G_0'[L_{i^-}^{i^+}],\eta)$ is $O(p+w_\kappa(f)+w_\kappa(f'))$, and the $w_\kappa$-weighted vertex-face distance between a vertex $v \in L_{i^-}^{i^+}$ and a face $f \in F_\eta(G_0'[L_{i^-}^{i^+}])$ is $O(p+w_\kappa(f))$.
Note that $w_\kappa(f) = O(|Z'|)$ for all $f \in F_\eta(G_0'[L_{i^-}^{i^+}])$, which implies that $\mathsf{diam}_{w_\kappa}^*(G_0'[L_{i^-}^{i^+}],\eta) = O(p+|Z'|)$.
\end{proof}

\section{Applications of the framework} \label{sec-app}
In this section, we apply our framework of Section \ref{sec:framework} on various problems to obtain sub-exponential algorithms on $H$-minor free graphs. Given the framework, its application on a problem is fairly straightforward, except we have to ensure that the problem satisfies the requirements of this framework. First, we apply this framework on \probOCT and describe all the steps in detail. For other problems, the applications are similar. Thus, for convenience, we will just describe the changes needed for those problems omitting the trivial details. In all of our problems, we assume that the input graph is $H$-minor-free.

\subsection{\probOCT and \probEB}
First, we consider the  \probOCT (OCT) problem. In OCT, we are given a graph $G=(V,E)$ and a parameter $k$, and we would like to decide whether there is a subset $V'\subseteq V$ of size at most $k$ such that for each odd cycle $C$ in $G$, there is a vertex of $C$ that is in $V'$. If there is such a subset, we refer to it as a solution. By abusing the notation, we refer to a solution to the OCT problem also as an OCT. The distinction should be clear from the context. Following our framework, the first step is to apply Lemma~\ref{lem-candidate} to obtain a candidate set $\mathsf{Cand} \subseteq V$ for the instance. Note that this takes only a polynomial time. For our purpose, it is sufficient to find a solution which is a subset of $\mathsf{Cand}$. Henceforth, by the term solution we always mean a solution in $\mathsf{Cand}$.       

Note that after removing the vertices in an OCT from $G$, a bipartite graph is left. For such a bipartition, we arbitrarily define the left and the right side (or partition). Consider any subgraph $G_1=(V_1,E_1)$ and fix any solution to OCT on $G_1$. We define a configuration $\lambda$ of $V'\subseteq V_1$ with respect to the solution as a map $V' \rightarrow \{0,1,2\}$.
The three values $0,1,2$ denote whether the corresponding vertex $v \in V'$ is in the OCT, is in the left side of the residual bipartite graph (after removing the solution OCT) or is in the right side of the residual bipartite graph.
As we want our solution to be a subset of $\mathsf{Cand}$, the 0-vertices must always come from $\mathsf{Cand}$.
Also, as the residual graph is bipartite, any two 1-vertices (resp., 2-vertices) cannot be neighboring.
Otherwise, we say the configuration is not valid.

\subsubsection{Contraction-friendly tree-decomposition DP} 

First, we describe the desired contraction-friendly tree-decomposition DP for OCT.
We are given an $n$-vertex graph $G=(V,E)$, a (possibly empty) subset $X \subseteq V$, and a tree-decomposition $\mathcal{T}$ of $G/X$. Let $T$ be the underlying tree of $\mathcal{T}$. For any $0 \le l \le k$, we would like to decide whether there is a solution to OCT on $G$ of size $l$ which is disjoint from $X$. In particular, we would like to use the tree-decomposition $\mathcal{T}$ to solve this problem efficiently. For that purpose, first we define a set of sub-problems on each node $t\in T$. As the bag $\beta(t)$ of $t$ might contain contracted vertices, we have to be careful in defining the sub-problems. 

To define the sub-problem, we need the following definition. For each subset $V'$ of vertices in $G/X$, let ${(G/X)(V')}^{-1}$ be the subgraph of $G$ that is the pre-image of ${(G/X)}[V']$, i.e., ${(G/X)(V')}^{-1}$ is obtained by undoing the contractions of all vertices in the induced subgraph ${(G/X)}[V']$. 

\begin{observation}\label{obs:distinct-config}
The number of distinct (valid) configurations of the vertices in ${(G/X)(V')}^{-1}$ with respect to the solutions disjoint from $X$ is $2^{O(|V'|)}$. 
\end{observation} 

\begin{proof}
For each contracted vertex $v$ of $V'$, consider the connected component $C$ in $G[X]$ corresponding to $v$. Then as the solutions we seek are disjoint from $X$, $C$ must be bipartite. Hence, the number of distinct valid configuration of the vertices in $C$ with respect to such  solutions is only 2 depending on which partition is designated as the left or right. Also, any vertex in $V'$ whose pre-image in $V(G)$ is a single vertex outside $X$ can have 3 configurations. Hence, the number of such distinct configurations of the vertices in ${(G/X)(V')}^{-1}$ is at most $3^{|V'|}=2^{O(|V'|)}$.  
\end{proof}

Note that given ${(G/X)(V')}^{-1}$, one can easily compute the $2^{O(|V'|)}$ such configurations. For any solution $S$ and a configuration $\lambda$ of a set of vertices $V'$, we say $S$ is consistent with $\lambda$ if the 0-vertices of $V'$ are in $S$, and 1- and 2-vertices of $V'$ are in the left and right side of the residual bipartite graph.
For every valid configuration $\lambda$ of the vertices in ${(G/X)(\sigma(t))}^{-1}$ and $0\le l\le k$, we define a sub-problem where the goal is to decide whether there is a solution to OCT on ${(G/X)(\gamma(t))}^{-1}$ which is disjoint from $X$, is consistent with $\lambda$ and uses exactly $l$ vertices of ${(G/X)(\gamma(t)\backslash \sigma(t))}^{-1}$.
Thus, in this sub-problem, the solution contains 0-vertices of $\sigma(t)$ and $l$ additional vertices from ${(G/X)(\gamma(t)\backslash \sigma(t))}^{-1}$.
We store a DP table $L(t)$ indexed by $(\lambda,l)$, for all valid configurations $\lambda$ of ${(G/X)(\sigma(t))}^{-1}$ and $0\le l\le k$.
The entry of $L(t)$ corresponding to the pair $(\lambda,l)$ stores the YES/NO solution for the sub-problem defined by $\lambda$ and $l$.
By Observation \ref{obs:distinct-config}, the number of entries in $L(t)$ is $2^{O(|\sigma(t)|)} \cdot k$, and all such relevant configurations can be computed in $2^{O(|\sigma(t)|)}n^{O(1)}$ time.

For our application, we cannot assume that the width of $\mathcal{T}$ is bounded. Instead, we make two assumptions: (i) the bag size of each non-leaf vertex in $T$ is bounded by a parameter $w$, and (ii) for each leaf $t\in T$, either the bag size is bounded by $w$ or the table $L(t)$ is given to us.
With these assumptions, we show that the table entries corresponding to all vertices of $T$ can be computed in time $2^{O(w)}n^{O(1)}$. 

Consider any leaf node $t \in T$.
If the table $L(t)$ is given to us, we are done.
Otherwise, $|\beta(t)| \leq w$ by assumption.
Suppose we want to compute the entry of the table $L(t)$ corresponding to $(\lambda,l)$.
Since $|\beta(t)| \leq w$, the valid configurations of the vertices in ${(G/X)(\beta(t))}^{-1}$ is $2^{O(w)}$.
We simply enumerate all these configurations.
Each such configuration gives us an OCT $S \subseteq \beta(t)$ of $G[\beta(t)]$.
If $S$ is a solution to the sub-problem $(\lambda,l)$, the entry should be YES.
If none of these $S$ is a solution to the sub-problem $(\lambda,l)$, the entry should be NO.

Now consider any non-leaf vertex $t \in T$. We can assume that the tables corresponding to the children of $t$ are already computed. Fix any configuration $\lambda$ of the vertices in ${(G/X)(\sigma(t))}^{-1}$ and an integer $0\le l\le k$. We want to compute the entry in $L(t)$ with respect to $(\lambda,l)$. First, we guess a configuration $\lambda'$ of the vertices in ${(G/X)(\beta(t)\backslash \sigma(t))}^{-1}$ with respect to the solutions disjoint from $X$. Thus, we have the configuration of all vertices in $\beta(t)$. This gives us the configuration $\lambda(s)$ of $\sigma(s)$ for each child $s$ of $t$. We note that the vertices of $\beta(t)$ might be shared by multiple children. But, in our sub-problems, the solution size does not account for the vertices in the adhesion. Thus, these solutions account for disjoint sets of vertices, and hence the sub-problems can be treated as independent. Let $l^u$ be the number of 0-vertices in $\beta(t)\backslash \sigma(t)$ with respect to $\lambda'$. Then, we would like to compute a solution, which consists of these 0-vertices and $l_s$ vertices from ${(G/X)(\gamma(s)\backslash \sigma(s))}^{-1}$ for each child $s$ of $t$, where $l=l^u+\sum_s l_s$. Using the tables of the children, we can easily compute such a solution if possible, and hence the entry $(\lambda,l)$ can be computed in polynomial time. Now, as we need to enumerate all possible $2^{O(|\beta(t)\backslash \sigma(t)|)}$ configurations $\lambda'$, it is not hard to see that computation of each entry takes $2^{O(|\beta(t)\backslash \sigma(t)|)}n^{O(1)}$ time. Hence, the table for $t$ can be computed in time $2^{O(|\beta(t)|)}n^{O(1)}=2^{O(w)}n^{O(1)}$.


\subsubsection{The algorithm}
Now, we describe our main algorithm. First, we apply Lemma \ref{lem-rsdecomp} to compute in polynomial time a tree decomposition $\mathcal{T}_\textnormal{RS}$ of $G$ with adhesion size at most $h$ in which the torso of each vertex is $h$-almost-embeddable. Note that $h$ is a constant which depends only on $H$. Given $\mathcal{T}_\textnormal{RS}$, the idea is to employ a DP on it to solve OCT on $G$. However, we have to be careful here, as the width of $\mathcal{T}_\textnormal{RS}$ can be very large. Next, we describe our DP based algorithm to solve OCT on $\mathcal{T}_\textnormal{RS}$. 

Let ${T}_\textnormal{RS}$ be the underlying tree of $\mathcal{T}_\textnormal{RS}$. The DP algorithm stores a table $M(t)$ for each vertex $t \in {T}_\textnormal{RS}$. 
Each entry of $M(t)$ is indexed by a tuple $(\lambda,l)$ for a configuration $\lambda$ of $\sigma(t)$ and an integer $0\le l\le k$ which stores a boolean value YES/NO based on whether or not there is an OCT for $G[\gamma(t)]$ such that the OCT is consistent with $\lambda$ and uses exactly $l$ vertices of $\gamma(t)\backslash \sigma(t)$.
That is, the entry with respect to $(\lambda,l)$ stores the solution of the problem $\mathsf{Prob}_{\lambda,l}$ as defined in our framework. Note that the size of the table $M(t)$ is $O(k)$, as $\sigma(t)$ is at most the constant $h$. Next, we show how to compute the entries of $M(t)$. 

Consider any vertex $t\in {T}_\textnormal{RS}$.   
We can assume that the table entries corresponding to the children (if any) of $t$ are already computed. This is true, as in the base case, we deal with the leaves of ${T}_\textnormal{RS}$.
Recall in our framework, we defined the sub-problem $\mathsf{Prob}_{\lambda,l,Y}$ for $Y \subseteq \gamma(t)$ as the problem asking whether there is a solution for $\mathsf{Prob}_{\lambda,l}$ that is disjoint from $Y$.
First, we apply Corollary~\ref{cor-extention} on $t$ with $p = \lfloor \sqrt{k} \rfloor$ to obtain the sets $Y_1,\dots,Y_p \subseteq \gamma(t) \backslash \sigma(t)$. Then, following the framework, we construct the set $\varPi$ of pairs $(i,Y')$, where $i \in [p]$ and $Y' \subseteq Y_i$, satisfying three conditions: \textbf{(1)} $|\varPi| = |\mathsf{Cand}|^{O(\sqrt{k})}$, \textbf{(2)} $|Y'| = O(k/p)$ for all $(i,Y') \in \varPi$, and \textbf{(3)} for any configuration $\lambda$ of $\sigma(t)$ and $0\le l\le k$, the answer to $\mathsf{Prob}_{\lambda,l}$ is YES iff the answer to $\mathsf{Prob}_{\lambda,l,Y_i \backslash Y'}$ is YES for some $(i,Y') \in \varPi$. Thus, we can focus on a fixed set $Y'\subseteq Y_i$ and solve $\mathsf{Prob}_{\lambda,l,Y_i \backslash Y'}$ for all $\lambda$ and $l$. Again, following the framework, we apply Lemma \ref{lem-specialtd} to obtain a tree-decomposition $\mathcal{T}^*$ of $G[\gamma(t)]/(Y_i \backslash Y')$ in time $2^{O(\sqrt{k})} n^{O(1)}$.
We shall apply the DP procedure described in the previous section on $\mathcal{T}^*$.
Specifically, on each node $t^*$ of the underlying tree ${T}^*$ of $\mathcal{T}^*$, we want to compute the table $L(t^*)$ indexed by $(\lambda,l)$, for all valid configurations $\lambda$ of ${(G[\gamma(t)]/(Y_i \backslash Y'))(\sigma(t^*))}^{-1}$ and $0\le l\le k$..
Any non-leaf vertex in the underlying tree ${T}^*$ of $\mathcal{T}^*$ has bag size $O(\sqrt{k})$.
However, for a leaf, the bag size can be large.
Nevertheless, by property \textbf{(iii)} in Lemma \ref{lem-specialtd}, for such a leaf $t^* \in {T}^*$, $\gamma(t^*) = \beta(t^*) = \gamma(s)$ and $\sigma(t^*) = \sigma(s)$ for some child $s$ of $t$.
Thus, the table $L(t^*)$ is nothing but the table $M(s)$ for $s$ (in the DP on ${T}_\textnormal{RS}$), and by assumption, the table $M(s)$ is already computed.
Therefore, as we have shown in the previous section, the DP on $\mathcal{T}^*$ can be done in $2^{O(\sqrt{k})} n^{O(1)}$ time.
Now, by Property \textbf{(ii)} in Lemma \ref{lem-specialtd}, the root $\mathsf{rt}$ of $T^*$ has only one child $\mathsf{rt}'$, where $\sigma(\mathsf{rt}') = \sigma(t)$ and $ \beta(\mathsf{rt}) = \sigma(t)$, which implies $\gamma(\mathsf{rt}')$ contains all vertices in $G[\gamma(t)]/(Y_i \backslash Y')$.
Also note that $\sigma(t) \cap Y_i=\emptyset$, and so $\sigma(t)\cap (Y_i \backslash Y')=\emptyset$. Thus, the vertices of $\sigma(t)$ are not contracted in $\mathcal{T}^*$.
Hence, the table $L(\mathsf{rt}')$ encodes the answers for the sub-problems $\mathsf{Prob}_{\lambda,l,Y_i \backslash Y'}$ for all $(\lambda,l)$.
After considering all pairs $(i,Y') \in \varPi$, we solve all sub-problems $\mathsf{Prob}_{\lambda,l,Y_i \backslash Y'}$ and compute the DP table $M(t)$ for $t$.
Note that in our solutions, we do not count the vertices of $\sigma(t)$. But, this can be easily fixed at the end by adding the 0-vertices for each configuration to the computed solution and computing the best combined solution.

Now, the DP on $T^*$ takes $2^{O(\sqrt{k})}n^{O(1)}$ time. Also, the total number of pairs $(i,Y')$ in $\varPi$ is ${|\mathsf{Cand}|}^{O(\sqrt{k})}$.
Thus, the computation of $M(t)$ takes in total $|\mathsf{Cand}|^{O(\sqrt{k})}\cdot 2^{O(\sqrt{k})}\cdot n^{O(1)}$ time.
By Lemma \ref{lem-candidate}, the size of $\mathsf{Cand}$ for OCT is $k^{O(1)}$.
Subsequently, we obtain the time bound of $2^{O(\sqrt{k}\log k)}n^{O(1)}$.
It follows that the OCT problem on $G$ can be solved in the same time.
The computation of the candidate set $\mathsf{Cand}$ is randomized with success probability $1-1/2^n$.
So we conclude the following.

\begin{theorem}
\probOCT on $H$-minor free graphs can be solved in $2^{O(\sqrt{k}\log k)}n^{O(1)}$ time by a randomized algorithm with success probability $1-1/2^n$. 
\end{theorem}

Besides computing the candidate set $\mathsf{Cand}$ using the randomized algorithm in Lemma~\ref{lem-candidate}, we can also use the trivial candidate set $\mathsf{Cand} = V(G)$.
In this case, we obtain an $n^{O(\sqrt{k})}$-time deterministic algorithm.
\begin{theorem}
\probOCT on $H$-minor free graphs can be solved in $n^{O(\sqrt{k})}$ time.
\end{theorem}

\subsubsection{\probEB}
The \probEB (EB) problem is similar to OCT, except here we need to remove a subset of edges. Formally, we are given an $H$-minor-free graph $G=(V,E)$ and a parameter $k$, and we would like to decide whether there is a subset $E'\subseteq E$ of size at most $k$ such that removal of the edges in $E'$ from $G$ gives us a bipartite graph. As we saw in the previous subsection, the main ingredient for applying the framework is the contraction-friendly tree-decomposition DP. Here, we describe such a DP for EB. The rest of the algorithm is similar to the one for OCT. 


Consider any subgraph $G_1=(V_1,E_1)$ and fix any solution $E_\mathsf{sol}$ on $G_1$. We say that $E_\mathsf{sol}$ is disjoint from a set of vertices $X$ if the endpoints of the edges in $E_\mathsf{sol}$ are disjoint from $X$. We define a configuration $\lambda$ of $V'\subseteq V_1$ with respect to the solution to be a set of values corresponding to the vertices  in $V'$, where each value is in $\{1,2\}$. We note that these values denote whether the corresponding vertex $v$ is in the left side of the residual bipartite graph (after removing the solution edges) or is in the right side of the residual bipartite graph. Next, we describe the contraction-friendly tree-decomposition DP for EB.  

We are given an $n$-vertex graph $G=(V,E)$, a (possibly empty) subset $X\subseteq V$, and a tree-decomposition $\mathcal{T}$ of $G/X$. Let $T$ be the underlying tree of $\mathcal{T}$. For any $0\le l\le k$, we would like to decide whether there is a solution to EB on $G$ of size $l$ which is disjoint from $X$. The following observation is similar to Observation \ref{obs:distinct-config}. 

\begin{observation}\label{obs:distinct-config-EB}
The number of distinct configurations of the vertices in ${(G/X)(V')}^{-1}$ with respect to the solutions disjoint from $X$ is $2^{O(|V'|)}$. 
\end{observation} 

For any solution $S$ and a configuration $\lambda$ of a set of vertices $V'$, we say $S$ is consistent with $\lambda$ if the 1- and 2-vertices of $V'$ are in the left and right side of the residual bipartite graph. For any configuration $\lambda$ of the vertices in ${(G/X)(\sigma(t))}^{-1}$ and $0\le l\le k$, we define a sub-problem where the goal is to decide whether there is a solution to EB on ${(G/X)(\gamma(t))}^{-1}$ which is disjoint from $X$, consistent with $\lambda$ and uses exactly $l$ edges such that at least one endpoint of each edge is not in ${(G/X)(\sigma(t))}^{-1}$. Thus, in this sub-problem, the solution contains the edges whose both endpoints are 1-vertices (or 2-vertices) among the set of edges $E^1$ in ${(G/X)(\sigma(t))}^{-1}$ and $l$ additional edges from ${(G/X)(\gamma(t))}^{-1}$ which are not in $E^1$. We store the boolean value (YES/NO) of the solution of this sub-problem in a table $L(t)$ indexed by $(\lambda,l)$. 

We note that it is sufficient to solve the above sub-problem, as the edges used in a solution from ${(G/X)(\sigma(t))}^{-1}$ can be easily derived from the corresponding configuration. The sub-problem is defined in this way to avoid any potential multiple counting of a solution edge at different levels. In this setting, if an edge $(u,v)$ is common to both parent and child sub-problems, the decision to take it or not is uniquely determined by the configuration corresponding to the children. 

By Observation \ref{obs:distinct-config-EB}, the number of entries that we need to store in $L(t)$ is $2^{O(|\sigma(t)|)}\cdot k$, and all such relevant configurations can be computed in $2^{O(|\sigma(t)|)}n^{O(1)}$ time. With these modified definitions of sub-problem and DP table, the rest of the algorithm is similar to  the one for OCT. Hence, we conclude that the table entries corresponding to all vertices of $T$ can be computed in time $2^{O(w)}n^{O(1)}$, where $w$ is the maximum size of the bags of the non-leaf vertices in $T$. 

Note that by our framework, the total running time of our algorithm is $(|\mathsf{Cand}|^{O(\sqrt{k})}\cdot 2^{O(w)})\cdot n^{O(1)}$, where $w=\sqrt{k}$. By Lemma \ref{lem-candidate}, the size of $\mathsf{Cand}$ for EB is $k^{O(1)}$. Subsequently, we obtain the time bound of $2^{O(\sqrt{k}\log k)}n^{O(1)}$.  

\begin{theorem}
\probEB on $H$-minor free graphs can be solved in $2^{O(\sqrt{k}\log k)}n^{O(1)}$ time by a randomized algorithm with success probability $1-1/2^n$.
\end{theorem}

Same as before, we can use the trivial candidate set $\mathsf{Cand} = E(G)$ to obtain an $n^{O(\sqrt{k})}$-time algorithm.

\begin{theorem}
\probEB on $H$-minor free graphs can be solved in $n^{O(\sqrt{k})}$ time.
\end{theorem}


\subsection{\probGFVS and \probGFES}
For a directed graph $G=(V,A)$, a finite group $\Sigma$ of size $g$ and a labeling function $\Lambda: A\rightarrow \Sigma$, $(G,\Lambda)$ is called a $\Sigma$-labeled graph if for each arc $(u,v)\in A$, $(v,u)$ is also in $A$ and $\Lambda((u,v))=\Lambda((v,u))^{-1}$. For a subset of vertices $D$, we denote by $(G\backslash D,\Lambda)$ the $\Sigma$-graph obtained by removing the vertices in $D$ from $(G,\Lambda)$. Similarly, for a subset of arcs $A'$, we denote by $(G\backslash A',\Lambda)$ the $\Sigma$-labeled graph obtained by removing the arcs in $A'$ from $(G,\Lambda)$. For a cycle $C=\{v_1,\ldots,v_t,v_1\}$, let $\Lambda(C)=\Lambda((v_1,v_2))\cdot \ldots \cdot \Lambda((v_{t-1},v_t))\cdot \Lambda((v_t,v_1))$. Here $\cdot $ is the group operation. A cycle $C$ is called a non-null cycle if $\Lambda(C)\ne 1_{\Sigma}$. In the \probGFVS (GFVS) problem, given a $\Sigma$-labeled graph $(G,\Lambda)$ and a parameter $k$, the goal is to decide whether there is a set of vertices $D$ such that $(G\backslash D,\Lambda)$ does not have any non-null cycle. The \probGFES (GFES) problem is similar, except here the goal is to decide whether there is a set of arcs $A'$ such that $(G\backslash A',\Lambda)$ does not have any non-null cycle. For simplicity, we assume that $\Sigma=\{1,\ldots,g\}$.  

First, we apply our framework on GFVS. We note that OCT is a specialization of GFVS when $\Sigma=\{0,1\}$ \cite{KratschW20}. Indeed, our algorithm for GFVS is an extension of the algorithm for OCT. In the following, we just describe the contraction-friendly tree-decomposition DP for GFVS. The remaining algorithm is similar to the one for OCT. Following our framework, the first step is to obtain a candidate set $\mathsf{Cand} \subseteq V$. Henceforth, by the term solution we always mean a subset of $\mathsf{Cand}$.       

For describing our algorithm, we need a definition. For a $\Sigma$-labeled graph $(G,\Lambda)$, $\mu: V\rightarrow \Sigma$ is called a consistent labeling iff for each arc $(u,v)$ in $G$, we have $\mu(v)=\mu(u)\cdot \Lambda((u,v))$. 

\begin{lemma}\label{lem:consistentlabels}
\cite{Guillemot11a} A $\Sigma$-labeled graph has a consistent labeling iff it does not contain any non-null cycle. 
\end{lemma}

Consider any subgraph $G_1=(V_1,A_1)$ and fix any solution $S$ of GFVS on $G_1$. Also, consider a consistent labeling $\mu$ of the vertices in the graph $(G\backslash D,\Lambda)$, which exists by the above lemma. We define a configuration $\lambda$ of $V'\subseteq V_1$ with respect to $S$ to be a set of values corresponding to the vertices in $V'$, where each value is in $\{0,1,\ldots,g\}$. We note that these values denote whether the corresponding vertex $v$ is in the solution (0) or is mapped to the element $i\in \Sigma$ via $\mu$. As we want our solution to be a subset of $\mathsf{Cand}$, the 0-vertices must always come from $\mathsf{Cand}$. We also say that $S$ is \emph{compatible} with the configuration $\lambda$.

\begin{observation}\label{obs:distinct-config-gfvs}
The number of distinct configurations of the vertices in ${(G/X)(V')}^{-1}$ with respect to the solutions disjoint from $X$ is $g^{O(|V'|)}$. 
\end{observation} 

\begin{proof}
For each contracted vertex $v$ of $V'$, consider the connected component $C$ in $G[X]$ corresponding to $v$. Then as the solutions we seek are disjoint from $X$, $C$ must have no non-null cycles. It follows that the vertices of $C$ has a consistent labeling $\mu$. Now, fix the label of any vertex in $C$. By definition of consistent labeling, this also fixes the labeling of the remaining vertices of $C$. Thus, the number of distinct labelings of the vertices of $C$ is only $g$. Hence, the number of distinct configurations of the vertices in $C$ with respect to the solutions disjoint from $X$ is only $g$. Also, any vertex $u$ of $V'\backslash X$ can have $g+1$ configurations. Hence, the number of such distinct configurations of the vertices in ${(G/X)(V')}^{-1}$ is at most $(g+1)^{|V'|}=g^{O(|V'|)}$.  
\end{proof}

Next, we describe the contraction-friendly tree-decomposition DP for GFVS. A sub-problem is defined in the same way as for OCT. For any configuration $\lambda$ of the vertices in ${(G/X)(\sigma(t))}^{-1}$ and $0\le l\le k$, we define a sub-problem where the goal is to decide whether there is a solution to GFVS on ${(G/X)(\gamma(t))}^{-1}$ which is disjoint from $X$, is compatible with $\lambda$ and  uses exactly $l$ vertices of ${(G/X)(\gamma(t)\backslash \sigma(t))}^{-1}$. We store the boolean value (YES/NO) of the solution of this sub-problem in a table $L(t)$ indexed by $(\lambda,l)$. For simplicity, we assume wlog that all configurations are valid, i.e., for each configuration $\lambda$, 1-vertices of ${(G/X)(\sigma(t))}^{-1}$ come from $\mathsf{Cand}$. 
By Observation \ref{obs:distinct-config-gfvs}, the number of entries that we need to store in $L(t)$ is $g^{O(|\sigma(t)|)}\cdot k$, and all such relevant configurations can be computed in $g^{O(|\sigma(t)|)}n^{O(1)}$ time. With these modified definitions of sub-problem and DP table, the rest of the algorithm is similar to  the one for OCT. Hence, we conclude that the table entries corresponding to all vertices of $T$ can be computed in time $g^{O(w)}n^{O(1)}$, where $w$ is the maximum size of the bags of the non-leaf vertices in $T$. 

Note that by our framework, the total running time of our algorithm is $(|\mathsf{Cand}|^{O(\sqrt{k})}\cdot g^{O(w)})\cdot n^{O(1)}$, where $w=\sqrt{k}$. By Lemma \ref{lem-candidate}, the size of $\mathsf{Cand}$ for GFVS is $k^{O(g)}$. Subsequently, we obtain the time bound of $2^{O(g \sqrt{k}\log(gk))} \cdot n^{O(1)}$.  

\begin{theorem}
\probGFVS on $H$-minor free graphs can be solved in $2^{O(g \sqrt{k}\log(gk))} \cdot n^{O(1)}$ time by a randomized algorithm with success probability $1-1/2^n$.
\end{theorem}

By using $\mathsf{Cand} = V(G)$, we have the following.

\begin{theorem}
\probGFVS on $H$-minor free graphs can be solved in $(ng)^{O(\sqrt{k})}$ time.
\end{theorem}

\subsubsection{\probGFES}
The algorithm for GFES is similar to the one for GFVS. Here, we describe the main differences. For simplicity, we use the terms edge and arc interchangeably. Consider any subgraph $G_1=(V_1,E_1)$ and fix any solution $E_\mathsf{sol}$ on $G_1$. We say that $E_\mathsf{sol}$ is disjoint from a set of vertices $X$ if the endpoints of the edges in $E_\mathsf{sol}$ are disjoint from $X$. Consider a consistent labeling $\mu$ of the vertices in the graph $(G\backslash E_\mathsf{sol},\Lambda)$, which exists by Lemma \ref{lem:consistentlabels}. We define a configuration $\lambda$ of $V'\subseteq V_1$ with respect to $E_\mathsf{sol}$ to be a set of values corresponding to the vertices in $V'$, where each value is in $\{1,\ldots,g\}$. We note that these values denote whether the corresponding vertex $v$ is mapped to the element $i\in \Sigma$ via $\mu$. We also say that $S$ is \emph{compatible} with the configuration $\lambda$. The following observation is similar to Observation \ref{obs:distinct-config-gfvs}.

\begin{observation}\label{obs:distinct-config-gfes}
The number of distinct configurations of the vertices in ${(G/X)(V')}^{-1}$ with respect to the solutions disjoint from $X$ is $g^{O(|V'|)}$. 
\end{observation} 

Next, we describe the contraction-friendly tree-decomposition DP for GFES. For any configuration $\lambda$ of the vertices in ${(G/X)(\sigma(t))}^{-1}$ and $0\le l\le k$, we define a sub-problem where the goal is to decide whether there is a solution to GFES on ${(G/X)(\gamma(t))}^{-1}$ which is disjoint from $X$, consistent with $\lambda$ and uses exactly $l$ edges such that at least one endpoint of each edge is not in ${(G/X)(\sigma(t))}^{-1}$. We store the boolean value (YES/NO) of the solution of this sub-problem in a table $L(t)$ indexed by $(\lambda,l)$. By Observation \ref{obs:distinct-config-gfes}, the number of entries that we need to store in $L(t)$ is $g^{O(|\sigma(t)|)}\cdot k$, and all such relevant configurations can be computed in $g^{O(|\sigma(t)|)}n^{O(1)}$ time. With these modified definitions of sub-problem and DP table, the rest of the algorithm is similar to  the one for EB. Hence, we conclude that the table entries corresponding to all vertices of $T$ can be computed in time $g^{O(w)}n^{O(1)}$, where $w$ is the maximum size of the bags of the non-leaf vertices in $T$. 

Note that by our framework, the total running time of our algorithm is $(|\mathsf{Cand}|^{O(\sqrt{k})}\cdot g^{O(w)})\cdot n^{O(1)}$, where $w=\sqrt{k}$. By Lemma \ref{lem-candidate}, the size of $\mathsf{Cand}$ for GFES is $k^{O(\log^3 k)}$. Subsequently, we obtain the time bound of $2^{O(\sqrt{k} \log g \log^4 k)} \cdot n^{O(1)}$.  

\begin{theorem}
\probGFES on $H$-minor free graphs can be solved in $2^{O(\sqrt{k} \log g \log^4 k)} \cdot n^{O(1)}$ time by a randomized algorithm with success probability $1-1/2^n$.
\end{theorem}

By using $\mathsf{Cand} = E(G)$, we have the following.

\begin{theorem}
\probGFES on $H$-minor free graphs can be solved in $(ng)^{O(\sqrt{k})}$ time.
\end{theorem}


\subsection{\probVMultiC and \probEMultiC}
In \probVMultiC (VMC), given a graph $G=(V,E)$, a collection $U$ of $r$ pairs of vertices $U \subseteq V \times V$, and an integer $k$, the goal is to decide if there is a set of $k$ vertices of $V \backslash R$ whose removal separates each pair of vertices in $U$ where $R$ is the set of vertices appearing in the pairs of $U$. \probEMultiC (EMC) is similar, except here we remove a set of $k$ edges. Here $R$ is called the set of terminals. First, we consider VMC. Following our framework, the first step is to obtain a candidate set $\mathsf{Cand} \subseteq V$. Henceforth, by the term solution we always mean a subset of $\mathsf{Cand}$. Next, we describe the contraction-friendly tree-decomposition DP for VMC. Our DP is a natural extension of a DP for bounded treewidth graphs \cite{GuoHKNU08}. 

We are given an $n$-vertex graph $G=(V,E)$, a (possibly empty) subset $X\subseteq V$, and a tree-decomposition $\mathcal{T}$ of $G/X$. Let $T$ be the underlying tree of $\mathcal{T}$. We would like to solve a sub-problem in ${(G/X)(\gamma(t))}^{-1}$. Consider any partition $\Pi$ of $R$ such that each pair in $U$ appears in two different parts. We can have $|R|^{O(|R|)}=r^{O(r)}$ such distinct partitions. Note that each such partition helps us guess which terminals will be on the same connected component after removal of a solution. We  define sub-problems with respect to $\Pi$, assuming $\Pi$ is the correct configuration of the terminals. Let $\alpha\le 2r$ be the number of parts. We treat each part as a color between $1$ to $\alpha$. 

Consider any subgraph $G_1=(V_1,E_1)$ and fix any solution $S$ on $G_1$. We define a configuration $\lambda$ of $V'\subseteq V_1$ with respect to the solution to be a set of values corresponding to the vertices in $V'$, where each value is in $\{0,1,2,\ldots,\alpha\}$. We note that these values denote whether the corresponding vertex $v$ is in $S$ (0), or is in the $i$-th connected component along with the vertices of $R$ in $i$-th part ($i$). We say that $S$ is consistent with $\lambda$.  

\begin{observation}\label{obs:distinct-config-vmc}
The number of distinct configurations of the vertices in ${(G/X)(V')}^{-1}$ with respect to the solutions disjoint from $X$ is ${\alpha}^{O(|V'|)}$. 
\end{observation} 

\begin{proof}
For each contracted vertex $v$ of $V'$, consider the connected component $C$ in $G[X]$ corresponding to $v$. Then as the solutions we seek are disjoint from $X$, all vertices in $X$ have the same value in any configuration with respect to such a solution. Hence, the number of distinct configurations of the vertices in $C$ with respect to such  solutions is only $\alpha$ depending on the connected component in which they end up. Also, any vertex $u$ of $V'\backslash X$ can have at most $\alpha+1$ configurations. Hence, the number of such distinct configurations of the vertices in ${(G/X)(V')}^{-1}$ is at most ${(\alpha+1)}^{|V'|}={\alpha}^{O(|V'|)}$.  
\end{proof}

For any configuration $\lambda$ of the vertices in ${(G/X)(\sigma(t))}^{-1}$ and $0\le l\le k$, we define a sub-problem where the goal is to decide whether there is a solution to VMC on ${(G/X)(\gamma(t))}^{-1}$ which is disjoint from $X$, is consistent  with $\lambda$ and uses exactly $l$ vertices of ${(G/X)(\gamma(t)\backslash \sigma(t))}^{-1}$. We store the boolean value (YES/NO) of the solution of this sub-problem in a table $L(t)$ indexed by $(\lambda,l)$. By Observation \ref{obs:distinct-config-vmc}, the number of entries that we need to store in $L(t)$ is $r^{O(|\sigma(t)|)}\cdot k$, and all such relevant configurations can be computed in $r^{O(|\sigma(t)|)}n^{O(1)}$ time. With these modified definitions of sub-problem and DP table, the rest of the algorithm is similar to  the one in \cite{GuoHKNU08}. Hence, we conclude that the table entries corresponding to all vertices of $T$ can be computed in time $r^{O(r+w)}n^{O(1)}$, where $w$ is the maximum size of the bags of the non-leaf vertices in $T$. 

Note that by our framework, the total running time of our algorithm is $(|\mathsf{Cand}|^{O(\sqrt{k})}\cdot r^{O(r+w)})\cdot n^{O(1)}$, where $w=\sqrt{k}$. By Lemma \ref{lem-candidate}, the size of $\mathsf{Cand}$ for VMC is $k^{O(\sqrt{r})}$. Subsequently, we obtain the time bound of $2^{O(r\sqrt{k} \log k)} \cdot n^{O(1)}$.  

\begin{theorem}
\probVMultiC on $H$-minor free graphs can be solved in $2^{O((\sqrt{rk}+r) \log rk))} \cdot n^{O(1)}$ time by a randomized algorithm with success probability $1-1/2^n$. 
\end{theorem}

By using $\mathsf{Cand} = V(G)$, we have the following.

\begin{theorem}
\probVMultiC on $H$-minor free graphs can be solved in 
$n^{O(r+\sqrt{k})}$ time. 
\end{theorem}


As mentioned in \cite{GuoHKNU08}, the DP for VMC is almost the same as the DP for EMC. Similarly, the contraction-friendly tree-decomposition DP for EMC is very similar to the one for VMC. Also, the candidate set size for EMC is $(r+k)^{O(\log^3(r+k))}$. Hence, we obtain the following theorem. 

\begin{theorem}
\probEMultiC on $H$-minor free graphs can be solved in $2^{O((r+\sqrt{k})\log^4 (rk))} \cdot n^{O(1)}$ 
time by a randomized algorithm with success probability $1-1/2^n$. 
\end{theorem}


\begin{theorem}
\probEMultiC on $H$-minor free graphs can be solved in 
$n^{O(r+\sqrt{k})}$ time. 
\end{theorem}

\subsection{\probVMC and \probEMC}
In \probVMC (VMWC), given a graph $G$, a subset of $r$ vertices $R \subseteq V (G)$, called terminals, and an integer $k$, the goal is to decide if there is a set of at most $k$ vertices of $V(G) \backslash R$ hitting every path between a pair of terminals. \probEMC (EMWC) is similar, except the goal is to decide if there is a set of at most $k$ edges of $E(G)$ hitting every path between a pair of terminals. 

Note that VMWC is a special case of VMC where we have to separate all $r^2$ pairs of terminals in $R\times R$. Also, we need to consider only one partition where each terminal belongs to a separate part. Thus, we readily obtain a contraction free DP that runs in $r^{O(w)}n^{O(1)}$ time.  

By our framework, the total running time of our algorithm is $(|\mathsf{Cand}|^{O(\sqrt{k})}\cdot r^{O(w)})\cdot n^{O(1)}$, where $w=\sqrt{k}$. By Lemma \ref{lem-candidate}, the size of $\mathsf{Cand}$ for VMWC is $k^{O(r)}$. Subsequently, we obtain the time bound of $2^{O(r\sqrt{k} \log k)} \cdot n^{O(1)}$.  

\begin{theorem}
\probVMC on $H$-minor free graphs can be solved in $2^{O(r\sqrt{k} \log k)} \cdot n^{O(1)}$ time by a randomized algorithm with success probability $1-1/2^n$. 
\end{theorem}

By using $\mathsf{Cand} = V(G)$, we have the following.

\begin{theorem}
\probVMC on $H$-minor free graphs can be solved in $n^{O(\sqrt{k})}$ time. 
\end{theorem}

\subsubsection{\probEMC}
We describe the contraction-friendly tree-decomposition DP for EMWC. This DP is partly motivated by the DP in \cite{DengLZ13} for Multi-Multiway cut which is a generalization of EMWC. We are given an $n$-vertex graph $G=(V,E)$, a (possibly empty) subset $X\subseteq V$, and a tree-decomposition $\mathcal{T}$ of $G/X$. Let $T$ be the underlying tree of $\mathcal{T}$. Consider any non-leaf vertex $t \in T$. Note that removal of the solution edges in EMWC induces a partition of the vertices where each part (or connected component) contains at most one terminal of $R$. We would like to guess an optimal partition. Consider any subset $V' \subseteq \beta(t)$, and any partition $\Pi=\{X_i\mid i\in I\}$ of the vertices in ${(G/X)(V')}^{-1}$ for some index set $I$. We say a set of edges $C$ in ${(G/X)(\gamma(t))}^{-1}$ is consistent with $\Pi$ and vice versa iff 
\begin{itemize}
    \item The endpoints of the edges in $C$ are disjoint from $X$;
    \item $C$ separates all the terminals; and
    \item Let $G'$ be the subgraph ${(G/X)(\gamma(t))}^{-1}$ after removing the edges in $C$. Let $\{P_i\mid i\in I\}$ be the family of connected components such that $P_i\cap V'\ne \emptyset$ for all $i\in I$. Then $X_i=P_i$ for all $i\in I$. 
\end{itemize}

\begin{observation}\label{obs:distinct-config-emwc}
The number of distinct partitions of the vertices in ${(G/X)(V')}^{-1}$ consistent with sets of edges is ${|V'|}^{O(|V'|)}$.
\end{observation} 

\begin{proof}
First note that the number of distinct partitions of the vertices in $V'$ is ${|V'|}^{O(|V'|)}$. However, $(G/X)$ $(V'){}^{-1}$ might contain many more vertices, as we undo the contraction of the component vertices in $V'$. For each contracted vertex $v$ of $V'$, consider the connected component $C$ in $G[X]$ corresponding to $v$. Then as the endpoints of the edges in consistent sets $C$ are disjoint from $X$, all vertices in $X$ must end up in the same connected component after removal of the edges from $C$. It follows that the number of such distinct partitions of the vertices in ${(G/X)(V')}^{-1}$ is again bounded by ${|V'|}^{O(|V'|)}$. 
\end{proof}

For any partition $\Pi$ of the vertices in ${(G/X)(\sigma(t))}^{-1}$, we define a sub-problem where the goal is to find the minimum size set of edges $C$ in ${(G/X)(\gamma(t))}^{-1}$ which is consistent with $\Pi$. It is sufficient to solve this sub-problem for our purpose, as if we know the minimum value, we can also solve the decision version of the problem. We store a solution $C$ of this sub-problem in a table $L(t)$ indexed by $\Pi$. To compute the table corresponding to $t$, we need to merge the tables of its children. In the following, we describe how to merge the tables of two children. One can easily extend it to more than two children.  

Consider any partition $\Pi_1$ of the vertices in ${(G/X)(\sigma(t))}^{-1}$. We extend this partition to obtain a partition $\Pi_2$ of the vertices in ${(G/X)(\beta(t))}^{-1}$. For the purpose of extension, we define compatibility of two partitions. Consider two partitions  $\Pi_1=\{X_i\mid i\in I\}$ and $\Pi_2=\{Y_j\mid j\in J\}$. Also, consider the bipartite graph $B=(I,J,E')$ such that $(i,j)\in E'$ iff $X_i\cap Y_j\ne \emptyset$. Then $\Pi_1$ and $\Pi_2$ are compatible if $E'$ is a perfect matching and for each edge $(i,j)\in E'$, $i=j$. 

Now, fix a partition $\Pi$ of the vertices in ${(G/X)(\sigma(t))}^{-1}$. We show how to compute the entry $L(t)[\Pi]$ by merging tables of two children $s_1$ and $s_2$. First, we take a partition $\Pi'$ of the vertices in ${(G/X)(\beta(t))}^{-1}$ which is compatible with $\Pi$. Given $\Pi'$, the partition of the vertices in  ${(G/X)(\sigma(s_1))}^{-1}$ (or ${(G/X)(\sigma(s_2))}^{-1}$) compatible with $\Pi'$, is also fixed. Let $\Pi_1$ and $\Pi_2$ be these two partitions with respect to $s_1$ and $s_2$. The entry $L(t)[\Pi]=\arg \min_{|C|} \{C=L(s_1)[\Pi_1]\cup L(s_2)[\Pi_2]\cup C'\}$, where $C'$ is a subset of edges $\{(u,v)\mid u,v \notin X\}$ such that $(u,v)$ is in ${(G/X)(\beta(t))}^{-1}$, but not in ${(G/X)(\gamma(s_1))}^{-1}$ or ${(G/X)(\gamma(s_2))}^{-1}$, and $u$ and $v$ are contained in different parts of $\Pi'$.   

The correctness follows from the fact that if an edge $(u,v)$ is common to both parent and child sub-problems, the decision to take it or not is uniquely determined by the partition corresponding to the children. Thus, no edge is counted more than once in the computed solution.  

Note that computation of a table entry requires enumerating all possible $\Pi'$. The number of such partitions is bounded by $w^{O(w)}$ by Observation \ref{obs:distinct-config-emwc}, where $w$ is the maximum size of the bags of the non-leaf vertices in $T$. By the same observation, the number of entries that we need to store in $L(t)$ is ${|\sigma(t)|}^{O(|\sigma(t)|)}\cdot k$, and all such relevant partitions can be computed in ${|\sigma(t)|}^{O(|\sigma(t)|)}n^{O(1)}$ time. Hence, we conclude that the table entries corresponding to all vertices of $T$ can be computed in time $w^{O(w)} n^{O(1)}$. 

By our framework, the total running time of our algorithm is $(|\mathsf{Cand}|^{O(\sqrt{k})}\cdot w^{O(w)})\cdot n^{O(1)}$, where $w=\sqrt{k}$. By Lemma \ref{lem-candidate}, the size of $\mathsf{Cand}$ for EMWC is $k^{O(\log^3 k)}$. Subsequently, we obtain the time bound of $2^{O(\sqrt{k}\log^4 k)} \cdot n^{O(1)}$.  

\begin{theorem}
\probEMC on $H$-minor free graphs can be solved in $2^{O(\sqrt{k}\log^4 k)} \cdot n^{O(1)}$ time by a randomized algorithm with success probability $1-1/2^n$. 
\end{theorem}

By using $\mathsf{Cand} = E(G)$, we have the following.

\begin{theorem}
\probEMC on $H$-minor free graphs can be solved in $n^{O(\sqrt{k})}$ time. 
\end{theorem}


\section{Conclusion and future work} \label{sec-conclusion}
In this paper, we presented a generic framework to design subexponential-time parameterized algorithms for various cut and cycle hitting problems on $H$-minor-free graphs, including \probOCT, \probEB, \probVMC, \probEMC, \probVMultiC, \probEMultiC, \probGFVS, and \probGFES.
Our framework is robust and can possibly be applied to other problems.
For example, it can lead to an $n^{O(\sqrt{k})}$-time algorithm for the classical problem \textsc{Bisection}, because \textsc{Bisection} is contraction-friendly and there is a known $2^{O(w)}$-time DP algorithm on a tree decomposition of width $w$.

Our framework is based on a new decomposition theorem on almost-embeddable graphs, which states that for every $h$-almost-embeddable graph $G$ and every $p \in \mathbb{N}$, there exist disjoint sets $Z_1,\dots,Z_p \subseteq V(G)$ such that for every $i \in [p]$ and every $Z' \subseteq Z_i$, $\mathbf{tw}(G / (Z_i \backslash Z')) = O(p+|Z'|)$.
This generalizes the classical contraction decomposition theorem on planar graphs.
We believe this decomposition theorem is quite powerful and should find applications beyond this paper.

Next, we pose some open questions for future study.
The first interesting question is whether our decomposition theorem above generalizes to $H$-minor-free graphs.
In this work, we were only able to prove the decomposition theorem for almost-embeddable graphs, and hence required a complicated two-layer DP procedure in order to achieve the algorithmic results on $H$-minor-free graphs.
If the same theorem holds even for $H$-minor-free graphs, all problems studied in this paper can be solved much more directly.
Furthermore, thanks to a simultaneous work by Marx et al. \cite{marx2022framework} which showed a similar decomposition theorem for planar graphs (among many other results), such a decomposition theorem for $H$-minor-free graphs can also result in subexponential parameterized algorithms for other important cycle hitting problems on $H$-minor-free graphs, such as \textsc{Subset Feedback Vertex Set}, \textsc{Subset Odd Cycle Transversal}, etc.
Another open question is whether one can obtain subexponential parameterized algorithms for the problems studied in the paper with dependence only on $k$.
Note that whereas the running time of our algorithms are subexponential in $k$, they also depend on additional parameters such as $r$ and $g$, for problems like \probVMC, \probVMultiC, \probEMultiC, \probGFVS, and \probGFES.
Removing the dependency on these parameters (or at least making the dependency subexponential) seems to be an interesting direction for future study.
Finally, in a very recent follow-up work of this paper \cite{bandyapadhyay2022true}, the authors showed that unit-disk graphs admit a decomposition theorem similar to the one proved in this paper.
One can then ask whether theorems of this type hold for other classes of geometric intersection graphs, such as disk graphs, unit-ball graphs, string graphs, etc.

\bibliographystyle{plainurl}
\bibliography{my_bib.bib}

\begin{thebibliography}{10}

\bibitem{Baker94}
Brenda~S. Baker.
\newblock Approximation algorithms for np-complete problems on planar graphs.
\newblock {\em J. {ACM}}, 41(1):153--180, 1994.
\newblock URL: \url{http://doi.acm.org/10.1145/174644.174650}, \href
  {https://doi.org/10.1145/174644.174650} {\path{doi:10.1145/174644.174650}}.

\bibitem{bandyapadhyay2022true}
Sayan Bandyapadhyay, William Lochet, Daniel Lokshtanov, Saket Saurabh, and Jie
  Xue.
\newblock True contraction decomposition and almost eth-tight bipartization for
  unit-disk graphs.
\newblock In {\em 38th International Symposium on Computational Geometry (SoCG
  2022)}. Schloss Dagstuhl-Leibniz-Zentrum f{\"u}r Informatik, 2022.

\bibitem{BousquetDT18}
Nicolas Bousquet, Jean Daligault, and St{\'{e}}phan Thomass{\'{e}}.
\newblock Multicut is {FPT}.
\newblock {\em {SIAM} J. Comput.}, 47(1):166--207, 2018.
\newblock \href {https://doi.org/10.1137/140961808}
  {\path{doi:10.1137/140961808}}.

\bibitem{BuiP92}
Thang~Nguyen Bui and Andrew Peck.
\newblock Partitioning planar graphs.
\newblock {\em {SIAM} J. Comput.}, 21(2):203--215, 1992.
\newblock \href {https://doi.org/10.1137/0221016} {\path{doi:10.1137/0221016}}.

\bibitem{ChenLLOR08}
Jianer Chen, Yang Liu, Songjian Lu, Barry O'Sullivan, and Igor Razgon.
\newblock A fixed-parameter algorithm for the directed feedback vertex set
  problem.
\newblock {\em J. {ACM}}, 55(5):21:1--21:19, 2008.
\newblock \href {https://doi.org/10.1145/1411509.1411511}
  {\path{doi:10.1145/1411509.1411511}}.

\bibitem{ChitnisCHPP16}
Rajesh Chitnis, Marek Cygan, MohammadTaghi Hajiaghayi, Marcin Pilipczuk, and
  Michal Pilipczuk.
\newblock Designing {FPT} algorithms for cut problems using randomized
  contractions.
\newblock {\em {SIAM} J. Comput.}, 45(4):1171--1229, 2016.
\newblock \href {https://doi.org/10.1137/15M1032077}
  {\path{doi:10.1137/15M1032077}}.

\bibitem{ChitnisCHM15}
Rajesh~Hemant Chitnis, Marek Cygan, Mohammad~Taghi Hajiaghayi, and D{\'{a}}niel
  Marx.
\newblock Directed subset feedback vertex set is fixed-parameter tractable.
\newblock {\em {ACM} Trans. Algorithms}, 11(4):28:1--28:28, 2015.
\newblock \href {https://doi.org/10.1145/2700209} {\path{doi:10.1145/2700209}}.

\bibitem{ChitnisHM13}
Rajesh~Hemant Chitnis, MohammadTaghi Hajiaghayi, and D{\'{a}}niel Marx.
\newblock Fixed-parameter tractability of directed multiway cut parameterized
  by the size of the cutset.
\newblock {\em {SIAM} J. Comput.}, 42(4):1674--1696, 2013.
\newblock \href {https://doi.org/10.1137/12086217X}
  {\path{doi:10.1137/12086217X}}.

\bibitem{Cohen-AddadVM21}
Vincent Cohen{-}Addad, {\'{E}}ric~Colin de~Verdi{\`{e}}re, and Arnaud
  de~Mesmay.
\newblock A near-linear approximation scheme for multicuts of embedded graphs
  with a fixed number of terminals.
\newblock {\em {SIAM} J. Comput.}, 50(1):1--31, 2021.
\newblock \href {https://doi.org/10.1137/18M1183297}
  {\path{doi:10.1137/18M1183297}}.

\bibitem{Cohen-AddadVMM19}
Vincent Cohen{-}Addad, {\'{E}}ric~Colin de~Verdi{\`{e}}re, D{\'{a}}niel Marx,
  and Arnaud de~Mesmay.
\newblock Almost tight lower bounds for hard cutting problems in embedded
  graphs.
\newblock In Gill Barequet and Yusu Wang, editors, {\em 35th International
  Symposium on Computational Geometry, SoCG 2019, June 18-21, 2019, Portland,
  Oregon, {USA}}, volume 129 of {\em LIPIcs}, pages 27:1--27:16. Schloss
  Dagstuhl - Leibniz-Zentrum f{\"{u}}r Informatik, 2019.
\newblock \href {https://doi.org/10.4230/LIPIcs.SoCG.2019.27}
  {\path{doi:10.4230/LIPIcs.SoCG.2019.27}}.

\bibitem{CyganFKLMPPS15}
Marek Cygan, Fedor~V. Fomin, Lukasz Kowalik, Daniel Lokshtanov, D{\'{a}}niel
  Marx, Marcin Pilipczuk, Michal Pilipczuk, and Saket Saurabh.
\newblock {\em Parameterized Algorithms}.
\newblock Springer, 2015.
\newblock \href {https://doi.org/10.1007/978-3-319-21275-3}
  {\path{doi:10.1007/978-3-319-21275-3}}.

\bibitem{CyganKLPPSW21}
Marek Cygan, Pawel Komosa, Daniel Lokshtanov, Marcin Pilipczuk, Michal
  Pilipczuk, Saket Saurabh, and Magnus Wahlstr{\"{o}}m.
\newblock Randomized contractions meet lean decompositions.
\newblock {\em {ACM} Trans. Algorithms}, 17(1):6:1--6:30, 2021.
\newblock \href {https://doi.org/10.1145/3426738} {\path{doi:10.1145/3426738}}.

\bibitem{CyganPPW13}
Marek Cygan, Marcin Pilipczuk, Michal Pilipczuk, and Jakub~Onufry Wojtaszczyk.
\newblock Subset feedback vertex set is fixed-parameter tractable.
\newblock {\em {SIAM} J. Discret. Math.}, 27(1):290--309, 2013.
\newblock \href {https://doi.org/10.1137/110843071}
  {\path{doi:10.1137/110843071}}.

\bibitem{DahlhausJPSY94}
Elias Dahlhaus, David~S. Johnson, Christos~H. Papadimitriou, Paul~D. Seymour,
  and Mihalis Yannakakis.
\newblock The complexity of multiterminal cuts.
\newblock {\em {SIAM} J. Comput.}, 23(4):864--894, 1994.
\newblock \href {https://doi.org/10.1137/S0097539792225297}
  {\path{doi:10.1137/S0097539792225297}}.

\bibitem{Verdiere17}
{\'{E}}ric~Colin de~Verdi{\`{e}}re.
\newblock Multicuts in planar and bounded-genus graphs with bounded number of
  terminals.
\newblock {\em Algorithmica}, 78(4):1206--1224, 2017.
\newblock \href {https://doi.org/10.1007/s00453-016-0258-0}
  {\path{doi:10.1007/s00453-016-0258-0}}.

\bibitem{demaine2005subexponential}
Erik~D. Demaine, Fedor~V. Fomin, Mohammad~Taghi Hajiaghayi, and Dimitrios~M.
  Thilikos.
\newblock Subexponential parameterized algorithms on bounded-genus graphs and
  \emph{H}-minor-free graphs.
\newblock {\em J. {ACM}}, 52(6):866--893, 2005.
\newblock \href {https://doi.org/10.1145/1101821.1101823}
  {\path{doi:10.1145/1101821.1101823}}.

\bibitem{DemaineHK05}
Erik~D. Demaine, Mohammad~Taghi Hajiaghayi, and Ken{-}ichi Kawarabayashi.
\newblock Algorithmic graph minor theory: Decomposition, approximation, and
  coloring.
\newblock In {\em 46th Annual {IEEE} Symposium on Foundations of Computer
  Science {(FOCS} 2005), 23-25 October 2005, Pittsburgh, PA, USA, Proceedings},
  pages 637--646, 2005.
\newblock \href {https://doi.org/10.1109/SFCS.2005.14}
  {\path{doi:10.1109/SFCS.2005.14}}.

\bibitem{demaine2005algorithmic}
Erik~D Demaine, Mohammad~Taghi Hajiaghayi, and Ken-ichi Kawarabayashi.
\newblock Algorithmic graph minor theory: Decomposition, approximation, and
  coloring.
\newblock In {\em 46th Annual IEEE Symposium on Foundations of Computer Science
  (FOCS'05)}, pages 637--646. IEEE, 2005.

\bibitem{DemaineHK11}
Erik~D. Demaine, MohammadTaghi Hajiaghayi, and Ken{-}ichi Kawarabayashi.
\newblock Contraction decomposition in $h$-minor-free graphs and algorithmic
  applications.
\newblock In {\em Proceedings of the 43rd {ACM} Symposium on Theory of
  Computing, {STOC} 2011, San Jose, CA, USA, 6-8 June 2011}, pages 441--450,
  2011.

\bibitem{DemaineHM10}
Erik~D. Demaine, MohammadTaghi Hajiaghayi, and Bojan Mohar.
\newblock Approximation algorithms via contraction decomposition.
\newblock {\em Combinatorica}, 30(5):533--552, 2010.

\bibitem{DengLZ13}
Xiaojie Deng, Bingkai Lin, and Chihao Zhang.
\newblock Multi-multiway cut problem on graphs of bounded branch width.
\newblock In Michael~R. Fellows, Xuehou Tan, and Binhai Zhu, editors, {\em
  Frontiers in Algorithmics \emph{and} Algorithmic Aspects in Information and
  Management, Third Joint International Conference, {FAW-AAIM} 2013, Dalian,
  China, June 26-28, 2013. Proceedings}, volume 7924 of {\em Lecture Notes in
  Computer Science}, pages 315--324. Springer, 2013.
\newblock \href {https://doi.org/10.1007/978-3-642-38756-2\_32}
  {\path{doi:10.1007/978-3-642-38756-2\_32}}.

\bibitem{DeVosDOSRSV04}
Matt DeVos, Guoli Ding, Bogdan Oporowski, Daniel~P. Sanders, Bruce~A. Reed,
  Paul~D. Seymour, and Dirk Vertigan.
\newblock Excluding any graph as a minor allows a low tree-width 2-coloring.
\newblock {\em J. Comb. Theory, Ser. {B}}, 91(1):25--41, 2004.
\newblock \href {https://doi.org/10.1016/j.jctb.2003.09.001}
  {\path{doi:10.1016/j.jctb.2003.09.001}}.

\bibitem{DornFLRS13}
Frederic Dorn, Fedor~V. Fomin, Daniel Lokshtanov, Venkatesh Raman, and Saket
  Saurabh.
\newblock Beyond bidimensionality: Parameterized subexponential algorithms on
  directed graphs.
\newblock {\em Inf. Comput.}, 233:60--70, 2013.
\newblock \href {https://doi.org/10.1016/j.ic.2013.11.006}
  {\path{doi:10.1016/j.ic.2013.11.006}}.

\bibitem{DowneyF99}
Rodney~G. Downey and Michael~R. Fellows.
\newblock {\em Parameterized Complexity}.
\newblock Monographs in Computer Science. Springer, 1999.
\newblock \href {https://doi.org/10.1007/978-1-4612-0515-9}
  {\path{doi:10.1007/978-1-4612-0515-9}}.

\bibitem{DBLP:conf/soda/Dvorak18}
Zdenek Dvor{\'{a}}k.
\newblock Thin graph classes and polynomial-time approximation schemes.
\newblock In {\em Proceedings of the Twenty-Ninth Annual {ACM-SIAM} Symposium
  on Discrete Algorithms, {SODA} 2018, New Orleans, LA, USA, January 7-10,
  2018}, pages 1685--1701, 2018.
\newblock \href {https://doi.org/10.1137/1.9781611975031.110}
  {\path{doi:10.1137/1.9781611975031.110}}.

\bibitem{eppstein2000diameter}
David Eppstein.
\newblock Diameter and treewidth in minor-closed graph families.
\newblock {\em Algorithmica}, 27(3):275--291, 2000.
\newblock \href {https://doi.org/10.1007/s004530010020}
  {\path{doi:10.1007/s004530010020}}.

\bibitem{FominLKPS20}
Fedor~V. Fomin, Daniel Lokshtanov, Sudeshna Kolay, Fahad Panolan, and Saket
  Saurabh.
\newblock Subexponential algorithms for rectilinear steiner tree and
  arborescence problems.
\newblock {\em {ACM} Trans. Algorithms}, 16(2):21:1--21:37, 2020.
\newblock \href {https://doi.org/10.1145/3381420} {\path{doi:10.1145/3381420}}.

\bibitem{FominLMPPS16}
Fedor~V. Fomin, Daniel Lokshtanov, D{\'{a}}niel Marx, Marcin Pilipczuk, Michal
  Pilipczuk, and Saket Saurabh.
\newblock Subexponential parameterized algorithms for planar and
  apex-minor-free graphs via low treewidth pattern covering.
\newblock In Irit Dinur, editor, {\em {IEEE} 57th Annual Symposium on
  Foundations of Computer Science, {FOCS} 2016, 9-11 October 2016, Hyatt
  Regency, New Brunswick, New Jersey, {USA}}, pages 515--524. {IEEE} Computer
  Society, 2016.
\newblock \href {https://doi.org/10.1109/FOCS.2016.62}
  {\path{doi:10.1109/FOCS.2016.62}}.

\bibitem{fomin2019kernelization}
Fedor~V Fomin, Daniel Lokshtanov, Saket Saurabh, and Meirav Zehavi.
\newblock {\em Kernelization: theory of parameterized preprocessing}.
\newblock Cambridge University Press, 2019.

\bibitem{GokeMM20}
Alexander G{\"{o}}ke, D{\'{a}}niel Marx, and Matthias Mnich.
\newblock Hitting long directed cycles is fixed-parameter tractable.
\newblock In Artur Czumaj, Anuj Dawar, and Emanuela Merelli, editors, {\em 47th
  International Colloquium on Automata, Languages, and Programming, {ICALP}
  2020, July 8-11, 2020, Saarbr{\"{u}}cken, Germany (Virtual Conference)},
  volume 168 of {\em LIPIcs}, pages 59:1--59:18. Schloss Dagstuhl -
  Leibniz-Zentrum f{\"{u}}r Informatik, 2020.
\newblock \href {https://doi.org/10.4230/LIPIcs.ICALP.2020.59}
  {\path{doi:10.4230/LIPIcs.ICALP.2020.59}}.

\bibitem{grohe2013simple}
Martin Grohe, Ken{-}ichi Kawarabayashi, and Bruce~A. Reed.
\newblock A simple algorithm for the graph minor decomposition - logic meets
  structural graph theory.
\newblock In Sanjeev Khanna, editor, {\em Proceedings of the Twenty-Fourth
  Annual {ACM-SIAM} Symposium on Discrete Algorithms, {SODA} 2013, New Orleans,
  Louisiana, USA, January 6-8, 2013}, pages 414--431. {SIAM}, 2013.
\newblock \href {https://doi.org/10.1137/1.9781611973105.30}
  {\path{doi:10.1137/1.9781611973105.30}}.

\bibitem{Guillemot11a}
Sylvain Guillemot.
\newblock {FPT} algorithms for path-transversal and cycle-transversal problems.
\newblock {\em Discret. Optim.}, 8(1):61--71, 2011.
\newblock \href {https://doi.org/10.1016/j.disopt.2010.05.003}
  {\path{doi:10.1016/j.disopt.2010.05.003}}.

\bibitem{GuoHKNU08}
Jiong Guo, Falk H{\"{u}}ffner, Erhan Kenar, Rolf Niedermeier, and Johannes
  Uhlmann.
\newblock Complexity and exact algorithms for vertex multicut in interval and
  bounded treewidth graphs.
\newblock {\em Eur. J. Oper. Res.}, 186(2):542--553, 2008.
\newblock \href {https://doi.org/10.1016/j.ejor.2007.02.014}
  {\path{doi:10.1016/j.ejor.2007.02.014}}.

\bibitem{HolsK18}
Eva{-}Maria~C. Hols and Stefan Kratsch.
\newblock A randomized polynomial kernel for subset feedback vertex set.
\newblock {\em Theory Comput. Syst.}, 62(1):63--92, 2018.
\newblock \href {https://doi.org/10.1007/s00224-017-9805-6}
  {\path{doi:10.1007/s00224-017-9805-6}}.

\bibitem{JansenPL19}
Bart M.~P. Jansen, Marcin Pilipczuk, and Erik~Jan van Leeuwen.
\newblock A deterministic polynomial kernel for odd cycle transversal and
  vertex multiway cut in planar graphs.
\newblock In Rolf Niedermeier and Christophe Paul, editors, {\em 36th
  International Symposium on Theoretical Aspects of Computer Science, {STACS}
  2019, March 13-16, 2019, Berlin, Germany}, volume 126 of {\em LIPIcs}, pages
  39:1--39:18. Schloss Dagstuhl - Leibniz-Zentrum f{\"{u}}r Informatik, 2019.
\newblock \href {https://doi.org/10.4230/LIPIcs.STACS.2019.39}
  {\path{doi:10.4230/LIPIcs.STACS.2019.39}}.

\bibitem{kawarabayashi2011simpler}
Ken-ichi Kawarabayashi and Paul Wollan.
\newblock A simpler algorithm and shorter proof for the graph minor
  decomposition.
\newblock In {\em Proceedings of the forty-third annual ACM symposium on Theory
  of computing}, pages 451--458, 2011.

\bibitem{KimKPW21}
Eun~Jung Kim, Stefan Kratsch, Marcin Pilipczuk, and Magnus Wahlstr{\"{o}}m.
\newblock Solving hard cut problems via flow-augmentation.
\newblock In D{\'{a}}niel Marx, editor, {\em Proceedings of the 2021 {ACM-SIAM}
  Symposium on Discrete Algorithms, {SODA} 2021, Virtual Conference, January 10
  - 13, 2021}, pages 149--168. {SIAM}, 2021.
\newblock \href {https://doi.org/10.1137/1.9781611976465.11}
  {\path{doi:10.1137/1.9781611976465.11}}.

\bibitem{Klein06}
Philip~N. Klein.
\newblock A subset spanner for planar graphs, : with application to subset
  {TSP}.
\newblock In {\em Proceedings of the 38th Annual {ACM} Symposium on Theory of
  Computing, Seattle, WA, USA, May 21-23, 2006}, pages 749--756, 2006.
\newblock URL: \url{http://doi.acm.org/10.1145/1132516.1132620}, \href
  {https://doi.org/10.1145/1132516.1132620}
  {\path{doi:10.1145/1132516.1132620}}.

\bibitem{Klein08}
Philip~N. Klein.
\newblock A linear-time approximation scheme for {TSP} in undirected planar
  graphs with edge-weights.
\newblock {\em {SIAM} J. Comput.}, 37(6):1926--1952, 2008.
\newblock \href {https://doi.org/10.1137/060649562}
  {\path{doi:10.1137/060649562}}.

\bibitem{KleinM12}
Philip~N. Klein and D{\'{a}}niel Marx.
\newblock Solving planar k -terminal cut in {$O(n^{c\sqrt{k}})$} time.
\newblock In Artur Czumaj, Kurt Mehlhorn, Andrew~M. Pitts, and Roger
  Wattenhofer, editors, {\em Automata, Languages, and Programming - 39th
  International Colloquium, {ICALP} 2012, Warwick, UK, July 9-13, 2012,
  Proceedings, Part {I}}, volume 7391 of {\em Lecture Notes in Computer
  Science}, pages 569--580. Springer, 2012.
\newblock \href {https://doi.org/10.1007/978-3-642-31594-7\_48}
  {\path{doi:10.1007/978-3-642-31594-7\_48}}.

\bibitem{KleinM14}
Philip~N. Klein and D{\'{a}}niel Marx.
\newblock A subexponential parameterized algorithm for subset {TSP} on planar
  graphs.
\newblock In Chandra Chekuri, editor, {\em Proceedings of the Twenty-Fifth
  Annual {ACM-SIAM} Symposium on Discrete Algorithms, {SODA} 2014, Portland,
  Oregon, USA, January 5-7, 2014}, pages 1812--1830. {SIAM}, 2014.
\newblock \href {https://doi.org/10.1137/1.9781611973402.131}
  {\path{doi:10.1137/1.9781611973402.131}}.

\bibitem{KratschPPW15}
Stefan Kratsch, Marcin Pilipczuk, Michal Pilipczuk, and Magnus Wahlstr{\"{o}}m.
\newblock Fixed-parameter tractability of multicut in directed acyclic graphs.
\newblock {\em {SIAM} J. Discret. Math.}, 29(1):122--144, 2015.
\newblock \href {https://doi.org/10.1137/120904202}
  {\path{doi:10.1137/120904202}}.

\bibitem{KratschW14}
Stefan Kratsch and Magnus Wahlstr{\"{o}}m.
\newblock Compression via matroids: {A} randomized polynomial kernel for odd
  cycle transversal.
\newblock {\em {ACM} Trans. Algorithms}, 10(4):20:1--20:15, 2014.
\newblock \href {https://doi.org/10.1145/2635810} {\path{doi:10.1145/2635810}}.

\bibitem{KratschW20}
Stefan Kratsch and Magnus Wahlstr{\"{o}}m.
\newblock Representative sets and irrelevant vertices: New tools for
  kernelization.
\newblock {\em J. {ACM}}, 67(3):16:1--16:50, 2020.
\newblock \href {https://doi.org/10.1145/3390887} {\path{doi:10.1145/3390887}}.

\bibitem{LokshtanovMRS17}
Daniel Lokshtanov, Pranabendu Misra, M.~S. Ramanujan, and Saket Saurabh.
\newblock Hitting selected (odd) cycles.
\newblock {\em {SIAM} J. Discret. Math.}, 31(3):1581--1615, 2017.
\newblock \href {https://doi.org/10.1137/15M1041213}
  {\path{doi:10.1137/15M1041213}}.

\bibitem{LokshtanovPSSZ20}
Daniel Lokshtanov, Fahad Panolan, Saket Saurabh, Roohani Sharma, and Meirav
  Zehavi.
\newblock Covering small independent sets and separators with applications to
  parameterized algorithms.
\newblock {\em {ACM} Trans. Algorithms}, 16(3):32:1--32:31, 2020.
\newblock \href {https://doi.org/10.1145/3379698} {\path{doi:10.1145/3379698}}.

\bibitem{LokshtanovSW12}
Daniel Lokshtanov, Saket Saurabh, and Magnus Wahlstr{\"{o}}m.
\newblock Subexponential parameterized odd cycle transversal on planar graphs.
\newblock In Deepak D'Souza, Telikepalli Kavitha, and Jaikumar Radhakrishnan,
  editors, {\em {IARCS} Annual Conference on Foundations of Software Technology
  and Theoretical Computer Science, {FSTTCS} 2012, December 15-17, 2012,
  Hyderabad, India}, volume~18 of {\em LIPIcs}, pages 424--434. Schloss
  Dagstuhl - Leibniz-Zentrum f{\"{u}}r Informatik, 2012.
\newblock \href {https://doi.org/10.4230/LIPIcs.FSTTCS.2012.424}
  {\path{doi:10.4230/LIPIcs.FSTTCS.2012.424}}.

\bibitem{Marx06}
D{\'{a}}niel Marx.
\newblock Parameterized graph separation problems.
\newblock {\em Theor. Comput. Sci.}, 351(3):394--406, 2006.
\newblock \href {https://doi.org/10.1016/j.tcs.2005.10.007}
  {\path{doi:10.1016/j.tcs.2005.10.007}}.

\bibitem{Marx12}
D{\'{a}}niel Marx.
\newblock A tight lower bound for planar multiway cut with fixed number of
  terminals.
\newblock In Artur Czumaj, Kurt Mehlhorn, Andrew~M. Pitts, and Roger
  Wattenhofer, editors, {\em Automata, Languages, and Programming - 39th
  International Colloquium, {ICALP} 2012, Warwick, UK, July 9-13, 2012,
  Proceedings, Part {I}}, volume 7391 of {\em Lecture Notes in Computer
  Science}, pages 677--688. Springer, 2012.
\newblock \href {https://doi.org/10.1007/978-3-642-31594-7\_57}
  {\path{doi:10.1007/978-3-642-31594-7\_57}}.

\bibitem{marx2022framework}
D{\'a}niel Marx, Pranabendu Misra, Daniel Neuen, and Prafullkumar Tale.
\newblock A framework for parameterized subexponential algorithms for
  generalized cycle hitting problems on planar graphs.
\newblock In {\em Proceedings of the 2022 Annual ACM-SIAM Symposium on Discrete
  Algorithms (SODA)}, pages 2085--2127. SIAM, 2022.

\bibitem{MarxOR13}
D{\'{a}}niel Marx, Barry O'Sullivan, and Igor Razgon.
\newblock Finding small separators in linear time via treewidth reduction.
\newblock {\em {ACM} Trans. Algorithms}, 9(4):30:1--30:35, 2013.
\newblock \href {https://doi.org/10.1145/2500119} {\path{doi:10.1145/2500119}}.

\bibitem{MarxPP18}
D{\'{a}}niel Marx, Marcin Pilipczuk, and Michal Pilipczuk.
\newblock On subexponential parameterized algorithms for steiner tree and
  directed subset {TSP} on planar graphs.
\newblock In Mikkel Thorup, editor, {\em 59th {IEEE} Annual Symposium on
  Foundations of Computer Science, {FOCS} 2018, Paris, France, October 7-9,
  2018}, pages 474--484. {IEEE} Computer Society, 2018.
\newblock \href {https://doi.org/10.1109/FOCS.2018.00052}
  {\path{doi:10.1109/FOCS.2018.00052}}.

\bibitem{MarxR14}
D{\'{a}}niel Marx and Igor Razgon.
\newblock Fixed-parameter tractability of multicut parameterized by the size of
  the cutset.
\newblock {\em {SIAM} J. Comput.}, 43(2):355--388, 2014.
\newblock \href {https://doi.org/10.1137/110855247}
  {\path{doi:10.1137/110855247}}.

\bibitem{Nederlof20a}
Jesper Nederlof.
\newblock Detecting and counting small patterns in planar graphs in
  subexponential parameterized time.
\newblock In Konstantin Makarychev, Yury Makarychev, Madhur Tulsiani, Gautam
  Kamath, and Julia Chuzhoy, editors, {\em Proccedings of the 52nd Annual {ACM}
  {SIGACT} Symposium on Theory of Computing, {STOC} 2020, Chicago, IL, USA,
  June 22-26, 2020}, pages 1293--1306. {ACM}, 2020.
\newblock \href {https://doi.org/10.1145/3357713.3384261}
  {\path{doi:10.1145/3357713.3384261}}.

\bibitem{Panolan0Z19}
Fahad Panolan, Saket Saurabh, and Meirav Zehavi.
\newblock Contraction decomposition in unit disk graphs and algorithmic
  applications in parameterized complexity.
\newblock In Timothy~M. Chan, editor, {\em Proceedings of the Thirtieth Annual
  {ACM-SIAM} Symposium on Discrete Algorithms, {SODA} 2019, San Diego,
  California, USA, January 6-9, 2019}, pages 1035--1054. {SIAM}, 2019.
\newblock \href {https://doi.org/10.1137/1.9781611975482.64}
  {\path{doi:10.1137/1.9781611975482.64}}.

\bibitem{PilipczukPSL18}
Marcin Pilipczuk, Michal Pilipczuk, Piotr Sankowski, and Erik~Jan van Leeuwen.
\newblock Network sparsification for steiner problems on planar and
  bounded-genus graphs.
\newblock {\em {ACM} Trans. Algorithms}, 14(4):53:1--53:73, 2018.
\newblock \href {https://doi.org/10.1145/3239560} {\path{doi:10.1145/3239560}}.

\bibitem{ReedSV04}
Bruce~A. Reed, Kaleigh Smith, and Adrian Vetta.
\newblock Finding odd cycle transversals.
\newblock {\em Oper. Res. Lett.}, 32(4):299--301, 2004.
\newblock \href {https://doi.org/10.1016/j.orl.2003.10.009}
  {\path{doi:10.1016/j.orl.2003.10.009}}.

\bibitem{robertson2003graph}
Neil Robertson and Paul~D Seymour.
\newblock Graph minors. xvi. excluding a non-planar graph.
\newblock {\em Journal of Combinatorial Theory, Series B}, 89(1):43--76, 2003.

\bibitem{schrijver2003combinatorial}
Alexander Schrijver.
\newblock {\em Combinatorial optimization: polyhedra and efficiency},
  volume~24.
\newblock Springer Science \& Business Media, 2003.

\bibitem{Tazari12}
Siamak Tazari.
\newblock Faster approximation schemes and parameterized algorithms on
  (odd-)h-minor-free graphs.
\newblock {\em Theor. Comput. Sci.}, 417:95--107, 2012.
\newblock \href {https://doi.org/10.1016/j.tcs.2011.09.014}
  {\path{doi:10.1016/j.tcs.2011.09.014}}.

\bibitem{Wahlstrom17}
Magnus Wahlstr{\"{o}}m.
\newblock Lp-branching algorithms based on biased graphs.
\newblock In Philip~N. Klein, editor, {\em Proceedings of the Twenty-Eighth
  Annual {ACM-SIAM} Symposium on Discrete Algorithms, {SODA} 2017, Barcelona,
  Spain, Hotel Porta Fira, January 16-19}, pages 1559--1570. {SIAM}, 2017.
\newblock \href {https://doi.org/10.1137/1.9781611974782.102}
  {\path{doi:10.1137/1.9781611974782.102}}.

\bibitem{Wahlstrom20}
Magnus Wahlstr{\"{o}}m.
\newblock On quasipolynomial multicut-mimicking networks and kernelization of
  multiway cut problems.
\newblock In Artur Czumaj, Anuj Dawar, and Emanuela Merelli, editors, {\em 47th
  International Colloquium on Automata, Languages, and Programming, {ICALP}
  2020, July 8-11, 2020, Saarbr{\"{u}}cken, Germany (Virtual Conference)},
  volume 168 of {\em LIPIcs}, pages 101:1--101:14. Schloss Dagstuhl -
  Leibniz-Zentrum f{\"{u}}r Informatik, 2020.
\newblock \href {https://doi.org/10.4230/LIPIcs.ICALP.2020.101}
  {\path{doi:10.4230/LIPIcs.ICALP.2020.101}}.

\bibitem{corr/abs-2002-08825}
Magnus Wahlstr{\"{o}}m.
\newblock On quasipolynomial multicut-mimicking networks and kernelization of
  multiway cut problems.
\newblock {\em CoRR}, abs/2002.08825, 2020.
\newblock URL: \url{https://arxiv.org/abs/2002.08825}, \href
  {http://arxiv.org/abs/2002.08825} {\path{arXiv:2002.08825}}.

\end{thebibliography}

\newpage

\appendix

\noindent
\textsc{\LARGE Appendix}

\section{Proof of Lemma~\ref{lem-rsdecomp}} \label{appx-proofrs}
The fact that any $H$-minor free graph admits a tree decomposition of adhesion size at most $c$ whose torsos are $c$-almost-embeddable (for some constant $c>0$ depending on $H$) follows from the profound work of Robertson and Seymour \cite{robertson2003graph}; we call such a tree decomposition \textit{Robertson-Seymour} decomposition.
Several algorithms have been developed to compute a Robertson-Seymour decomposition (and the almost-embeddable structures of the torsos) in polynomial time \cite{demaine2005algorithmic,grohe2013simple,kawarabayashi2011simpler}.
So it suffices to show how to make a Robertson-Seymour decomposition have the additional property.

In fact, one can always modify a given tree decomposition of a graph $G$ to make it satisfy the property that $G[\gamma(t) \backslash \sigma(t)]$ is connected and $\sigma(t) = N_G(\gamma(t) \backslash \sigma(t))$ for all node $t$ such that \textbf{(i)} the adhesion size of the tree decomposition does not increase after the modification and \textbf{(ii)} each torso in the modified tree decomposition is a subgraph of a torso in the original tree decomposition.
To see this, let $\mathcal{T}$ be a tree decomposition of a graph $G$ and $T$ be the underlying (rooted) tree of $\mathcal{T}$.
A \textit{redundant leaf} refers to a leaf $t \in T$ such that $\beta(t) \subseteq \beta(t')$ where $t'$ is the parent of $t$.
We define an operation $\mathsf{Prune}$, which modifies $\mathcal{T}$ by repeatedly removing redundant leaves until there is no redundant leaf (one can easily verify that the order of removing the redundant leaves does not matter).
Note that after the operation $\mathsf{Prune}$, $\mathcal{T}$ is still a tree decomposition of $G$.
Also note that after $\mathsf{Prune}$, the number of leaves in $T$ is at most $n = |V(G)|$.
Indeed, after $\mathsf{Prune}$, the bag of each leaf $t \in T$ contains a vertex of $G$ that is not contained in the bag of the parent of $t$ (and thus not contained in the bag of any other node), implying that the number of leaves in $T$ is at most the number of vertices in $G$.
Besides $\mathsf{Prune}$, we define two operations $\mathsf{Split}(t)$ and $\mathsf{Clean}(t)$ for a node $t \in T$.
Let $T_t$ denote the subtree of $T$ rooted at $t$.
\begin{itemize}
    \item $\mathsf{Split}(t)$.
    Let $C_1,\dots,C_q$ be the connected components of $G[\gamma(t) \backslash \sigma(t)]$.
    Define $V_i = \sigma(t) \cup V(C_i)$ for $i \in [q]$.
    We create $q$ copies of $T_t$, denoted by $T_t^{(1)},\dots,T_t^{(q)}$.
    For each node $s \in T_t$, let $s^{(i)}$ denote the copy of $s$ in $T_t^{(i)}$ and define $\beta(s^{(i)}) = \beta(s) \cap V_i$ as the bag of $s^{(i)}$, for $i \in [q]$.
    The $\mathsf{Split}(t)$ operation replaces $T_t$ with $T_t^{(1)},\dots,T_t^{(q)}$ in $T$, that is, it deletes $T_t$ from $T$ and then adds $T_t^{(1)},\dots,T_t^{(q)}$ as the subtrees of the parent of $t$.
    Note that after doing $\mathsf{Split}(t)$, $\mathcal{T}$ is still a tree decomposition of $G$.
    Indeed, if $(u,v) \in E(G)$ is an edge where $u,v \in \gamma(t)$, then $u,v \in V_i$ for some $i \in [q]$ and hence $u,v \in \beta(s^{(i)})$ for some $s \in T_t$ such that $u,v \in \beta(s)$.
    Furthermore, after doing $\mathsf{Split}(t)$, $G[\gamma(t^{(i)}) \backslash \sigma(t^{(i)})]$ is connected for all $i \in [q]$.
    \item $\mathsf{Clean}(t)$.
    We have $N_G(\gamma(t) \backslash \sigma(t)) \subseteq \sigma(t)$ because every vertex in $\gamma(t) \backslash \sigma(t)$ is only contained in the bags of the nodes in $T_t$ and hence can only be adjacent to the vertices in $\gamma(t)$.
    The $\mathsf{Clean}(t)$ operation removes the vertices in $\sigma(t) \backslash N_G(\gamma(t) \backslash \sigma(t))$ from the bags of all nodes in $T_t$.
    Note that after doing $\mathsf{Clean}(t)$, $\mathcal{T}$ is still a tree decomposition of $G$.
    Indeed, if $(u,v) \in E(G)$ is an edge where $u \in \sigma(t) \backslash N_G(\gamma(t) \backslash \sigma(t))$ and $v \in \gamma(t)$, then we must have $v \in \sigma(t)$ and hence $u,v \in \beta(t')$ where $t'$ is the parent of $t$.
    Furthermore, after doing $\mathsf{Clean}(t)$, $\sigma(t) = N_G(\gamma(t) \backslash \sigma(t))$.
\end{itemize}

With the three operations $\mathsf{Prune}$, $\mathsf{Split}$, and $\mathsf{Clean}$ in hand, we now modify $\mathcal{T}$ in two phases.
We say a node $t \in T$ is \textit{bad} if $G[\gamma(t) \backslash \sigma(t)]$ is not connected.
In the first phase, we repeat the following procedure: finding the highest bad node $t \in T$ and applying $\mathsf{Split}(t)$ followed by $\mathsf{Prune}$.
This phase terminates when all nodes in $T$ are good.
We claim that this phase can terminate in $O(dn)$ rounds where $d$ is the original depth of $T$ and $n = |V(G)|$.
To see this, we define a variant $(c_1,c_2)$ for $\mathcal{T}$, where $c_1$ is equal to $d$ minus the depth of the highest bad node $t \in T$ and $c_2$ is the number of bad nodes in the same level as $t$.
It is easy to see that after each round, $(c_1,c_2)$ decreases lexicographically, i.e., either $c_1$ decreases or $c_1$ remains unchanged and $c_2$ decreases.
Indeed, for any node $t \in T$, doing $\mathsf{Split}$ on a descendant of $t$ does not change $\gamma(t)$ and $\sigma(t)$.
Thus, $c_1$ can never increase.
Furthermore, when doing $\mathsf{Split}(t)$ on a bad node $t \in T$, $T_t$ is replaced with its copies $T_t^{(1)},\dots,T_t^{(q)}$ in which 
$t^{(1)},\dots,t^{(q)}$ are not bad nodes.
So the number of bad nodes in the same level as $t$ deceases by 1.
It follows that after each round, if $c_1$ remains unchanged, then $c_2$ must decrease by 1.
Therefore, $(c_1,c_2)$ decreases lexicographically after each round.
Note that $0 \leq c_1 \leq d$ because the $\mathsf{Split}$ and $\mathsf{Prune}$ operations do not increase the depth of $T$.
In addition, $0 \leq c_2 \leq n$ after each round, since the number of leaves in $T$ is at most $n = |V(G)|$ after $\mathsf{Prune}$ (as observed before).
As a result, the first phase has at most $O(dn)$ rounds and can be done in $n^{O(1)}$ time as long as the original size of $T$ is $n^{O(1)}$.
In the second phase, we repeat the following procedure: finding a node $t \in T$ such that $\sigma(t) \neq N_G(\gamma(t) \backslash \sigma(t))$ and applying $\mathsf{Clean}(t)$ followed by $\mathsf{Prune}$.
This phase terminates when $\sigma(t) = N_G(\gamma(t) \backslash \sigma(t))$ for all $t \in T$.
Clearly, the second phase has at most $O(dn^2)$ rounds and hence can be done in polynomial time, since the resulting tree after the first phase has size $O(dn)$ and each round of the second phase removes at least one vertex from one bag.
Note that a $\mathsf{Clean}$ operation does not change $\gamma(t) \backslash \sigma(t)$ for any node $t \in T$.
So after the second phase, $G[\gamma(t) \backslash \sigma(t)]$ is still connected for all $t \in T$.
Therefore, when the two-phase modification terminates, the resulting $\mathcal{T}$ has the desired property that $G[\gamma(t) \backslash \sigma(t)]$ is connected and $\sigma(t) = N_G(\gamma(t) \backslash \sigma(t))$ for all $t \in T$.
Furthermore, the adhesion size of $\mathcal{T}$ does not increase after the modification (as all of the three operations do not increase the adhesion size of $\mathcal{T}$) and each torso in the resulting $\mathcal{T}$ is a subgraph of a torso in the original $\mathcal{T}$.
To see the latter, observe that after $\mathsf{Clean}$ and $\mathsf{Prune}$ operations, the torso of each node $t \in T$ becomes a subgraph of the original torso of $t$, as these two operations only remove vertices from bags or remove nodes from $T$.
In a $\mathsf{Split}(t)$ operation, if the subtree $T_t$ is replaced by its copies $T_t^{(1)},\dots,T_t^{(q)}$, then the torso of $s^{(i)}$ is a subgraph of the torso of $s$ for any $s \in T_t$ and $i \in [q]$.

With the above argument, now we are ready to prove the lemma.
We first compute a Robertson-Seymour decomposition $\mathcal{T}_\text{RS}$ of $G$ (and the almost-embeddable structures of the torsos) using a polynomial-time algorithm \cite{demaine2005algorithmic,grohe2013simple,kawarabayashi2011simpler}, where the adhesion size of $\mathcal{T}_\text{RS}$ is at most $c$ and the torso of each node is $c$-almost embeddable.
Then we apply the above modification to $\mathcal{T}_\text{RS}$.
After this, $\mathcal{T}_\text{RS}$ satisfies the additional condition that $G[\gamma(t) \backslash \sigma(t)]$ is connected and $\sigma(t) = N_G(\gamma(t) \backslash \sigma(t))$ for all node $t$.
Furthermore, the adhesion size of $\mathcal{T}_\text{RS}$ is still bounded by $c$ and the torso of each node in the modified $\mathcal{T}_\text{RS}$ is a subgraph of a torso in the original $\mathcal{T}_\text{RS}$, where the latter is $c$-almost embeddable.
In what follows, we show that any subgraph of an $c$-almost-embeddable graph is $3c$-almost-embeddable, and thus setting $h = 3c$ completes the proof of the lemma.
(The fact that a subgraph of an almost-embeddable graph is also almost-embeddable should be a known result, while we cannot find it in the literature. So we include a proof below for completeness.)

Let $G$ be an $c$-almost-embeddable graph (see Definition~\ref{def-almost} for the formal definition of an almost-embeddable graph) with apex set $A \subseteq V(G)$, embeddable skeleton $G_0$ with an embedding $\eta$ to a genus-$c$ surface $\varSigma$, and vortices $G_1,\dots,G_c$ attached to disjoint facial disks $D_1,\dots,D_c$ in $(G_0,\eta)$ with witness pairs $(\tau_1,\mathcal{P}_1),\dots,(\tau_c,\mathcal{P}_c)$.
Let $G'$ be a subgraph of $G$.
Define $A' = A \cap V(G')$ and $G_i' = G_i \cap G'$ for $i \in \{0\} \cap [c]$.
So we have $G' - A' = G_0' \cup G_1' \cup \cdots \cup G_c'$.
If the vertices in the permutations $\tau_1,\dots,\tau_c$ are all vertices in $G'$, we are done, because in this case $G'$ is a $c$-almost-embeddable graph with apex set $A'$, embeddable skeleton $G_0'$ (embedded in $\varSigma$ by the embedding induced by $\eta$), vortices $G_1',\dots,G_c'$ attached to $D_1,\dots,D_c$ with witness pairs $(\tau_1,\mathcal{P}_1'),\dots,(\tau_c,\mathcal{P}_c')$, where $\mathcal{P}_i'$ is the path decomposition of $G_i'$ induced by $\mathcal{P}_i$, i.e., every bag of $\mathcal{P}_i'$ is the intersection of the corresponding bag of $\mathcal{P}_i$ with $V(G_i')$.
The difficult case is that some vertices in $\tau_1,\dots,\tau_c$ are not included in $G'$.
In this case, we need to add some vertices to $G_0'$ and modify $(\tau_1,\mathcal{P}_1),\dots,(\tau_c,\mathcal{P}_c)$ to obtain the witness pairs for the vortices $G_1',\dots,G_c'$.
Since the vortices $G_1',\dots,G_c'$ and the disks $D_1,\dots,D_c$ are disjoint, we can consider them individually.
Let $i \in [c]$ and suppose $\tau_i = (v_{i,1},\dots,v_{i,q_i})$.
Also, suppose $\pi_i = (u_{i,1},\dots,u_{i,q_i})$ is the underlying path of the path decomposition $\mathcal{P}_i$.
By definition, we have $v_{i,j} \in \beta(u_{i,j})$ for all $j \in [q_i]$.
For convenience, we write $\beta(u_{i,0}) = \emptyset$.
For $j \in [q_i]$, we say $u_{i,j}$ is an \textit{anchor} if $v_{i,j} \notin V(G')$ and there exists $a_j \in \beta(u_{i,j}) \cap V(G_i')$ such that $a_j \notin \beta(u_{i,j-1})$ and $a_j \notin \{v_{i,1},\dots,v_{i,q_i}\}$.
Note that for any two anchors $u_{i,j}$ and $u_{i,j'}$ where $j < j'$, we have $a_j \neq a_{j'}$, because $a_{j'} \in \beta(u_{i,j'})$ but $a_{j'} \notin \beta(u_{i,j'-1})$, which implies $a_{j'} \notin \beta(u_{i,j})$.
For each anchor $u_{i,j}$, we add the vertex $a_j$ to $G_0'$ and embed it at the position $\eta(v_{i,j}) \in \varSigma$; this is fine because $v_{i,j} \notin V(G')$.
For $j \in [q_i]$ such that $v_{i,j} \in V(G')$, we define $a_j = v_{i,j}$.
Let $J = \{j \in [q_i]: v_{i,j} \in V(G') \text{ or } u_{i,j} \text{ is an anchor}\}$ and suppose $J = \{j_1,\dots,j_{q_i'}\}$ where $j_1 < \cdots < j_{q_i'}$.
Then we have $V(G_0') \cap V(G_i') = \{a_{j_1},\dots,a_{j_{q_i'}}\}$, as we added the $a_j$'s of the anchors to $G_0'$.
Furthermore, $a_{j_1},\dots,a_{j_{q_i'}}$ are distinct and are embedded on the boundary of $D_i$ in clockwise or counterclockwise order.
Now we define $\tau_i' = (a_{j_1},\dots,a_{j_{q_i'}})$.
For $j^-,j^+ \in [q_i]$, we write $U(j^-,j^+) = \bigcup_{j=j^-}^{j^+} (\beta(u_{i,j}) \cap V(G'))$.
For convenience, set $j_0 = 0$ and $j_{q_i'+1} = q_i+1$.
Define a path $\pi_i' = (u_{i,j_1}',\dots,u_{i,j_{q_i'}}')$ where each node $u_{i,j_x}'$ is associated with a bag $\beta(u_{i,j_x}') = U(j_{x-1}+1,j_{x+1}-1)$.
Clearly, $a_{j_x} \in \beta(u_{i,j_x}')$.
Let $\mathcal{P}_i'$ be the path $\pi_i'$ with the associated bags.
One can easily verify that $\mathcal{P}_i'$ is a path decomposition of $G_i'$, because the bag of each node in $\pi_i'$ is the union of some consecutive bags in $\mathcal{P}_i$ intersecting $V(G')$ and these unions cover all bags in $\mathcal{P}_i$.
We show that the width of $\mathcal{P}_i'$ as at most $3c$.
Consider a node $u_{i,j_x}'$ in $\pi_i'$.
We claim that $\beta(u_{i,j_x}') = \beta(u_{i,j^-}) \cup \beta(u_{i,j_x}) \cup \beta(u_{i,j^+})$ where $j^- = j_{x-1}+1$ and $j^+ = j_{x+1} -1$.
We have $\beta(u_{i,j^-}) \cup \beta(u_{i,j_x}) \cup \beta(u_{i,j^+}) \subseteq \beta(u_{i,j_x}')$ since $\beta(u_{i,j_x}') = U(j^-,j^+)$.
To see $\beta(u_{i,j_x}') \subseteq \beta(u_{i,j^-}) \cup \beta(u_{i,j_x}) \cup \beta(u_{i,j^+})$, suppose there is a vertex $v \in V(G_i')$ that is in $\beta(u_{i,j_x}')$ but not in $\beta(u_{i,j^-}) \cup \beta(u_{i,j_x}) \cup \beta(u_{i,j^+})$.
As $v \in \beta(u_{i,j_x}') = U(j^-,j^+)$, we must have either $v \in U(j^-+1,j_x-1)$ or $v \in U(j_x+1,j^+-1)$.
Without loss of generality, assume $v \in U(j^-,j_x-1)$.
Observe that $v$ can be contained in $\beta(u_{i,j})$ only for $j \in \{j^-+1,\dots,j_x-1\}$, simply because $v \notin \beta(u_{i,j^-}) \cup \beta(u_{i,j_x})$ and the nodes in $\pi_i$ whose bags contain $v$ must form a connected subset.
As such, we have $v \notin \{v_{i,1},\dots,v_{i,q_i}\}$, because $j^-+1,\dots,j_x-1 \notin J$.
Also, we have $v \notin V(G_i') \backslash \{v_{i,1},\dots,v_{i,q_i}\}$, for otherwise $u_{i,j^*}$ is an anchor, where $j^* \in \{j^-+1,\dots,j_x-1\}$ is the smallest index such that $v \in \beta(u_{i,j^*})$.
Thus, $v \notin V(G_i')$ and we get a contradiction.
It follows that $\beta(u_{i,j_x}') = \beta(u_{i,j^-}) \cup \beta(u_{i,j_x}) \cup \beta(u_{i,j^+})$, which implies $\beta(u_{i,j_x}') \leq 3c$ as the width of $\mathcal{P}_i$ is $c$.
Now we use $(\pi_i',\mathcal{P}_i')$ as the witness pair of $G_i'$.
After we do this for all $i \in [c]$, we obtain the witness pairs of all vortices $G_1',\dots,G_c'$, which in turn gives us an almost-embeddable structure for $G'$ showing it is $3c$-almost embeddable.

\section{Proof of Lemma~\ref{lem-candidate}} \label{appx-proofcand}
The result of this lemma is implicit in literature. Thus, for each problem, we point out the references that give us the desired candidate set. Known kernels for all these problems actually first compute a candidate set and then apply reduction rules (in fact, an appropriate torso operation) to reduce the graph size.  

\medskip

\noindent 
{\bf EB and OCT:}   It follows from Theorems $6.9$, and $6.7$, which relies on Lemma $6.5$ of~\cite{KratschW20}. Also, see ~\cite[Lemma $3$]{LokshtanovSW12}. Here, replace the approximation algorithm for OCT on planar graphs with an approximation algorithm with factor $O(\sqrt{\log k})$ on general graphs~\cite[Theorem $6.8$]{KratschW20}.  

\smallskip

\noindent 
{\bf EMWC:} See~\cite[Corollary~$4$]{corr/abs-2002-08825}. Let $(G,T,k)$ be an instance of EMWC. We first reduce the capacity of $T$ to at most $2k$ (the number of edges incident to the terminal vertices).  
Let $T=\{t_1, \ldots, t_r\}$ be the terminal set, and let $k_i=\lambda(t_i, T-t_i)$ be the size of the ``isolating min-cut" between terminal $t_i$ and all other terminals. That is, an isolating cut for a terminal $t_i$ is a minimum set of edges which disconnects $t_i$ from each of the other terminals in $T$. It is known that there is a half-integral path-disjoint packing of paths of weight $\sum_{i=1}^r \frac{k_i}{2}$, hence if 
$\sum_i k_i > 2k$ we can reject the instance~(see the book of Schrijver,~\cite[Corollary $73.2d$]{schrijver2003combinatorial}).  By applying a reduction rule that contracts an edge into the terminal we can assume 
  that the edges incident to $t_i$ is a min-cut between $t_i$ and $T-t_i$.  This follows from the pushing lemma (\cite[Lemma $8.18$]{CyganFKLMPPS15}), that says that for any solution $X$, for any one terminal $t$, we can {\em push} $X$ so that it contains an ``important-cut" between $t$ and $T-t$. Further, we know that important-cut is found by a process of starting with a {\em furthest min-cut}, then doing a branching on  edges. That is, we select an edge $(x,y)$ leaving the furthest min-cut and recurse whether an important-cut contains $(x,y)$ or not (see~\cite[Theorem $8.11$]{CyganFKLMPPS15}). Hence, it  is guaranteed to exist an optimal solution $X'$ which is at least as far away from $t$ as the furthest $(t,T-t)$-min-cut. This implies the reduction rule mentioned above. More explicitly: 

\smallskip
\noindent 
{\sf Reduction Rule:} Let $(G, T,k)$ be an instance of the EMWC problem. Let $(t_1,v)$ be an edge adjacent to a terminal $t_1$,  and assume that there is a minimum isolating cut of $t_1$ that does not contain the edge $(t_1,v)$. Then the instances $(G',T,k)$ and $(G,T,k)$ are equivalent. Here, $G'$ is the graph that has been obtained by contracting the edge $(t_1,v)$, and calling the newly vertex as $t_1$. 

The above description implies that we can obtain an equivalent instance which is a minor of the original graph and   capacity of $T$ to at most $2k$. After this we follow the algorithm for EMC given in \cite[Corollary~$4$]{corr/abs-2002-08825} and design the candidate set $Z$. 

\smallskip
\noindent 
{\bf EMC:} Construction of set $Z$ in \cite[Corollary~$4$]{corr/abs-2002-08825}. 

\smallskip
\noindent 
{\bf VMWC:} Construction of set $Z$ in \cite[Lemma~$7.3$]{KratschW20}. 

\smallskip
\noindent 
{\bf VMC:} Construction of set $Z\cup T$ in \cite[Lemma~$7.4$]{KratschW20}. 

\smallskip
\noindent 
{\bf GFVS:} Construction of set $Z\cup T$ in \cite[Lemma~$7.10$]{KratschW20}. 

\smallskip
\noindent
{\bf GFES:} The proof is given  in \cite[Corollary~$4$]{corr/abs-2002-08825}. Important steps of this are as follow.  Let $(G,\phi, k)$ be an instance of GFES. Here, $\phi$ is a direction-dependent labelling of the edges of $G$ from some multiplicative group.  The algorithm first computes an approximate solution  $X_0$ of size $O(k \log k)$. Then follows the proof of  \cite[Lemma~$7.10$]{KratschW20}, and applies  \cite[Theorem~$2$ and Corollary~$3$]{corr/abs-2002-08825} to get the desired set $Z$, as our candidate set lemma. 
This concludes the lemma. 

\end{document}